\documentclass[nofootinbib,superscriptaddress,eqsecnum,tightenlines,11pt]{revtex4}

\usepackage{hyperref}
\usepackage{graphicx}
\usepackage{amsmath,amssymb,amsfonts,amsthm,stmaryrd,mathtools}
\usepackage{multirow,bigdelim}
\usepackage{dsfont}

\makeatletter
\newcounter{subsubsubsection}[subsubsection]
\renewcommand\thesubsubsubsection{\thesubsubsection .\@alph\c@subsubsubsection}
\newcommand\subsubsubsection{\@startsection{subsubsubsection}{4}{\z@}%
                                     {-3.25ex\@plus -1ex \@minus -.2ex}%
                                     {1.5ex \@plus .2ex}%
                                     {\centering\normalfont\small\textit}}
\newcommand*\l@subsubsubsection{\@dottedtocline{3}{10.0em}{4.1em}}
\newcommand*{\subsubsubsectionmark}[1]{}
\makeatother

\newtheorem{Theorem}{Theorem}[section]
\newtheorem{Definition}{Definition}[section]
\newtheorem{Lemma}{Lemma}[section]

\def\be{\begin{equation}}
\def\ee{\end{equation}}
\def\ba{\begin{eqnarray}}
\def\ea{\end{eqnarray}}
\def\bas{\begin{subequations}\begin{eqnarray}}
\def\eas{\end{eqnarray}\end{subequations}}

\def\eps{\varepsilon}

\def\tr{\text{tr}}

\def\de{\mathrm{d}}
\def\Pexp{\overrightarrow{\exp}}
\def\f{\frac}
\def\lb{\big\lbrace}
\def\rb{\big\rbrace}
\def\SU{\text{SU}}
\def\SO{\text{SO}}

\def\SL{\text{SL}}
\def\su{\mathfrak{su}}

\def\vp{\vphantom{-1}}
\def\openone{\mathds{1}}
\def\time{\,{\scriptstyle{\times}}\,}

\def\nn{\nonumber}
\def\q{\qquad}

\def\i{\mathrm{i}}
\def\g{\mathrm{g}}
\def\bX{\mathbf{X}}

\begin{document}

\title{Flux formulation of loop quantum gravity:\\ Classical framework}

\author{Bianca Dittrich}
\affiliation{Perimeter Institute for Theoretical Physics,\\ 31 Caroline Street North, Waterloo, Ontario, Canada N2L 2Y5}
\author{Marc Geiller}
\affiliation{Institute for Gravitation and the Cosmos \& Physics Department,\\ Penn State, University Park, PA 16802, U.S.A.}

\begin{abstract}
We recently introduced a new representation for loop quantum gravity, which is based on the BF vacuum and is in this sense much nearer to the spirit of spin foam dynamics. In the present paper we lay out the classical framework underlying this new formulation. The central objects in our construction are the so-called integrated fluxes, which are defined as the integral of the electric field variable over surfaces of codimension one, and related in turn to Wilson surface operators. These integrated flux observables will play an important role in the coarse graining of states in loop quantum gravity, and can be used to encode in this context the notion of curvature-induced torsion. We furthermore define a continuum phase space as the modified projective limit of a family of discrete phase spaces based on triangulations. This continuum phase space yields a continuum (holonomy-flux) algebra of observables. We show that the corresponding Poisson algebra is closed by computing the Poisson brackets between the integrated fluxes, which have the novel property of being allowed to intersect each other.
\end{abstract}

\maketitle

\section{Introduction}

\noindent Quantum gravity aims at providing a description of the dynamics of quantum geometry and of its interaction with quantum matter fields. As such, it requires in particular an understanding of the notion of quantum geometry. One of the central achievements of loop quantum gravity (LQG hereafter) is precisely to provide a rigorous and background-independent construction of a (Hilbert space) representation supporting geometrical degrees of freedom \cite{lqg1,lqg2,lqg3}. These degrees of freedom are encoded as intrinsic and extrinsic geometrical data in the so-called kinematical holonomy-flux algebra of observables. The representation of this latter has been worked out in the 90's by Ashtekar, Lewandowski and Isham \cite{ali1,ali2,ali3,ali4}, and is known as the Ashtekar--Lewandowski (AL hereafter) representation. It can be understood as a quantization of the (extended) space of connection fields, where the connections themselves appear through their holonomies.

The goal of the present work is to construct an alternative representation for quantum geometry. As we will see, instead of being based on the holonomies, this new representation gives a much more prominent role to the fluxes (which encode the intrinsic geometry). This therefore opens the road to quantizing the (extended) space of flux configurations. We believe and will argue that this new representation based on the fluxes could facilitate very much the imposition of a suitable dynamics, and thus constitutes a better starting point for the construction of a physical Hilbert space supporting solutions to all the constraints of the theory.

One central achievement of the AL representation is to deal successfully and for the first time with (spatial) diffeomorphisms \cite{abhaydiffeos}. This follows essentially from the topological character of the Poisson bracket relations for the holonomy-flux algebra, and from the fact that excitations have a very distributional nature. Here we will not alter these two properties, and therefore leave open the possibility of solving the spatial diffeomorphism constraint in a similar way to what is done in the AL representation.

Another notable (yet peculiar) property of the AL representation is that it is built over a vacuum state which is totally squeezed, and in which all expectation values of fluxes are vanishing. Excitations on top of this vacuum are generated by holonomy variables, and are therefore based on 1-dimensional curves along which the fluxes are then non-vanishing\footnote{The fluxes are based on $(d-1)$-dimensional surfaces, and if such a surface cuts the path of an holonomy it leads to a non-vanishing commutator.}. This picture has the drawback of making quite difficult the construction of states describing smooth geometries, which is in turn related (but not strictly equivalent) to the issue of constructing semi-classical states. In addition to this, the fact that LQG is based on such distributional geometries makes the contact with spin foam models slightly unclear, since these attempts to formulate the covariant dynamics are rather based on piecewise-flat manifolds.

These difficulties have motivated the attemps to construct and the search for alternative representations and setups for the quantization of theories of connections \cite{otherBaez}. Somehow paradoxically, investigations \cite{hanno1,hanno2} of the conditions under which the AL representation could be generalized, lead to the so-called F--LOST uniqueness theorem for the kinematical structure of LQG \cite{FLOST1,FLOST2}. This theorem states that the AL representation is the only one satisfying a certain number of assumptions, including irreducibility, the requirement that spatial diffeomorphisms act as automorphisms and leave the vacuum invariant, and the requirement that fluxes exist either as operators or as a weakly continuous operator family of exponentiated fluxes.

In light of this, it is clear that alternative representations of the holonomy-flux algebra should necessarily violate at least one of the assumptions of the uniqueness theorem. An example is given by the so-called Koslowski--Sahlmann representation \cite{KS1,KS2,KS3}, which is based on a vacuum defined over a non-vanishing background flux, and is therefore not invariant (but rather covariant) under spatial diffeomorphisms \cite{varadarajan1,varadarajan2,varadarajan3}. This background structure does however present the advantage of allowing for the description of condensate states corresponding to ``macroscopic manifolds''. Very recently, a proposal was made \cite{lanery1,lanery2,lanery3,lanery4}, based on earlier work \cite{kijowski,okolow1,okolow2}, to generalize the inductive limit Hilbert space construction underlying the known representations. This series of papers aims in particular at representing semi-classical states (in the sense of states that are not squeezed). This proposal has to be explored further in order to understand the way in which it deals with irreducibility, spatial diffeomorphisms, and the $\SU(2)$ gauge symmetry.

In the present paper, we develop the classical framework underlying the representation introduced in \cite{paper1}. As expected, this new representation violates one of the assumptions of the F--LOST uniqueness theorem. This is due to the fact that the integrated fluxes in the quantum theory will only exist in their exponentiated form (and not form weakly continuous operator families). In many aspects, this new (BF or flux) representation can be understood as being dual to the AL representation\footnote{The BF representation discussed here should not be confused with the ``non-commutative flux representation'' introduced in \cite{aristideandco}. This latter rather results from a unitary transformation (in fact, a non-commutative Fourier transform) on the holonomy representation of the usual AL framework. Thus, the word ``representation'' in \cite{aristideandco} is referring to a choice of (generalized) basis in the AL Hilbert space. This is in sharp contrast with the BF representation (of the flux-algebra), which will turn out to be unitarily inequivalent to the AL representation.}. Although the fact that such an alternative representation should exist has been suggested early on \cite{pullin,lewandowski,bianchi,FGZ,cylconsis,guedes}, the need to dualize every ingredient of the usual AL representation has been realized only recently in \cite{paper1}.

The underlying vacuum for this new representation is again a totally squeezed state, which is however peaked on vanishing curvature. Therefore, excitation on top of this vacuum will turn out to describe distributional curvature. The attractiveness of this proposal is that the new vacuum and its excitations coincide with the building blocks of spin foam (or even Regge) quantum gravity. In particular, the vacuum is a physical state of BF theory, and this latter underlies the whole construction of spin foam models. This representation therefore seems to be especially suited for the construction of the dynamics (i.e. the imposition of the constraints) of the theory.

Another strength of the new representation is the geometric interpretation, which is made more transparent and straightforward in the context of ``almost everywhere flat'' than in the context of ``almost everywhere degenerate'' configurations. For this reason, we will make our construction as geometrical as possible, and choose in particular to work with the so-called simplicial fluxes \cite{husain,qsd7}, which transform in a gauge covariant way. In fact, we are going to see that the new representation is based on the ability to compose or ``coarse grain'' fluxes, which can only be done consistently when considering these simplicial fluxes.

The description of how fluxes compose (as integrated fluxes) is going to be central for the construction of the inductive limit Hilbert space, and we are therefore going to study this point in details. One of the main results of the paper will be the definition of these integrated fluxes, and a discussion of their geometric interpretation. We will furthermore prove for the first time that the Poisson algebra of holonomies and fluxes is closed. In particular, contrary to what one may expect, the commutator between fluxes whose underlying surfaces are intersecting does not lead to (even more) distributional objects.

In this work we are also going to describe a continuum realization of the underlying phase space supporting the new BF representation. This continuum formulation will also provide a description of the holonomy-flux observable algebra that will be represented in the quantum theory (since we use simplicial fluxes, this algebra differs slightly from the one used for the AL representation). To this end, we will have to modify the standard way of constructing continuum phase spaces out of discrete phase spaces, which is via a projective limit. More precisely, we will introduce a modified projective limit taking into account the (flatness) constraints. This modification will restrict the applicability of the projections so that these do not erase any (curvature) information. Such constraints are also essential in the description of a dynamics involving varying phase space dimensions \cite{hoehn1,hoehn2}, which appear in the context of simplicial discretizations. Furthermore, the constraints can be used to describe the discrete phase spaces as reduced phase spaces \cite{FGZ}. We will prove here that the coarser phase spaces result from symplectic reductions of finer phase spaces, and that the (restricted) projections provide the symplectic maps. We will point out that all these techniques can also be used to define a classical framework underlying the AL representation. In fact the appearance of (first class) constraints in both cases explains the squeezed nature of both the AL and the BF vacuum states.

\subsection*{Outline}

\noindent This paper is organized as follows. In section \ref{sec:setup}, we describe the setup of our construction, and recall some general definitions related to triangulations of manifolds and their dual complexes. As required for the modified projective limit, we introduce the partially ordered set of triangulations, and specify the refinement operations which define the partial order.

Section \ref{sec:classical} briefly reviews the classical phase spaces based on a fixed triangulation, and, most importantly, also defines the flux observables and how these flux observables are composed to form ``integrated flux observables''.

In section \ref{sec:geom}, we discuss the geometric interpretation of the integrated flux observables and the way in which they depend on the underlying curve (in $d=2$ spatial dimensions) or surface (in $d=3$ spatial dimensions). We will point out that the coarse-grained fluxes may violate the (coarse) Gauss constraints, leading to an effect which we coin ``curvature-induced torsion''.

Section \ref{connecting} is devoted to the definition and the characterization of the continuum phase space starting from the family of discrete phase spaces. In order to do so, we will introduce (restricted) projection and (generalized) embedding maps connecting coarser and finer phase spaces. This will allow us to define a (modified) projective limit, and to prove that coarser phase spaces arise as symplectic reductions of finer phase spaces. We will also discuss the definition of the continuum observables as (consistent) families of observables defined on discrete phase spaces.

All this material will then enable us to consider, in section \ref{fluxalgebra}, the commutator algebra of fluxes. We will show that the Poisson algebra of holonomies and fluxes is closed, and discuss various cases for the commutator of the fluxes.

In section \ref{sec:diffeos} we will sketch the way in which spatial diffeomorphisms are going to act in the quantum theory. As in the AL representation, this action changes the embedding of the excitations, which for the BF representation are curvature excitations. We will provide an heuristic argument (for $d=2$ spatial dimensions) showing that the action which we propose is indeed related to an imposition of the spatial diffeomorphism constraints.

We will finally discuss various possibilities for the imposition of the dynamics in section \ref{sec:dyn}, and close in section \ref{sec:summary} with a summary of the new results and a discussion of open issues and possible generalizations.

The appendices contain a short discussion of the notion of inductive and projective limits, some basic formulas involving the group $\SU(2)$ and its Lie algebra that needed for the computation of various Poisson brackets, and detailed computations of the Poisson brackets of the fluxes on the discrete and continuum phase spaces.

\section{Setup}
\label{sec:setup}

\noindent Since the setup for the present work is a simplicial formulation of LQG, it is helpful to recall some basic definitions and notations that will be used throughout the main text. In particular, for the definition of the continuum phase space via a (modified) projective limit, we need to specify a partially ordered set of (discrete) structures labelling the phase spaces. This set will be given by the set of triangulations (for a review of some basic concepts, see for instance \cite{carfora-marzuoli}). In addition to this, we also need to specify the refinement operations that will define the partial order.

In this work, we will denote by $\Sigma$ a closed spacelike $d$-dimensional manifold (with $d=2$ or $d=3$ depending on the context). The group $G$ will be a semi-simple compact Lie group (typically $\SU(2)$) with Lie algebra $\mathfrak{g}$, but most of the results will carry over to the case of finite groups as well.

\subsection{Triangulations and dual complexes}
 
\begin{Definition}[Simplices]
A $k$-simplex $\sigma^k\equiv(v_0,\ldots,v_k)$ with vertices $v_0,\ldots,v_k$ is the subspace of $\mathbb{R}^d$ (with $k\leq d$) defined by
\be
\sigma^k=\left\{\sum_{i=0}^k\lambda_iv_i\Bigg|\sum_{i=0}^k\lambda_i=1,\ 0\leq\lambda_i\leq1,\ v_i\in\mathbb{R}^d\right\}.
\ee
The dimension of a k-simplex is $\dim\sigma^k=k$. For $j<k$, we will call a subsimplex of $\sigma^k$ any simplex $\sigma^j\subset\sigma^k$ whose vertices are a subset of those of $\sigma^k$.
\end{Definition}

We introduce the common notation according to which 0-simplices $\sigma^0$ are called vertices and denoted by $v$, $1$-simplices $\sigma^1$ are called edges and denoted by $e$, $2$-simplices $\sigma^2$ are called triangles and denoted by $t$, and $3$-simplices are called tetrahedra and denoted by $\tau$.

\begin{Definition}[Simplicial complex]
A finite simplicial complex $K$ is a finite collection of simplices such that:\\
i) If $\sigma^k$ is in $K$, then all of its subsimplices are in $K$ as well;\\
ii) If $\sigma^j,\sigma^k\in K$, then $\sigma^j\cap\sigma^k$ is either a subsimplex of both $\sigma^j$ and $\sigma^k$ or is empty.\\
The underlying space, or underlying polyhedron, denoted by $|K|$, is the set-theoretic union of the simplices forming $K$ equipped with the topology inherited from $\mathbb{R}^d$.
\end{Definition}

\begin{Definition}[$k$-skeleton]
The $k$-skeleton of a simplicial complex $K$ is the subcomplex that contains all the simplices of dimension lower than and equal to $k$.
\end{Definition}

\begin{Definition}[Star]
The star $\mathrm{st}(\sigma^k)$ of a simplex $\sigma^k$ is the union of all the simplices that have $\sigma^k$ as a subsimplex. The closure $\overline{\mathrm{st}}(\sigma^k)$ of the star of a simplex is the smallest subcomplex that contains $\mathrm{st}(\sigma^k)$.
\end{Definition}

\begin{Definition}[Triangulation]
A triangulation $\Delta$ of the manifold $\Sigma$ (this latter being viewed as a topological space), is a simplicial complex $K$ together with a homeomorphism from $|K|$ to $\Sigma$.
\end{Definition}

In a slight abuse of language, we will therefore refer to the simplices of the triangulation, and denote the $k$-skeleton by $\Delta_k$. We will denote the number of vertices of a triangulation by $|v|$, and use a similar notation for the other elements of the triangulation as well.

\begin{Definition}[Geometric triangulation]
Let $q$ be a fiducial $d$-dimensional Riemannian metric tensor on $\Sigma$. This fiducial metric can be used to equip the triangulation with a geometric structure. This can be done by embedding the vertices $v$ (i.e. by giving them coordinates), and considering that the edges and triangles are respectively geodesics and minimal surfaces with respect to the fiducial metric. This defines the notion of a geometric triangulation.
\end{Definition}

We assume that the triangulation is sufficiently fine so that the geodesics and minimal surfaces between its vertices are unique. We will also restrict ourselves to orientable manifolds, and assume that the top-dimensional simplices $\sigma^d$ carry a positive orientation. This excludes the formation of ``spikes'', that is for instance a subdivided tetrahedron for which the inner vertex lies outside the tetrahedron. Note that these requirements are formulated only with respect to the auxiliary metric. Any physical metric could still lead to such spikes. This will also restrict the action of spatial diffeomorphisms, which will be discussed later on.

In $d=2$ spatial dimensions, the fact that we choose to embed only the vertices will be justified later on as well. In fact, it will turn out that the (integrated) flux observables are labeled by curves, and only depend on the relative position of these curves with respect to the vertices. This will be different in $d=3$ spatial dimensions, where the position of the edges of the triangulation will also be relevant (for instance in the sense that they carry distributional curvature). For $d=3$, we could therefore also allow for an embedding of the edges (different from a geodesic embedding), but still define the 2-simplices as minimal surfaces.

\begin{Definition}[Dual complex]
For each triangulation, we consider the dual complex $\Upsilon$, which is a cellular complex (i.e. its cells are not necessarily simplices) consisting of $0$-cells called nodes $n$, $1$-cells called links $l$, $2$-cells called faces $f$, and $3$-cells. The duality is defined by a one-to-one correspondence between the $k$-dimensional cells of $\Upsilon$ and the $(d-k)$-dimensional simplices of $\Delta$. Furthermore, we assign an orientation to the links and the $(d-1)$-dimensional simplices $\sigma^{d-1}$ such that the orientation of each pair $(l,\sigma^{d-1})$ is direct (i.e. $l\wedge\sigma^{d-1}$ leads to a positive volume form with respect to the auxiliary metric). Finally the links $l\in\Upsilon$ meet the corresponding dual simplices $\sigma^{d-1}$ in the triangulation (i.e. the edges in $d=2$ and the triangles in $d=3$) transversally in the sense that:\\
i) The intersection $u=l\cap\sigma^{d-1}$ is given by a single point;\\
ii) There exists an open neighborhood $U$ of $u$ together with a diffeomorphism mapping $U$ to $\mathbb{R}^d$ (whose coordinates we denote by $x^1,\ldots,x^d$), $u$ to the origin of $\mathbb{R}^d$, $\sigma^{d-1}\cap U$ to the plane $x^d=0$, and $l\cap U$ to the line $x^1=\ldots=x^{d-1}=0$.
\end{Definition}

The nomenclature which we use to denote the elements of a $d$-dimensional triangulation $\Delta$ and of its dual complex $\Upsilon$ are summarized in table \ref{fig:1}, together with various notations which we will introduce and use later on. For simplicity, the graph corresponding to the 1-skeleton $\Upsilon_1$ of the dual complex $\Upsilon$ will be denoted by $\Gamma$. When writing dimension-independent expressions, we will sometimes denote for example the $(d-1)$-dimensional simplex $\sigma^{d-1}$ dual to the link $l$ simply by $l^*$.
\begin{center}
\begin{table}[h]
\begin{tabular}{|c|c||c|c|}
\hline
\multicolumn{2}{|c||}{$d=2$}&\multicolumn{2}{c|}{$d=3$}\\
\hline
~~~~~~~~~~~~$\Delta$~~~~~~~~~~~~&~~$\Upsilon$~~&~~~~~~~~~~~~~~$\Delta$~~~~~~~~~~~~~~&~~$\Upsilon$~~
\\\hline
vertex $v$ & face $f$ & vertex $v$ & 3-dimensional cell
\\
edge $e=l^*$ & ~~link $l$~ \rdelim\}{2}{0mm}[ 1-skeleton $\Gamma$]\q\q\q\q\q & edge $e$ & face $f$
\\
triangle $t$ & ~node $n$ \q\q\q\q\q & triangle $t=l^*$ & ~~link $l$~ \rdelim\}{2}{0mm}[ 1-skeleton $\Gamma$]\q\q\q\q\q
\\
& & tetrahedron $\tau$ & ~node $n$ \q\q\q\q\q
\\
edge path $\pi\subset\Delta_1$ & shadow graph $\Gamma_\pi\subset\Gamma$ & surface path $\pi\subset\Delta_2$ & shadow graph $\Gamma_\pi\subset\Gamma$
\\
& & & shadow tree $\mathcal{T}_\pi\subset\Gamma_\pi$
\\
& path $\gamma\subset\Gamma$ & & path $\gamma\subset\Gamma$
\\
& leaf $\ell$ & & leaf $\ell$
\\
& root $r$ & & root $r$
\\
& branch $b$ & & branch $b$
\\
& tree $\mathcal{T}$ & & tree $\mathcal{T}$
\\
\hline
\end{tabular}
\caption{Elements of a $d$-dimensional triangulation $\Delta$ and of its dual complex $\Upsilon$.}
\label{fig:1}
\end{table}
\end{center}

It will be convenient later on, in order to introduce the new flux variables, to fix a reference node $n$ in the dual complex $\Upsilon$ and call it the root $r$. This root is dual to a $d$-simplex $\sigma^d_r$ of the triangulation, and specifies a reference frame in which closed holonomies can be based and the fluxes can be transported. A path in the graph $\Gamma$ connecting the root $r$ to any node $n$ (or more generally any two nodes) can be specified uniquely by a choice of spanning tree, which we now define.

\begin{Definition}[Spanning tree]
A spanning tree $\mathcal{T}$ is a connected subgraph of the $1$-skeleton $\Gamma$ of the dual complex $\Upsilon$, which does not contain any cycles but includes all the nodes of the dual complex. We call the links of the tree branches and denote them by $b$, while the links of $\Gamma$ that are not in the tree are called leaves and denoted by $\ell$. These leaves are in one-to-one correspondence with the fundamental cycles $c$ of the graph $\Gamma$. A rooted spanning tree with root $r$ is a spanning tree where a preferred node $n$ is identified and called the root.
\end{Definition}

A spanning tree will be used later on to perform a partial gauge fixing by setting the group elements on the branches to the identity, thereby leaving degrees of freedom only on the leaves of the tree. If we denote by $|n|$ the number of nodes of the graph $\Gamma$ and $|l|$ the number of links, then a spanning tree has $|n|-1$ branches and $|\ell|=|l|-|n|+1$ leaves (and fundamental cycles).

The spanning tree serves a dual purpose, since it can be used to perform a partial gauge fixing of the group elements, but also to specify uniquely a path between any two nodes.

\subsection{Alexander moves}
\label{subsec:pachner}

\noindent As opposed to the AL representation, for which the refinement is based on the graph, the refinement for the BF representation will be based on the triangulation itself. This indeed seems to be the only possible choice since the excitations with respect to the BF vacuum are themselves based on the triangulation.

Triangulations can be refined by the so-called Alexander moves, which are also known as star subdivisions. Alexander proved that any two triangulations of a polyhedron can be transformed into each other by a finite sequence of star subdivisions and their inverses \cite{alexander}. The Alexander moves are defined as follows.

\begin{Definition}[Refining Alexander moves]\label{def:alexander}
The refining Alexander moves are obtained by placing a vertex in the interior of a $k$-dimensional simplex $\sigma^k$, for $k\in\{1,\ldots,d\}$, and connecting this vertex via new edges to the vertices of the closure $\overline{\mathrm{st}}(\sigma^k)$ of the star of $\sigma^k$. For $d=2$, these are the \hbox{$1$-$3$} and $2$-$4$ moves, which arise respectively from subdividing a triangle and an edge shared by two triangles. For $d=3$, the refining moves are the $1$-$4$, $2$-$6$, and $n$-$(2n)$ moves, which arise respectively from subdividing a tetrahedron, a triangle shared by two tetrahedra, and an edge shared by $n\geq3$ tetrahedra.
\end{Definition}

An example of Alexander 2-4 move in $d=2$ spatial dimensions is illustrated on figure \ref{fig:A2-4}, along with the behavior of the root which will be described below.

Every refining move involves the placement of a new vertex, for which new (embedding) variables need to be specified. The various moves differ in whether the vertices are placed into the bulk of a top-dimensional simplex $\sigma^d$ or onto its $(d-1)$ or $(d-2)$-dimensional boundaries, in which case the embedding variables are constrained to be on the various submanifolds.

Notice that the Alexander 1-3 and 1-4 moves are in fact equivalent to the Pachner 1-3 and 1-4 moves \cite{pachner}. There are however also Pachner moves which are neither refining nor coarse graining moves, which is the reason for which we choose to work with the Alexander moves instead\footnote{In $d=2$ spatial dimensions it is possible to work with the Pachner moves instead of the Alexander moves \cite{paper1}. However, this prevents the introduction of a new vertex exactly on an existing edge.}.

We will eventually define phase spaces associated to the triangulations and parametrized by the holonomy and flux variables. We will work with almost gauge-invariant phase spaces (the only gauge transformations left will be those acting at the root). We therefore have to actually work with the category of rooted triangulations and to specify the behavior of the root under refinements.

The behavior of the root under refinements can be described in two ways. One possibility is to specify a point in the manifold $\Sigma$ as the root. This root point singles out the enclosing $d$-dimensional simplex as the dual to the root node. For this, one has however to exclude refinements that result in a placement of this root point on some lower-dimensional simplex. Note that triangulations with different roots cannot be refined into each other. However, the phase spaces associated to (otherwise equivalent) triangulations with different roots will be connected by a global gauge transformation.

The second possibility to describe the behavior of the root under refinements is to introduce the notion of flagged structure for the root node.

\begin{Definition}[Flagged structure]\label{def:flag}
The flag $\mathrm{fl}(\sigma^k)$ of a $k$-simplex $\sigma^k$ is a set of subsimplices $\sigma^0\subset\sigma^1\subset\cdots\subset\sigma^k$ such that, if $j<k$, the elements of the flag $\mathrm{fl}(\sigma^j)$ of a simplex $\sigma^j\in\mathrm{fl}(\sigma^k)$ are again elements of $\mathrm{fl}(\sigma^k)$, i.e.
\be
\mathrm{fl}(\sigma^j)\subset\mathrm{fl}(\sigma^k),\ \forall\,\sigma^j\in\mathrm{fl}(\sigma^k).
\ee
\end{Definition}

Let $\sigma^d_r$ be the $d$-dimensional simplex dual to the root node, and $\mathrm{fl}(\sigma^d_r)$ a choice of flag for this simplex. After a refining Alexander move affecting the simplex $\sigma^d_r$, we simply define the new root node as the dual of the $d$-dimensional simplex of the refined triangulation that contains the subsimplices $\sigma^0\subset\cdots\subset\sigma^{d-1}$ of the initial flag $\mathrm{fl}(\sigma^d_r)$. In other words, in a refining move affecting the flagged $d$-simplex dual to the root, there is always a unique $d$-simplex in the refined triangulation that inherits the initial flag, and which therefore defines canonically a new root. An example of the behavior of the root under a refining Alexander move is represented in figure \ref{fig:A2-4} for the 2-4 move.

\begin{center}
\begin{figure}[h]
\includegraphics[scale=0.7]{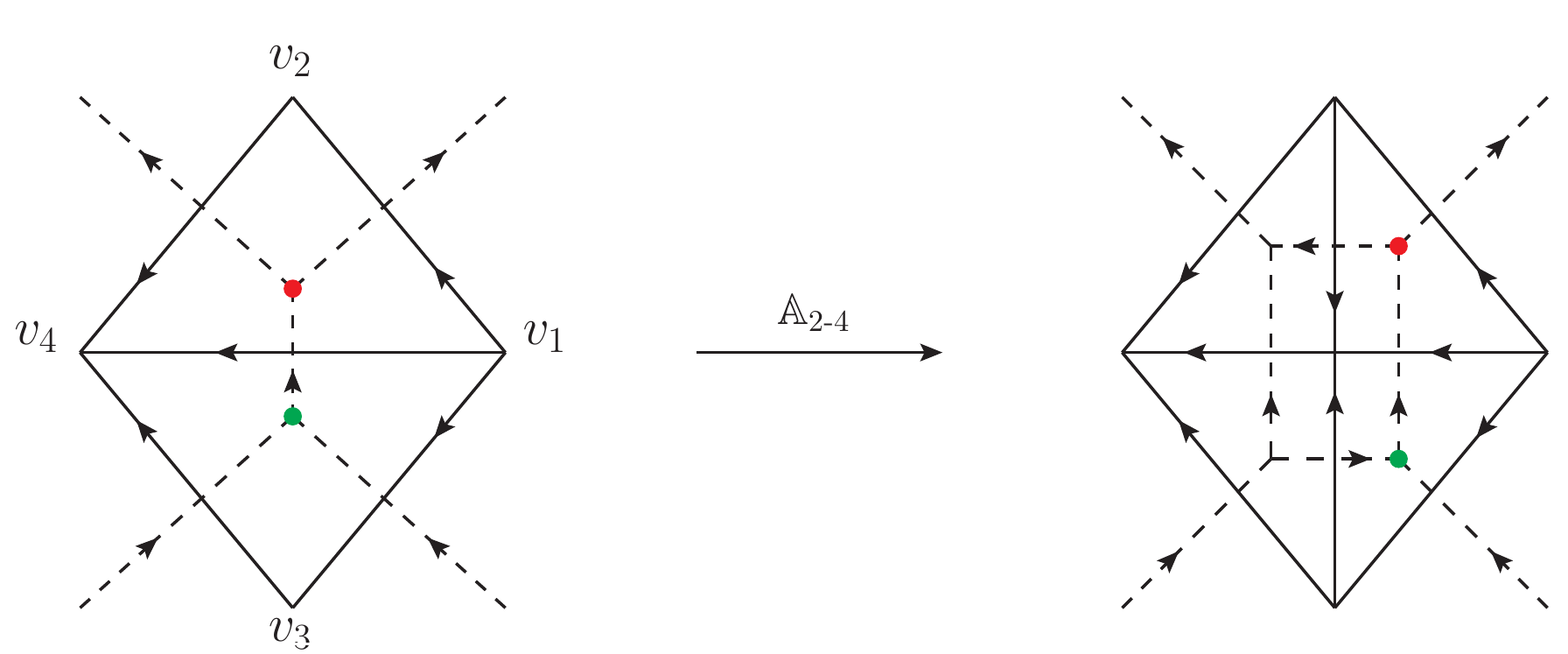}
\caption{Alexander 2-4 move obtained by placing a vertex in the edge connecting the vertices $v_1$ and $v_4$, and then connecting this new vertex to $v_2$ and $v_3$. We have illustrated the behavior of two possible choices for the root (red and green nodes). If we choose the flag of the upper triangle to consist of the vertex $v_1$ and the edge connecting $v_1$ and $v_2$, one can see that the triangle which inherits these flag simplices in the refined triangulation is chosen to be the new root. For the bottom triangle the flag is given by $v_1$ and the edge connecting $v_1$ and $v_3$.}
\label{fig:A2-4}
\end{figure}
\end{center}

Finally, and most importantly for the construction of the new representation, the refining Alexander moves can be used to equip the set of triangulations with a partial order as follows.

\begin{Definition}[Partial order]
A (rooted) triangulation $\Delta'$ is said to be finer than a (rooted) triangulation $\Delta$, which we denote by $\Delta\prec\Delta'$, if $\Delta'$ can be obtained from $\Delta$ by a finite series of refining Alexander moves.
\end{Definition}

This partial order can also be used for geometric triangulations, in which case we have to require in addition that the vertices of $\Delta$ after the refining Alexander moves have the same coordinates as the vertices of $\Delta'$. Note that we will consider two triangulations as being equivalent (or as refinements of each other) even if the orientation of their simplices disagrees.

\begin{Definition}[Common refinenement]
A common refinement of two triangulations $\Delta$ and $\Delta'$ is a triangulation $\mathrm{cr}(\Delta,\Delta')$ which is such that $\Delta\prec\mathrm{cr}(\Delta,\Delta')$ and $\Delta'\prec\mathrm{cr}(\Delta,\Delta')$.
\end{Definition}

The set of geometric triangulations sharing the same root is directed, which means that for any two such triangulations one can construct a common refinement.

\section{Classical phase space}
\label{sec:classical}

\noindent In this section we describe the basic phase space functionals that are of interest for our construction. We will first define and explain the various observables on a fixed triangulation and its dual complex, and then discuss the consistency relations that arise if one wants to connect observables based on triangulations related by a refinement.

The configuration space for theories of connection like LQG is the space $\mathcal{A}$ of smooth connections on a principal $G$-bundle over a base manifold which here is the spatial hypersurface $\Sigma$ (see \cite{lqg1,lqg2,lqg3} for an introduction). By choosing a local trivialization of this bundle, one can see the connection $A\in\Omega^1(\mathfrak{g},\Sigma)$ as a Lie algebra-valued 1-form, and its conjugate variable on the phase space $T^*\mathcal{A}$ as a $\mathfrak{g}$-valued $(d-1)$-form $E$. We will from now on specialize to the case $G=\SU(2)$, and choose the generators of the Lie algebra $\mathfrak{g}=\su(2)$ to be $\tau_i=-\i\sigma_i/2$, where $\sigma_i$ for $i=1,2,3$ are the Pauli matrices, and in terms of which the $\su(2)$ commutation relations are $[\tau_i,\tau_j]=\eps^{ijk}\tau_k$. In terms of these generators, the connection and its conjugated electric field can be written as $A^i\tau_i$ and $E_i\tau^i$.

Usually, one starts with the gauge-variant phase space, which is parametrized by (possibly open) holonomies and smeared fluxes. For a fixed dual graph, this amounts to having a pair of holonomy and flux variables for every link of the graph (see for instance \cite{husain,qsd7}).

Here we will however restrict ourselves to the $\SU(2)$ (almost) gauge-invariant phase space, and leave a global symmetry by the adjoint action on the root node (see also \cite{eteraSL2C,bahrthiemann2} for a gauge-fixed description). Therefore, the phase space associated to a fixed triangulation will be parametrized by closed holonomies and conjugated simplicial flux observables transported to the root node. Whereas the usual discrete phase space is equivalent to $(T^*\SU(2))^{|l|}$ where $|l|$ is the number of links of $\Gamma$, the gauge-reduced phase space will be given by $(T^*\SU(2))^{|\ell|}$ where $|\ell|$ is the number of leaves of $\Gamma$ (or equivalently the number of independent cycles). This will become clear with the Poisson bracket structure for this almost gauge-invariant phase space, which we give below.

Let us first start by recalling some basic facts about the usual LQG holonomy-flux phase space on a single link.

\subsection{Gauge-variant phase space}
\label{gauge-variantPS}

\noindent We recall in this subsection the structure of the phase space of LQG on a fixed graph $\Gamma$ dual to a $d$-dimensional triangulation (see also \cite{qsd7}). This phase space is parametrized by holonomies associated to the links $l$ of the graph, and by simplicial fluxes associated to the $(d-1)$-simplices $l^*$ dual to these links.

Let us assume that the oriented links $l(s)$ are parametrized by a continuous parameter $s\in[0,1]$ such that $l(0)$ is the source node of the link and $l(1)$ is its target node. The holonomy associated with this link is given by the path ordered exponential\footnote{Throughout this work, we will use $h$ to denote holonomies along single links, $g$ to denote holonomies along paths, and $\mathrm{g}$ to denote parameters of gauge transformations. When writing the holonomy along a path $\gamma$, we will use as a subscript either the path itself and write $g_\gamma$, or the source and target nodes and write $g_{\gamma(0)\gamma(1)}$.}
\be
h_l(A)=h_{l(0)l(1)}(A)=\Pexp\left(-\int_{l(0)}^{l(1)}A_a\big(l(s)\big)\dot{l}^a(s)\de s\right),
\ee
where $A=A^i_a\tau_i\de x^a\in\Omega^1(\su(2),\Sigma)$ is the connection. Under an orientation reversal $l\mapsto l^{-1}$ of the link, the holonomy becomes $h_l\mapsto h_{l^{-1}}=h_l^{-1}$, and under the action of finite $\SU(2)$ gauge transformations one has
\be
\mathrm{g}_n\triangleright h_l=\g_{l(1)}h_l\g_{l(0)}^{-1},
\ee
where $\mathrm{g}_n$ is a group element acting at the source and target nodes of the link. If $l_1$ and $l_2$ are two consecutive oriented links with holonomies $h_1$ and $h_2$, and such that $l_1(1)=l_2(0)$, we define the composition of holonomies along the path $l_2\circ l_1$ from $l_1(0)$ to $l_2(1)$ as $g_{l_1(0)l_2(1)}=h_2h_1$.

The conjugate variables to the holonomies are the so-called simplicial (or geometrical) fluxes. A simplicial flux is associated to a link $l$, and defined as the integral of a $(d-1)$-form over the $(d-1)$-dimensional simplex dual to the link $l$ (i.e. an edge in the case $d=2$ and a triangle in the case $d=3$). The explicit definition is given by
\be\label{simplicial flux}
X_l\coloneqq\int_{l^*}g^{-1}_{l(0)l^*(s)}(\star E)\big(l^*(s)\big)g_{l(0)l^*(s)}\de^{d-1}s,
\ee
where $l^*$ denotes the $(d-1)$-dimensional simplex dual to the link $l$, the object $\star E$ is a $(d-1)$-form obtained by dualizing $E^a_i$ in its spatial indices, and $g_{l(0)l^*(s)}$ is the holonomy that starts at the source $l(0)$ of the link $l$, goes along $l$ to the intersection point $u=l\cap l^*$, and then goes from $u$ to a point $s$ in $l^*$. Note that this parallel transport depends on a choice of path, which we live implicit for the sake of notational simplicity. There is in fact a canonical choice for the path in $d=2$ but not for $d=3$.

The simplicial fluxes differ from the standard fluxes one usually uses in LQG by the presence of an explicit parallel transport in \eqref{simplicial flux}. As we will see, the notion of cylindrical consistency that is imposed by using the BF dynamics for refining a given state, involves the addition of simplicial fluxes associated to subdivided edges or triangles. This addition has to take place in a common reference frame, which is why we indeed need the parallel transport in the fluxes.

These simplicial fluxes present the advantage of transforming locally and in a covariant way under gauge transformations, and one has
\be
\mathrm{g}_n\triangleright X_l=\text{Ad}_{\g_{l(0)}}(X_l)=\g_{l(0)}X_l\g_{l(0)}^{-1}.
\ee
Under an orientation reversal of the link, the fluxes transform in the following way:
\be\label{fluxinverse}
X_l\mapsto X_{l^{-1}}=-\text{Ad}_{h_l}(X_l)=-h_lX_lh_l^{-1},
\ee
where $h_l$ is the holonomy associated to the link.

The Poisson brackets between the holonomies and the simplicial fluxes can be computed from the knowledge of the basic continuum Poisson brackets between the connection and the electric field \cite{qsd7}. As is well-known, the holonomy-flux Poisson structure reproduces for each link that of the cotangent bundle $T^*\SU(2)=\SU(2)\time\su(2)^*$, and one has that
\be\label{SU(2)brackets1}
\lb X^i_l,X^j_{l'}\rb=\delta_{l,l'}\eps^{ijk}X^k_l,\q
\lb X^i_l,h_{l'}\rb=\delta_{l,l'}h_l\tau^i-\delta_{l^{-1},l'}\tau^ih_{l'},\q
\lb h_l,h_{l'}\rb=0.
\ee

\subsection{Gauge-invariant phase space}

\noindent The holonomy observables in which we are interested are closed holonomies starting and ending at the root, with loops going along the fundamental cycles of the dual graph $\Gamma$. As explained in the previous section, a description of these fundamental cycles can be obtained by choosing a spanning tree $\mathcal{T}$, the leaves $\ell$ of which are in one-to-one correspondence with the fundamental cycles $c_\ell$. Every such cycle contains exactly one leaf $\ell$, and an arbitrary number $|b_c|\geq2$ of branches. A choice of spanning tree defines a unique path between any two nodes of the graph, with this path going along branches of the tree only. Similarly, for a rooted tree there is a unique path between the root $r$ and any node $n$.

This path can be used to define an holonomy associated to a fundamental cycle $c_\ell$ and based at the root in the following canonical manner. Let $g_{r\ell(0)}$ be the holonomy that starts at the root $r$ and goes to the source node $\ell(0)$ of the leaf $\ell$ along the unique path in the tree. One can then define the closed holonomy
\be\label{leaf closed holonomy}
g_\ell\coloneqq g^{-1}_{r\ell(0)}\left(\overrightarrow{\prod_{b\in c_\ell}}h_b\right)h_\ell g_{r\ell(0)}=g^{-1}_{r\ell^{-1}(0)}h_\ell g_{r\ell(0)^{\vp}}.
\ee
This holonomy goes from the root $r$ to the source $\ell(0)$ of the leaf $\ell$, then along the cycle $c_\ell$ following the orientation of the leaf, and then from the source $\ell(0)$ back to the root. In this expression $h_b$ denotes a group element associated to a branch of the tree, $h_\ell$ is the holonomy along the leaf itself, and $\ell^{-1}$ denotes the leaf with opposite orientation. We are going to use this set of holonomies $g_\ell$ associated to the fundamental cycles as our point-separating set for the gauge-invariant configuration space.

The tree can be used to perform a gauge fixing of the gauge freedom at all the nodes except the root. This can be done by simply setting all the group elements associated to the branches of the tree to the identity, i.e. $h_b=\openone,\ \forall\,b\in\mathcal{T}$.

The conjugated variables to the holonomies $g_\ell$ are the simplicial fluxes $X_\ell$ associated to the leaves and transported to the root along the unique path defined by the tree. We call these variables the rooted fluxes and denote them by
\be\label{rooted fluxes}
\bX_\ell\coloneqq g^{-1}_{r\ell(0)}X_\ell g_{r\ell(0)}=\text{Ad}_{g_{r\ell(0)}^{-1}}(X_\ell).
\ee
This definition is taking the simplicial flux $X_\ell$, which is defined in the frame of the $d$-dimensional simplex dual to the node $\ell(0)$, and transporting it to the frame of the $d$-dimensional simplex dual to the root.

Now, using \eqref{SU(2)brackets1}, one can find the Poisson brackets between the phase space functions $g_\ell$ and $\bX_\ell$. These reproduce the symplectic structure of $(T^*\SU(2))^{|\ell|}$ and are given by
\be\label{SU(2)brackets}
\lb\bX^i_\ell,\bX^j_{\ell'}\rb=\delta_{\ell,\ell'}\eps^{ijk}\bX^k_\ell,\q
\lb\bX^i_\ell,g_{\ell'}\rb=\delta_{\ell,\ell'}g_\ell\tau^i-\delta_{\ell^{-1},\ell'}\tau^ig_{\ell'},\q
\lb g_\ell,g_{\ell'}\rb=0.
\ee
More complicated phase space functions can now be constructed starting from this basic set of holonomies and fluxes associated to the leaves and transported to the root. As noted earlier, the leaves define a set of fundamental cycles from which one can describe all possible cycles of the graph $\Gamma$.

So far we have only considered the fluxes associated to the leaves, and we need to show that it is also possible to reconstruct out of them the fluxes associated to the branches. To this end, recall that in terms of the simplicial fluxes the Gauss law at a node $n\in\Gamma$ is given by
\be\label{gauss1}
\mathcal{G}_n\coloneqq\sum_{l|l(0)=n}X_l+\sum_{l|l(1)=n}X_{l^{-1}}=0.
\ee
We need to express this constraint for the node $n$ in terms of the rooted fluxes. For links $l$ such that $l(0)=n$, we can simply parallel transport the fluxes to the root with the holonomy $g_{rl(0)}=g_{rn}$. Now, with $\bX_l=g^{-1}_{rl(0)}X_lg_{rl(0)}$ being the rooted flux associated to a link $l\in\Gamma$, we can define the rooted flux associated to the inverse link as\footnote{For branches, this formula becomes simply $\bX_{b^{-1}}=-\bX_b$. Because the paths from $r$ to $l(0)$ go only along branches of the tree, the holonomies $g_{rb(0)}$ and $g_{rb^{-1}(0)}$ differ only by an holonomy $h_{b(0)b(1)}$ along the branch itself, which can then be used in relation \eqref{fluxinverse}.} $\bX_{l^{-1}}=g^{-1}_{rl^{-1}(0)}X_{l^{-1}}g_{rl^{-1}(0)}$. For links with $l(1)=n$, we obtain $\bX_{l^{-1}}=g^{-1}_{rn}X_{l^{-1}}g_{rn}$. Therefore, we can just parallel transport all terms in equation \eqref{gauss1}, and arrive at the Gauss constraint in the form
\be
g^{-1}_{rn}\mathcal{G}_ng_{rn}=\sum_{l|l(0)=n}\bX_l+\sum_{l|l(1)=n}\bX_{l^{-1}}=0.
\ee
One can use a tree to solve the Gauss constraints (except the one at the root) iteratively and to find in this way the fluxes associated to all the branches of this tree.

We are now going to discuss more elaborate versions of the rooted fluxes, and in particular define integrated fluxes associated to a set of edges in $d=2$ and a set of triangles in $d=3$. In $d=2$, these integrated flux observables generalize the Dirac observables introduced for $(2+1)$-dimensional gravity with point particles in \cite{FL1}. For $d=3$, the integrated fluxes are related to the so-called Wilson surfaces operators \cite{wilsonsurfaces1,wilsonsurfaces2}.

\subsection{Integrated fluxes}
\label{sec:intflux}

\noindent The integrated fluxes $\bX_\pi$ are constructed from the elementary simplicial fluxes $X_l$, but instead of being associated to a single link (i.e. to a single edge or triangle), they are labeled by co-paths $\pi$ of $(d-1)$-dimensional simplices in the triangulation $\Delta$, that is, a collection of adjacent edges or triangles depending on the dimension. The definition of these co-paths is as follows.

\begin{Definition}[Co-path]\label{def:path}
A co-path $\pi$ in $\Delta$ is a collection of adjacent $(d-1)$-dimensional simplices connected via $(d-2)$-dimensional simplices. We require these co-paths to be such that every $(d-2)$-dimensional simplex of $\pi$ is shared by at most two $(d-1)$-dimensional simplices of $\pi$ (and a possibly arbitrary number of $(d-1)$-dimensional simplices of $\Delta/\pi$). Since all the $(d-1)$-dimensional simplices are oriented and their orientation can be reversed, one can always choose the same orientation for all the elements of $\pi$ and thereby define a global orientation for $\pi$.
\end{Definition}

By reversing the orientation of a $(d-1)$-dimensional simplex, we mean that one has to consider the flux element $X_{l^{-1}}$ defined in \eqref{fluxinverse} instead of $X_l$.

Note that in this definition we do not allow for self-intersections (along $(d-2)$-dimensional simplices) of $\pi$, as these would lead to situations in which more than two $(d-1)$-dimensional simplices of $\pi$ share a $(d-2)$-dimensional simplex. However, one may form self-intersecting co-paths by composing (more elementary) integrated fluxes. We will discuss this operation of composition later on.

If we denote by $|\pi|$ the number of $(d-1)$-dimensional simplices in the co-path $\pi$, we can label its elements by $l^*_1,\ldots,l^*_{|\pi|}$. This corresponds to a collection $e_1,\ldots,e_{|\pi|}$ of edges when $d=2$, and to a collection $t_1,\ldots,t_{|\pi|}$ of triangles when $d=3$. Notice that these are a priori unordered sets of simplices. However, in $d=2$ the global orientation of the co-path $\pi$ induces naturally a total order on the set of its edges, and one can unambiguously call $e_1$ and $e_{|\pi|}$ the first and last edges of $\pi$. In $d=3$ this is not true anymore, but although the set of triangles of a co-path $\pi$ is unordered we will still need to choose a ``first'' triangle since the definition of the integrated fluxes requires a choice of common frame. This common frame will therefore be chosen arbitrarily in $d=3$, and correspond to the node $l_1(0)$ dual to the frame in which the first triangle of $\pi$ is defined. Likewise, in $d=2$ it will be the node $l_1(0)$ dual to the frame in which the first edge of $\pi$ is defined.

To be more precise, the integrated fluxes are defined by transporting the individual fluxes $X_l$ associated to the $(d-1)$-dimensional simplices of $\pi$ into a common frame where they can be added, and then transporting the resulting sum of fluxes to the root. We choose the common frame to be the $d$-dimensional simplex dual to the node $l_1(0)$ of the first link. As we will see, it turns out to be more convenient to define the transport to this common frame independently from the choice of tree for the entire triangulation which could be used to define a the further parallel transport to the root. In $d=2$ spatial dimensions, the sole knowledge of the co-path $\pi$ can be used to define a canonical path in the dual graph $\Gamma$. This canonical path in $\Gamma$ defined by $\pi$ goes along the shadow graph of $\pi$, which we define below. In $d=3$, this notion of shadow graph does not specify a unique parallel transport, and we will need to further specify a choice of shadow tree of this shadow graph.

\begin{Definition}[Shadow graph]\label{def:shadow graph}
The shadow graph $\Gamma_\pi$ of a $(d-1)$-dimensional oriented co-path $\pi$ is a connected subgraph of $\Gamma$ which is uniquely defined by $\pi$ and its global orientation in the following way. The links of the shadow graph connect the source nodes $l(0)$ of all the links $l$ dual to the $(d-1)$-dimensional simplices of $\pi$, while staying as close as possible to $\pi$ in the sense that:\\
i) The simplices dual to the nodes and links of the shadow graph are included in the union of the stars of the (sub-) simplices which form the co-path $\pi$;\\
ii) The links of the shadow graph do not cross the co-path\footnote{If one allows for self-intersections in $\pi$ and crossing links are needed in order to connect all the source nodes $l(0)$, then we require that the only $(d-1)$-dimensional simplices of $\pi$ which should be crossed are those included in the union of the stars of the simplices forming the intersection. In general, the shadow graph in the case of self-intersecting co-paths can be defined by first decomposing the co-paths into non-self-intersecting (but mutually intersecting) parts, and then connecting these parts together.} $\pi$.
\end{Definition}

\begin{center}
\begin{figure}[h]
\includegraphics[scale=0.7]{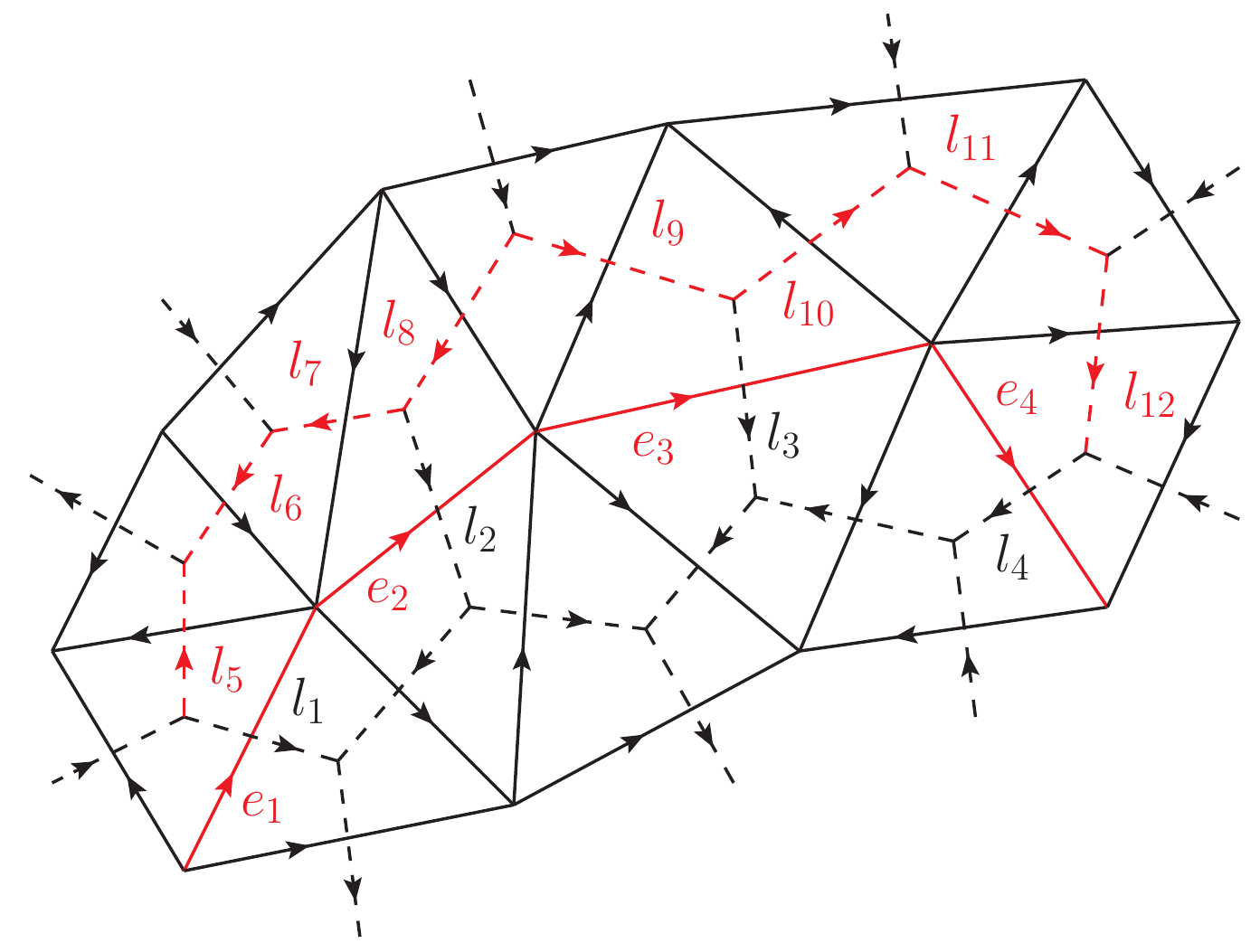}
\caption{Example of a 1-dimensional co-path $\pi$ (solid red) which consists of the four edges $e_1,\ldots,e_4$, and its shadow graph $\Gamma_\pi$ (dashed red) which consists of the links $l_5,\ldots,l_{12}$ connecting the source nodes of the links $l_1,\ldots,l_4$ while staying as close as possible to $\pi$.}
\label{fig:def-path-2d}
\end{figure}
\end{center}

In the case $d=2$, there is always a unique path going along $\Gamma_\pi$ between the reference nodes of any two fluxes in $\pi$ (as can be seen on the example of figure \ref{fig:def-path-2d}). The reason for this is that $\Gamma_\pi$ always has the structure of a spanning tree for the nodes $l_1(0),\ldots,l_{|\pi|}(0)$. In this sense, for $d=2$, the knowledge of $\pi$ is enough in order to uniquely define, via $\Gamma_\pi$, the parallel transport of the individual simplicial fluxes to the common frame $l_1(0)$.

For $d=3$ however, this is not the case anymore since the shadow graph of a 2-dimensional co-path $\pi$ can have a complicated structure and in particular contain closed cycles (it is therefore not a tree). This prevents the parallel transport along its links from being uniquely defined. Therefore, in the case $d=3$ it will be necessary to further introduce a choice of shadow tree $\mathcal{T}_\pi$ in the shadow graph $\Gamma_\pi$. It is then possible to uniquely define, via this tree $\mathcal{T}_\pi$, the parallel transport of the individual simplicial fluxes to the common frame $l_1(0)$.

In order to fix the notations, let us denote the path between the common frame $l_1(0)$ and the reference frame $l_i(0)$ of a flux $X_{l_i}$ (for $i=1,\ldots,|\pi|$) by $\gamma(l_1(0),l_i(0))$, and the corresponding holonomy by $g_{l_1(0)l_i(0)}$. This path always goes along $\Gamma_\pi$ and is uniquely defined in $d=2$, while in $d=3$ it requires a choice of tree $\mathcal{T}_\pi$ in $\Gamma_\pi$. Let us now discuss the precise definition of the integrated fluxes. Since, in light of the above discussion, this definition does depend slightly on the dimension $d$, we study the two cases of interest separately.

\subsubsection{Integrated fluxes in $d=2$ spatial dimensions}

\noindent In the case $d=2$, the links $l$ of the graph $\Gamma$ are dual to the 1-simplices of $\Delta$ that we call edges $e$. Our convention is such that the pairs $(l,e)$ are positively oriented, in the sense that $l$ is pointing to the right if $e$ is pointing upwards (as can be seen on figure \ref{fig:def-path-2d}).

In the previous subsection we have introduced the rooted fluxes $\bX_l$ dual to links $l$ and transported to the root. We are now going to introduce integrated fluxes associated to co-paths $\pi$ in the 1-skeleton $\Delta_1$ and transported to the root. A co-path $\pi$ in $\Delta_1$, as defined in definition \ref{def:path}, consists of a collection of adjacent edges $e_1,\ldots,e_{|\pi|}$, where $|\pi|$ is the number of these edges. The individual orientation of these edges can always be adjusted in such a way that they all have the same orientation, which defines the global orientation of the co-path $\pi$, and ensures that the beginning $\pi(0)$ of the co-path coincides with the beginning of the first edge. By virtue of definition \eqref{simplicial flux}, each simplicial flux $X_l$ is defined in the reference frame of the triangle dual to the node $l(0)$. Therefore, in order to sum each of the fluxes $X_l$ associated with the edges of the co-path $\pi$, these have to be transported to a common frame, which we choose to be the source $l_1(0)$ of the link $l_1$ dual to the first edge $e_1$ of $\pi$ (if the edge is pointing upwards, this is the triangle on its left). We therefore need to introduce, for each simplicial flux $X_l$, a path in $\Gamma$ going from $l_1(0)$ to $l(0)$. This path can be defined canonically by going along the shadow graph $\Gamma_\pi$ of the co-path $\pi$. As introduced in definition \ref{def:shadow graph}, this shadow graph connects all the nodes $l_1(0),\ldots,l_{|\pi|}(0)$ while staying as close as possible to the edges $e$ of $\pi$, and goes through the triangles to the left of $\pi$ (as seen when the edges are pointing upwards). We can then define the integrated fluxes
\be\label{2d integrated flux}
\bX_\pi\coloneqq g^{-1}_{rl_1(0)}\left(\sum_{i=1}^{|\pi|}g^{-1}_{l_1(0)l_i(0)}X_{l_i}g_{l_1(0)l_i(0)}\right)g_{rl_1(0)}.
\ee
In this formula, $g_{rl_1(0)}$ is an holonomy going from the root node to the node $l_1(0)$ where all the fluxes are transported and summed, and $g_{l_1(0)l_i(0)}$ is the holonomy along the unique path in the shadow graph $\Gamma_\pi$ that goes from $l_1(0)$ to the reference frame $l_i(0)$ of each individual flux $X_{l_i}$ (for the first flux this is therefore the identity). An example is represented in figure \ref{fig:def-intflux-2d}. Notice that we label the integrated flux $\bX_\pi$ only by the co-path $\pi$, since the path in $\Gamma$ used for the parallel transport can be defined canonically from the knowledge of $\pi$ by using the shadow graph $\Gamma_\pi$. For the sake of notational simplicity, we also drop the explicit dependence of the integrated fluxes on the path used to parallel transport from $l_1(0)$ to the root.

\begin{center}
\begin{figure}[h]
\includegraphics[scale=0.7]{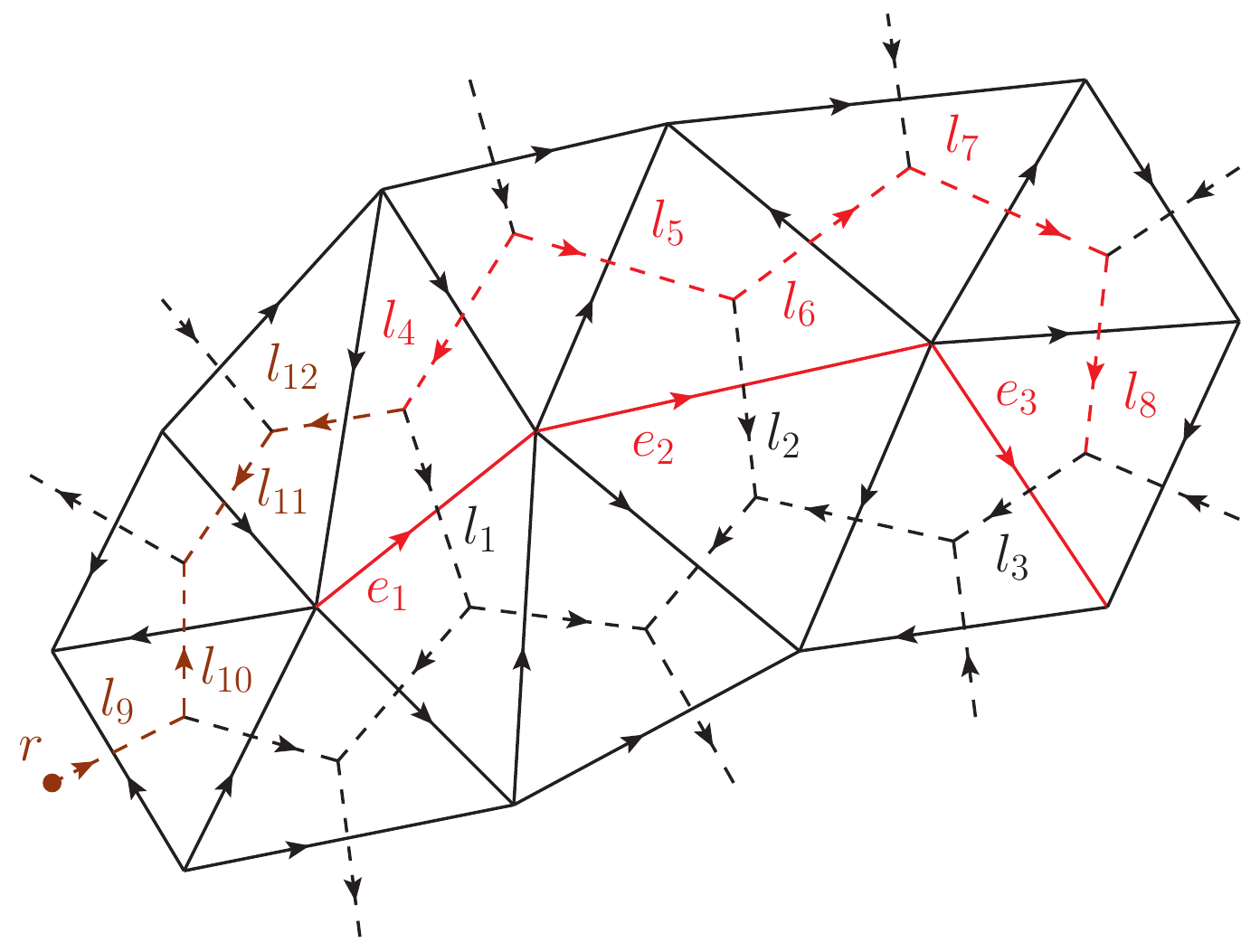}
\caption{Construction of a 2-dimensional integrated simplicial flux $\bX_\pi$ according to definition \eqref{2d integrated flux}. The holonomy $g_{rl_1(0)}=h_{12}^{-1}h_{11}^{-1}h_{10}h_9$ goes from the root $r$ to the node $l_1(0)$ along the tree (of portion of which is represented in dashed brown). The holonomy $g_{l_1(0)l_i(0)}$ goes from the node $l_1(0)$ to the node $l_i(0)$ along the shadow graph $\Gamma_\pi$ (dashed red). For example, we have $g_{l_1(0)l_2(0)}=h_5h_4^{-1}$, and $g_{l_1(0)l_3(0)}=h_8h_7h_6h_5h_4^{-1}$.}
\label{fig:def-intflux-2d}
\end{figure}
\end{center}

These integrated fluxes can be composed in a natural manner that involves only closed holonomies and integrated fluxes. To see this, consider two consecutive co-paths $\pi_1$ and $\pi_2$, i.e. such that the end vertex of $\pi_1$ is the first vertex of $\pi_2$, and denote their respective edges by $e^1_1,\ldots,e^1_{|\pi_1|}$ and $e^2_1,\ldots,e^2_{|\pi_2|}$. Let $\gamma_i\coloneqq\gamma(l^i_1(0),l^i_{|\pi_i|}(0))$ be the path in $\Gamma_{\pi_i}$ going from the triangle on the left of the first edge of $\pi_i$ to the one on the left of the final edge of $\pi_i$. Furthermore, we denote by $\gamma_{12}\coloneqq\gamma(l^1_{|\pi_1|}(0),l^2_1(0))$ the path connecting $l^1_{|\pi_1|}(0)$ to $l^2_1(0)$, i.e. the path going from the triangle on the left of the final edge of $\pi_1$ to the triangle on the left of the first edge of $\pi_2$. This path is defined as before as being as close as possible and to the left of the composed co-path $\pi_2\circ\pi_1$. We can then define a loop associated to $\pi_1$ and $\pi_2$ as
\be\label{2d composition loop}
\lambda_{12}\coloneqq\gamma(r,l^2_1(0))^{-1}\circ\gamma_{12}\circ\gamma_1\circ\gamma(r,l^1_1(0)),
\ee
where $\gamma(r,l^i_1(0))$ denotes the unique path in the tree going from the root to the starting node $l^i_1(0)$. With this data, we can finally define the composition of two fluxes $\bX_{\pi_1}$ and $\bX_{\pi_2}$ as
\be\label{compi2d}
\bX_{\pi_2\circ\pi_1}\coloneqq\bX_{\pi_2}\circ\bX_{\pi_1}=\bX_{\pi_1}+g^{-1}_{\lambda_{12}}\bX_{\pi_2}g_{\lambda_{12}},
\ee
where
\be
g_{\lambda_{12}}=g^{-1}_{rl^2_1(0)}g_{\gamma_{12}}g_{\gamma_1}g_{rl^1_1(0)}=g^{-1}_{rl^2_1(0)}g_{l^1_{|\pi_1|}(0)l^2_1(0)}g_{l^1_1(0)l^1_{|\pi_1|}(0)}g_{rl^1_1(0)}
\ee
is the holonomy along the loop $\lambda_{12}$. An example of this construction is given in figure \ref{fig:def-composition-2d}. This composition rule corresponds to (one component of) a semi-direct product structure, with the group acting on its Lie algebra via the adjoint action.

\begin{center}
\begin{figure}[h]
\includegraphics[scale=0.7]{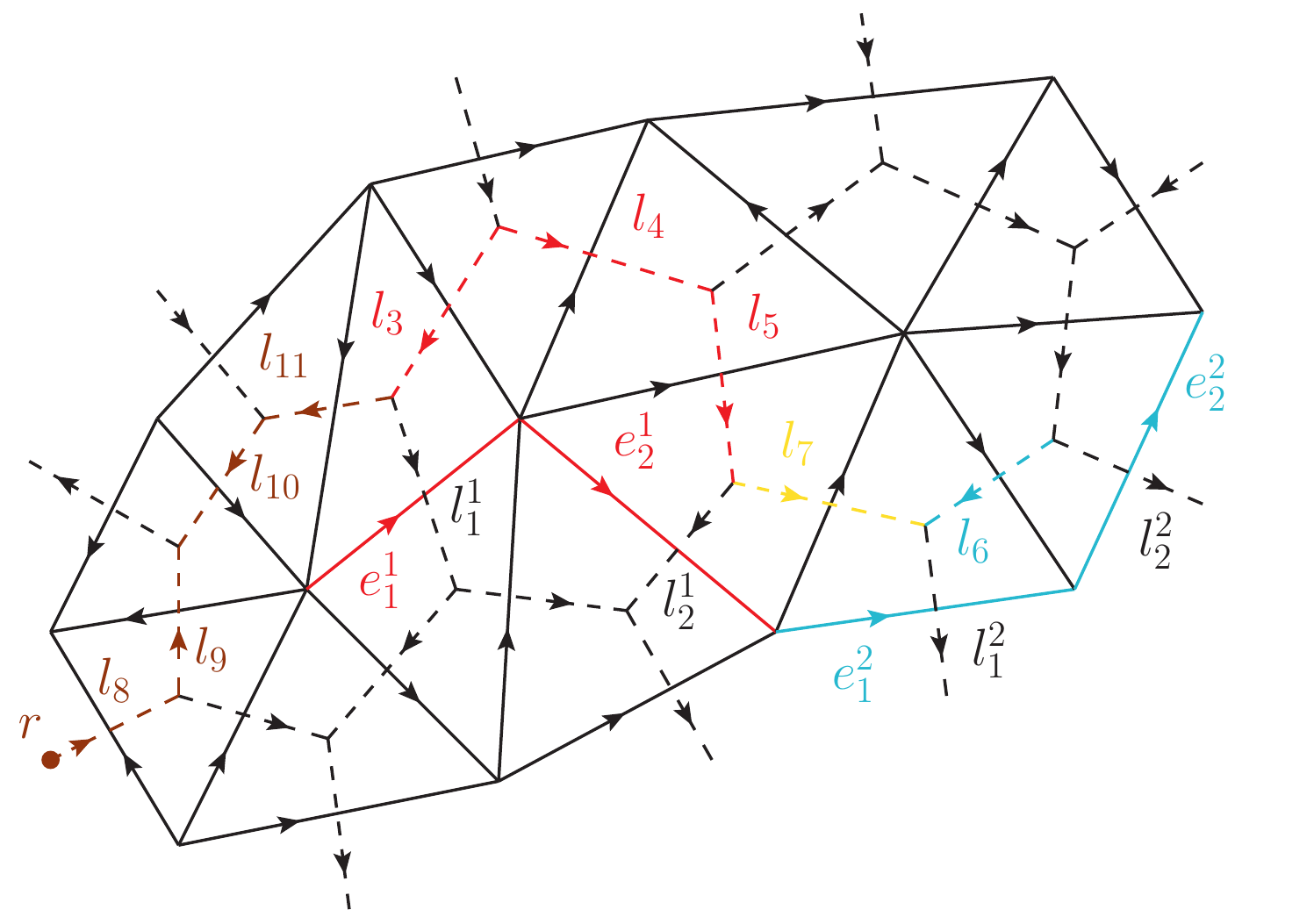}
\caption{Composition of the two integrated fluxes $\bX_{\pi_1}$ and $\bX_{\pi_2}$ associated to the co-paths $\pi_1=e^1_2\circ e^1_1$ (solid red) and $\pi_2=e^2_2\circ e^2_1$ (solid blue). The integrated flux $\bX_{\pi_1}$ is given by $\bX_{\pi_1}=g^{-1}_{rl^1_1(0)}\left(X_{l^1_1}+g^{-1}_{l^1_1(0)l^1_2(0)}X_{l^1_2}g_{l^1_1(0)l^1_2(0)}\right)g_{rl^1_1(0)}$, where $g_{l^1_1(0)l^1_2(0)}=h_5h_4h_3^{-1}$ and $g_{rl^1_1(0)}=h_{11}^{-1}h_{10}^{-1}h_9h_8$. The integrated flux $\bX_{\pi_2}$ is given by $\bX_{\pi_2}=g^{-1}_{rl^2_1(0)}\left(X_{l^2_1}+g^{-1}_{l^2_1(0)l^2_2(0)}X_{l^2_2}g_{l^2_1(0)l^2_2(0)}\right)g_{rl^2_1(0)}$, where $g_{l^2_1(0)l^2_2(0)}=h_6^{-1}$ and $g_{rl^2_1(0)}$ is the holonomy along the tree from the root node to the node $l^2_1(0)$ (and is not represented for the sake of clarity). To compose the two fluxes, one has to take $\bX_{\pi_2}$, undo its parallel transport to the root using $g_{rl^2_1(0)}$, transport it to the frame $l^1_2(0)$ using the holonomy $h_{7}$ along the path $\gamma_{12}$ connecting $l^1_2(0)$ to $l^2_1(0)$, transport it to the frame $l^1_1(0)$ using the holonomy $g_{l^1_1(0)l^1_2(0)}$, transport it to the root using the holonomy $g_{rl^1_1(0)}$, and finally add it with $\bX_{\pi_1}$.}
\label{fig:def-composition-2d}
\end{figure}
\end{center}

\subsubsection{Integrated fluxes in $d=3$ spatial dimensions}
\label{sf3d}

\noindent The integrated fluxes in $d=3$ dimensions are associated to 2-dimensional co-paths $\pi$ in the triangulation. Such a surface path consists of a collection of adjacent triangles as defined in \ref{def:path}. Our convention is such that a triangle and its dual link have a direct orientation, in the sense that if the triangle has a counter-clockwise (clockwise) orientation then the source of its dual link is under (above) it. For simplicity, we will consider only edge-connected surface paths. This means that if the total surface does not consist only of one triangle, then every triangle of the surface shares at least one edge with some other triangle of this surface. Also, we require the triangles to be in a consistent orientation, so that we can assign a global orientation to the surface.

As mentioned above, the set of triangles of a 2-dimensional co-path $\pi$ being unordered, the notion of first triangle $l^*_1=t_1$ exists only because we choose a reference frame $l_1(0)$, and there is no notion of ``last triangle'' of $\pi$. The definition of the integrated flux associated to a surface co-path $\pi$ requires, like in the previous subsection, a path in $\Gamma$ in order to bring all the individual fluxes in the same reference frame $l_1(0)$. However, we saw that at the difference with the case $d=2$ treated in the previous subsection, in the case $d=3$ this path cannot be defined canonically in the shadow graph $\Gamma_\pi$. Therefore, the definition of the integrated flux associated with a surface path $\pi$ requires an additional structure. This additional structure is given by a spanning tree in the dual of the surface co-path $\pi$ seen as a 2-dimensional triangulation, or equivalently by a shadow spanning tree $\mathcal{T}_\pi$ in the shadow graph $\Gamma_\pi$ (an example is given on figure \ref{fig:def-intflux-3d}). Given such a spanning tree in $\Gamma_\pi$, there is then a unique path going from the source node $l_1(0)$ (dual to the tetrahedron below the first triangle of $\pi$), to the source node $l_i(0)$ of the link dual to the triangle of interest. This path is again as close as possible to the surface path $\pi$, and going along the shadow tree $\mathcal{T}_\pi$ through the tetrahedra below $\pi$. The integrated fluxes are then defined as
\be
\bX_\pi\coloneqq g^{-1}_{rl_1(0)}\left(\sum_{i=1}^{|\pi|}g^{-1}_{l_1(0)l_i(0)}X_{l_i}g_{l_1(0)l_i(0)}\right)g_{rl_1(0)},
\ee
where $g_{rl_1(0)}$ is the holonomy from the root of $\Gamma$ to the node $l_1(0)$ dual to the tetrahedron under the first triangle of $\pi$.

\begin{center}
\begin{figure}[h]
\includegraphics[scale=0.7]{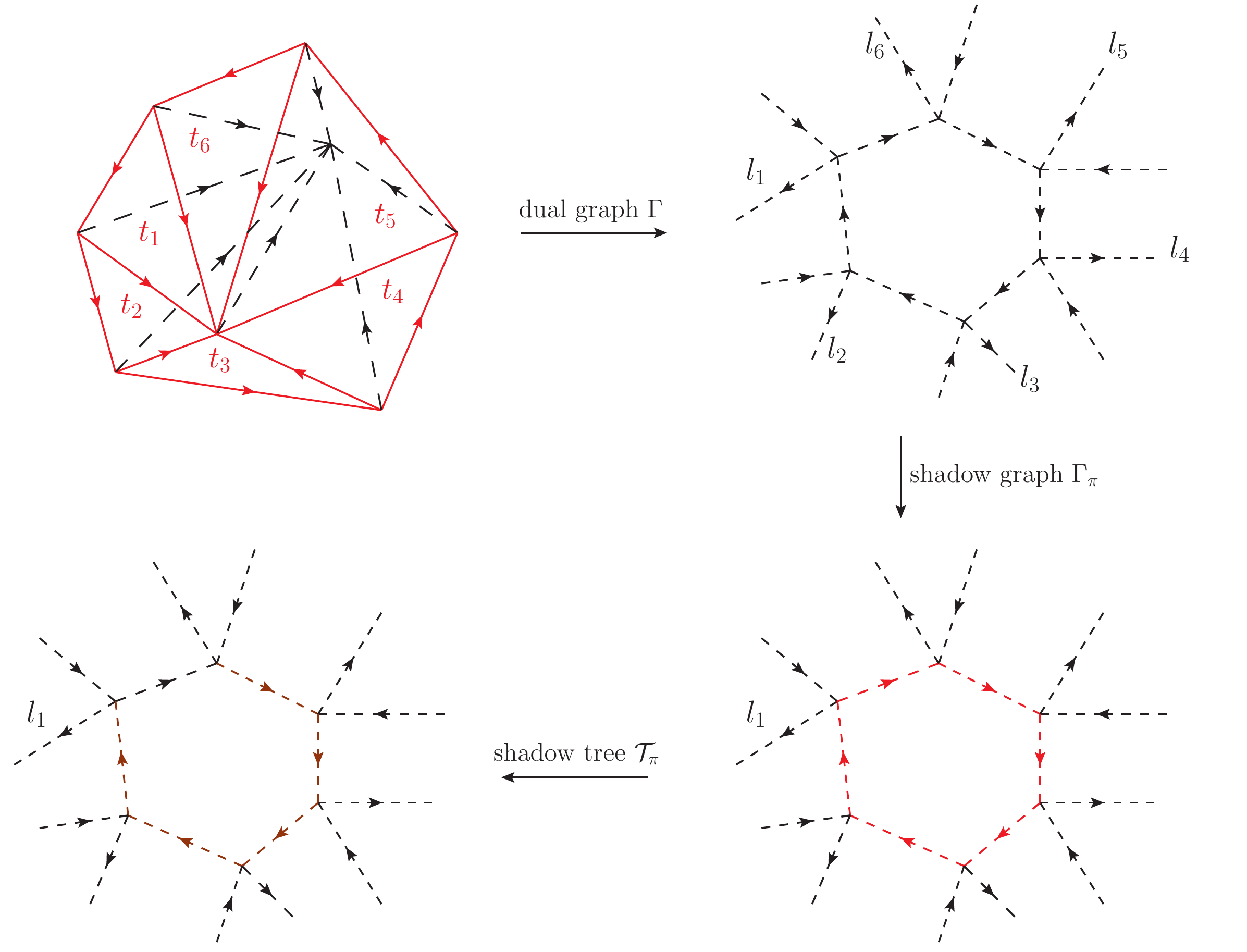}
\caption{A piece of 3-dimensional triangulation consisting of six tetrahedra glued together, and a choice of surface path $\pi$ consisting of six triangles (red). Because of the orientation of the edges, these six triangles all have a counter-clockwise orientation. The dual graph to this piece of triangulation consists of six 4-valent nodes, and the links dual to the triangles of $\pi$ are labelled by $l_1,\ldots,l_6$. The tetrahedron dual to the source $l_1(0)$ has been chosen as the reference frame in which the fluxes have to be transported and added. The shadow graph (dashed red) connects the source nodes of the links dual to the triangles of $\pi$, and a choice of shadow tree (dashed brown) enables to define uniquely the parallel transport between the frame of each flux and the frame $l_1(0)$. Notice that there are five other possible choices for the shadow tree.}
\label{fig:def-intflux-3d}
\end{figure}
\end{center}

These integrated fluxes can also be composed in a natural manner. We are going to describe this composition in the case of two surface paths $\pi_1$ and $\pi_2$ that have a disjoint set of triangles and at least one edge in common. More general situations are possible, for instance the gluing of a triangle onto itself but with an opposite orientation, but we will however not consider them here. We assume that the two surfaces $\pi_1$ and $\pi_2$ are such that the composed surface obtained by the gluing has a consistent orientation, and that the gluing is done along a common edge $e_{12}\in\pi_1\cap\pi_2$. For the sake of definiteness, the triangle of $\pi_i$ that contains the edge $e_{12}$ can then be called the gluing triangle $t^i_{|\pi_i|}$ of $\pi_i$. Therefore, in addition to the notion of first triangle $t^i_1$ for each path $\pi_i$, which is provided by the choice of reference frame $l^i_1(0)$, the existence of a gluing edge $e_{12}$ enables us to define a gluing triangle for each path $\pi_i$. Now, let $\gamma_i\coloneqq\gamma(l^i_1(0),l^i_{|\pi_i|}(0))$ be the path in the tree $\mathcal{T}_{\pi_i}$ of $\Gamma_{\pi_i}$ going from the tetrahedron under the first triangle of $\pi_i$ to the one under the gluing triangle of $\pi_i$. Furthermore, we denote by $\gamma_{12}$ the path connecting $l^1_{|\pi_1|}(0)$ to $l^2_{|\pi_2|}(0)$, i.e. the path going from the tetrahedron under the gluing triangle of $\pi_1$ to the tetrahedron under the gluing triangle of $\pi_2$ (notice the important difference with the case $d=2$). This path is defined as before as being as close as possible and under the composed surface $\pi_2\circ|_{e_{12}}\pi_1$, and its role is to glue the two shadow trees $\mathcal{T}_{\pi_1}$ and $\mathcal{T}_{\pi_2}$ together. We can then define the loop\footnote{Notice the presence of an additional path $\gamma_2^{-1}$ in this expression, as opposed to \eqref{2d composition loop}. This is due to the fact that the integrated flux $\bX_{\pi_2}$ can a priori be defined in a frame different from the triangle containing the edge $e_{12}$ used for the gluing. This means that the first triangle and the gluing triangle of $\pi_2$ can a priori be different. If they coincide we simply have $g_{\gamma_{12}}=\openone$.}
\be
\lambda_{12}\coloneqq\gamma(r,l^2_1(0))^{-1}\circ\gamma_2^{-1}\circ\gamma_{12}\circ\gamma_1\circ\gamma(r,l^1_1(0)),
\ee
where $\gamma(r,l^i_1(0))$ denotes a path along the dual graph $\Gamma$ going from the root to the reference node $l^i_1(0)$. With these data, we can finally define the composition of two fluxes $\bX_{\pi_1}$ and $\bX_{\pi_2}$ as
\be\label{compi3d}
\bX_{\pi_2\circ|_{e_{12}}\pi_1}\coloneqq\bX_{\pi_2}\circ|_{e_{12}}\bX_{\pi_1}=\bX_{\pi_1}+g^{-1}_{\lambda_{12}}\bX_{\pi_2}g_{\lambda_{12}}.
\ee
This composition corresponds to (one component of) a semi-direct product structure, with the group acting on its Lie algebra via the adjoint action. The new flux has a shadow graph which is obtained by connecting the shadow graphs $\Gamma_{\pi_1}$ and $\Gamma_{\pi_2}$ via $\gamma_{12}$, and the trees $\mathcal{T}_{\pi_1}$ and $\mathcal{T}_{\pi_2}$ are connected accordingly. The transport to the root of the composed flux is determined by the transport for the flux $\bX_{\pi_1}$.

\section[blabla]{Geometric interpretation}
\label{sec:geom}

\noindent In this section, we are going to present the geometric interpretation of the discrete phase spaces and the integrated flux obervables. In particular, we will see that the integrated flux observables can be used to define macroscopic variables which will be important for the coarse graining of spin networks \cite{eteradeformed,howmany} and spin foams \cite{eckert,holonomy,sffinite,decorated}. We will explain that the integrated (macroscopic) fluxes do not necessarily need to satisfy the Gauss constraints. In the context of coarse graining of spin networks, this was first observed in \cite{eteradeformed}. Here, we give a simple geometric interpretation of this phenomenon as ``curvature inducing torsion under coarse graining''. This effect is specific to non-Abelian groups.

\subsection{Geometric interpretation in $\boldsymbol{d=2}$ spatial dimensions}

\noindent In $d=2$ spatial dimensions, the phase space corresponds to the (kinematical) phase space of $(2+1)$-dimensional gravity. Since the physical solutions are locally flat, the holonomies around vertices should vanish. However, in the presence of point particles\footnote{Here the particles should be non-spinning since the Gauss constraint is assumed to hold. However, a generalization to defects that violate the Gauss constraints might be possible.} we will have curvature defects at the position of the particles. In fact, the (inductive limit) Hilbert space based on the BF vacuum constructed in \cite{paper1} can be understood as allowing physical solutions with an arbitrary number of particle insertions, and therefore allows for particles to meet or separate.

Note that although the Gauss-reduced phase space is still based on (a subset of) edges, all the gauge-invariant information is actually attached to the vertices. Therefore, the curvature defects are associated to the vertices of the triangulation. Since excitations in the new representation \cite{paper1} are described by the curvature defects, it is therefore appropriate to embed the vertices of the triangulation\footnote{For the AL representation one embeds the dual edges since these carry the (flux) excitations.}. We have defined the edges of the triangulation as arising as geodesics of an auxiliary (unphysical) metric. In fact, changing the embedding prescription of the edges will not change the physical content of the configurations.

This property, stating that the relevant geometric information is attached to the vertices, holds also with respect to the fluxes. There is however a slight caveat. Flux observables depend only on the homotopy class of the underlying co-path $\pi$ and the associated parallel transport curves. Here, the notion of homotopy equivalence treats every vertex of the triangulation as a puncture in the 2-dimensional manifold. This means that homotopy-equivalent curves (here those along which the parallel transport is defined) are the ones that can be deformed into each other without crossing any vertex. Heuristically, the deformation of the parallel transport across a vertex changes this parallel transport by the holonomy around the vertex. This prevents the Gauss constraints from being applicable, as we will show explicitly below. The equivalence classes can be extended in a phase space dependent way, which would in turn allow to cross vertices which are not carrying curvature.

We are going to explain these aspects in the following in more detail. The integrated flux observables are generalizations of Dirac observables, which are based on closed co-paths, and built from fluxes for $(2+1)$-dimensional gravity \cite{FL1}. In the context of $(2+1)$ gravity, i.e. for configurations with vanishing curvature, the integrated flux observables based on an open co-path $\pi$ correspond to the three-vector (in the frame of the root) pointing from the source vertex of the co-path $\pi$ to the target vertex of $\pi$. Naively, one might therefore expect that integrated fluxes based on closed co-paths should evaluate to zero due to the Gauss constraint. This makes it also clear that one should not associate the norm of the flux observables with the squared length of the paths. In fact, this association can only be done for the shortest possible paths, i.e. those consisting of only one edge.

Let us first discuss an example in which the co-path $\pi$ can be deformed without changing the value of the associated flux. This example is displayed in figure \ref{Fex1}. Here, we can consider the integrated flux $\bX_\pi$ associated to the co-path $\pi=e_7\circ e_1\circ e_8$, and whose expression is (we assume that the root is at the node $l_8(0)$ for notational simplicity)
\be\label{piflux}
\bX_\pi=X_{l_8}+g_{l_8(0)l_1(0)}^{-1}X_{l_1}g_{l_8(0)l_1(0)}+g_{l_8(0)l_7(0)}^{-1}X_{l_7}g_{l_8(0)l_7(0)}.
\ee
Now, because the edges $e_1$, $e_2$ and $e_3$ form a triangle, the closure constraint holds and takes the form
\be
X_{l^{-1}_1} +X_{l_2}+X_{l_3}=0,
\ee
which in turn implies that
\be\label{gau1}
X_{l_1}=h_1^{-1}X_{l_2}h_1+h_1^{-1}X_{l_3}h_1.
\ee
Using \eqref{gau1} in \eqref{piflux}, one can immediately see that $\bX_\pi=\bX_{\pi'}$, where $\pi'=e_7\circ e_3\circ e_2\circ e_8$ is the deformed path. The reason for which this equality holds is that the parallel transport necessary in order to reach $X_{l_2}$ and $X_{l_3}$ crosses (twice) the edge $e_1$, but not a vertex.

\begin{center}
\begin{figure}[h]
\includegraphics[scale=0.7]{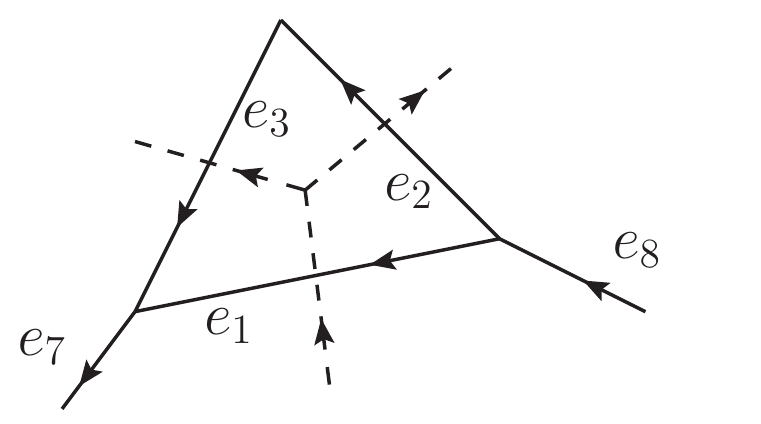}
\caption{Example of a triangle formed by the edges $e_1$, $e_2$, and $e_3$, and for which the path $\pi=e_7\circ e_1\circ e_8$ can be deformed to $\pi'=e_7\circ e_3\circ e_2\circ e_8$ without changing the associated integrated flux observable.}
\label{Fex1}
\end{figure}
\end{center}

This situation changes drastically if we invert the orientation of all the edges (and therefore of the co-path). The parallel transport (according to our conventions) that we have to consider in this case is indicated in figure \ref{Fex2}. In this example, we can consider the integrated flux $\bX_\pi$ associated to the co-path $\pi=e_8\circ e_1\circ e_7$, and the integrated flux $\bX_{\pi'}$ associated to the deformed co-path $\pi'=e_8\circ e_2\circ e_3\circ e_7$. It is easy to see that $\bX_{\pi'}$ can be obtained from $\bX_\pi$ by replacing $h_3^{-1}X_{l_1}h_3$ by $h_4^{-1}h^{-1}_5h_6^{-1}X_{l_2}h_6h_5h_4+X_{l_3}$. However, the Gauss constraint now implies that
\be\label{gau2}
h_3^{-1}X_{l_1}h_3=h_3^{-1}h_2X_{l_2}h_2^{-1}h_3+X_{l_3}=g^{-1}_vh_4^{-1}h^{-1}_5h_6^{-1}X_{l_2}h_6h_5h_4g_v+X_{l_3},
\ee
where $g_v=h_4^{-1}h_5^{-1}h_6^{-1}h_2^{-1}h_3$ is the holonomy around the vertex. We therefore see that, because of the presence of an additional adjoint action of $g_v$, the Gauss constraint cannot be used to rewrite $\bX_\pi$ as $\bX_{\pi'}$. Thus, in general we will have that $\bX_\pi\neq\bX_{\pi'}$ for this example. The reason for this is that the deformation of the parallel transport needed in order to go from $\bX_\pi$ to $\bX_{\pi'}$ crosses the vertex $v$. One could relax the definition of the parallel transport for the integrated fluxes in order for the Gauss constraint to be applicable, but this would delocalize the parallel transport from the co-path $\pi$, and just lead to more (dependent) flux observables, which just differ in their parallel transport.

\begin{center}
\begin{figure}[h]
\includegraphics[scale=0.7]{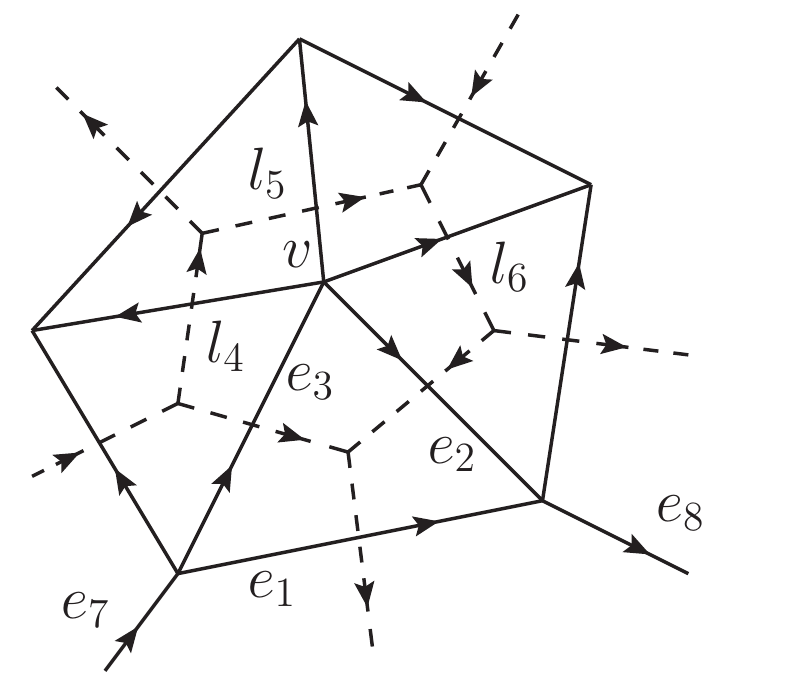}
\caption{Example of a triangle formed by the edges $e_1$, $e_2$, and $e_3$. The co-path $\pi=e_7\circ e_1\circ e_8$ is deformed to the co-path $\pi'=e_7\circ e_3\circ e_2\circ e_8$. Thus the co-paths are deformed into each other by adding or subtracting the triangle. If the vertex $v$ carries curvature, the Gauss constraint associated to the triangle deforming the co-paths can however not be used in order to make the two associated integrated fluxes coincide.}
\label{Fex2}
\end{figure}
\end{center}

Also, let us consider in both examples above the closed co-paths given respectively by $\pi=e_3\circ e_2\circ e_1$ and $\pi=e_1\circ e_2\circ e_3$, which are obtained by inverting the orientation of the edge $e_1$. One finds that in figure \ref{Fex1}, where the parallel transport for  adding up the fluxes is trivial, the closed co-path observable vanishes. However, in figure \ref{Fex2}, the parallel transport goes around the triangle, which prevents once again the Gauss constraint from being applicable.

Interestingly, this leads to ``curvature-induced torsion" and  shows up in the coarse graining of spin networks \cite{eteradeformed} and spin foams \cite{asger}. It is also important in the context of quantum groups describing constant curvature geometries \cite{maite1,maite2,aldoetal}, as it shows that the Gauss constraints have to be deformed in a specific way in order to hold.

In order to illustrate the way in which curvature induces torsion, let us consider the example in figure \ref{Fex3}. There, we consider a closed co-path $\pi$ formed by the three edges of a triangle which is subdivided into three smaller triangles. One could now expect that the Gauss constraints for the three smaller triangles imply that $\bX_\pi$ is vanishing. Indeed, we expect $\bX_\pi$ to give the distance vector from the source vertex to the target vertex of $\pi$, which happen to coincide. However, as in the previous example, a parallel transport is involved in the precise definition of $\bX_\pi$, and instead of a vanishing flux one finds
\be
\bX_\pi=g_{rl_1(0)}^{-1}h_1^{-1}\left(h_4X_{l_4}h^{-1}_4-g^{-1}_vh_4X_{l_4}h_4^{-1}g_v\right)h_1g_{rl_1(0)},
\ee
where $g_v=h_4h_6h_5$ is the holonomy around the vertex subdividing the triangle.

\begin{center}
\begin{figure}[h]
\includegraphics[scale=0.7]{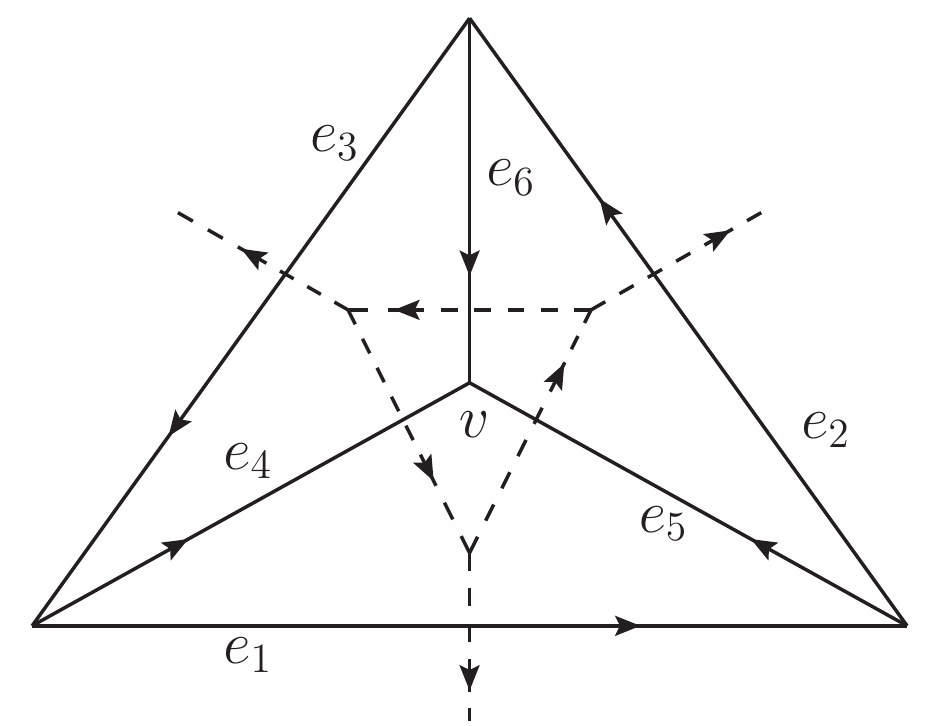}
\caption{Example of a triangle subdivided into three smaller triangles, and for which, if there is curvature around the vertex $v$, the integrated flux observable defined on $\pi=e_3\circ e_2\circ e_1$ is not vanishing.}
\label{Fex3}
\end{figure}
\end{center}

To summarize, in $d=2$ spatial dimensions the geometric quantities are associated to the vertices of the triangulations (we could indeed just keep the (embedded) vertices in our construction). These carry the (eventually) distributional curvature. A set of phase space point-separating variables is given by the holonomies around every vertex (parallel transported to the root) and a set of integrated fluxes with underlying co-path going from a vertex adjacent to the root (triangle) to all the other vertices.

Although the fluxes are a priori associated to co-paths made out of triangulation edges, it is rather the source and target vertices of these co-paths, and the way in which the associated parallel transport is defined with respect to the other vertices, that determines the flux observable. Without curvature defects, the flux observable would indeed give simply the distance vector (with $G=\SU(2)$) from the source to the target vertex of $\pi$ in the embedding $(2+1)$-dimensional flat geometry. In this case the flux observables would be independent of the chosen co-path and only depend on the source and target vertices.

\subsection{Geometric interpretation in $\boldsymbol{d=3}$ spatial dimensions}

\noindent The case $d=3$ is quite analogous to the case $d=2$ treated above, at the difference that we need to replace vertices in $d=2$ with edges in $d=3$. We have however defined an embedding based on vertices (edges are again determined as geodesics with respect to an auxiliary metric on the underlying manifold), since eventually we wish to understand diffeomorphisms as vertex displacements \cite{williams,dittrichryan,bahrdittrich09a,bonzomdittrich}. The translational symmetry of BF theory in $(2+1)$ dimensions can indeed be interpreted (on-shell) in this way \cite{louaprediff}. However, the translational symmetry of BF theory in $(3+1)$ dimensions leads to edge displacements instead of vertex displacements. In order to obtain gravity from the topological BF theory, one would therefore like to break this symmetry via the imposition of simplicity constraints \cite{zapata,dittrichryan}.

This does not exclude the possibility of changing our initial definitions, and of working with triangulations (or other polyhedral complexes) where the edges are embedded. An interesting question is whether and how these different choices lead to different continuum limits.

In $d=3$ spatial dimensions, curvature defects are associated to the edges of the triangulation. Without curvature, the integrated flux observables would only depend on the boundary of the surfaces (made out of edges) associated to the fluxes. This does however change if the choice of parallel transport matters. In this case a parallel transport crossing an edge of the triangulation can lead to curvature terms that prevent the application of the Gauss constraints.

For the $\SU(2)$ case, the fluxes can be interpreted as normals to the triangles (weighted by the triangle areas). Considering a tetrahedron, one can reconstruct uniquely out of these normals the geometry of this tetrahedron (the Gauss constraints hold since we are working on the gauge-invariant phase space) \cite{minkowski,polyhedra}.

However, it is not guaranteed that neighboring tetrahedra can be glued consistently, since the shapes of the triangles reconstructed from the normals do not need to match. This feature has been identified and discussed in \cite{dittrichryan} and led to the term ``twisted geometries'' \cite{twisted,polyhedra}. Gluing (or shape-matching) constraints \cite{areaangle} can be imposed, and lead to a phase space describing proper (Regge) geometries. It has been argued that these constraints arise as secondary simplicity constraints \cite{dittrichryan,simplicity2,anza}, and they might therefore be required in order to realize diffeomorphisms as vertex displacements. However, these additional constraints are second class, which makes their imposition at the quantum level difficult, and so far it is not known whether this imposition would still allow to use the powerful techniques associated to the $\SU(2)$ phase space.

As in $d=2$ spatial dimensions, we will also have the effect of curvature-induced torsion. We can define the integrated flux associated to the surface of a tetrahedron, and then subdivide this tetrahedron into four smaller tetrahedra. Then the curvature associated to the edges in the subdivided tetrahedron can again lead to a non-vanishing flux observable associated to the boundary of this subdivided tetrahedron.

\section{Connecting the discrete phase spaces}
\label{connecting}

\noindent So far we have introduced the (gauge-invariant) phase space $\mathcal{M}_\Delta$ associated to a given triangulation $\Delta$, and in section \ref{sec:setup} a partial order on the set of triangulations. One can therefore define in principle either the inductive or projective limit of structures (i.e. of phase spaces) labelled by the elements of this partial order. For this construction, the definition of (consistent) embedding or projection maps respectively is essential. These maps ``stitch'' the discrete Hilbert spaces or phase spaces together, and are required to satisfy certain consistency conditions. The inductive or projective limit then defines the corresponding continuum structure. We review these constructions in the appendix \ref{appendix:limits}. We are going to use an inductive limit for the construction of the continuum Hilbert space, as is also done for the AL representation \cite{lqg3}.

For phase spaces, it is customary to use a projective limit, as was done for example for a phase space construction corresponding to (a variant of) the AL representation by Thiemann in \cite{qsd7} (see also \cite{lanery1}). However, the construction of \cite{qsd7} relies on a family of regular discretizations (cubical lattices). Indeed, the projective maps defined in \cite{qsd7} would not be consistent if the partially ordered set was to include more general graphs, and if the refinement operations were to allow for an inversion of the edges. This foreshadows a difficulty that we will also meet with the BF representation. In fact, since the definition of the simplicial fluxes involves parallel transports with the gauge connection and in $d=3$ a choice of surface tree, the construction of consistent projection maps defined on the full phase spaces will turn out to be problematic (at least if we consider general triangulations). Another difficulty is that, as discussed in section \ref{sec:geom}, curvature induces torsion for the coarser fluxes. As we will see, an inductive limit cannot be defined either, as it is equally difficult to construct consistent embedding maps (as proper maps).

One could attempt to alleviate this situation by changing the set of labels, for instance turning to flagged triangulations, and adjusting the refinement operations. However, we do want to keep the setup as simple as possible, and also to stay as near as possible to the spirit of the quantum theory.

Interestingly, these difficulties at the phase space level do not appear for the construction of the continuum Hilbert space via an inductive limit. The quantum embedding maps will be well defined and consistent. The reason for this is that the quantum maps need in some sense less information than the classical maps. To be more precise, the quantum embeddings $\iota_{\Delta,\Delta'}:\mathcal{H}_\Delta\rightarrow\mathcal{H}_{\Delta'}$ map states on a coarser triangulation to states in a finer triangulation and since this finer triangulation supports in general more degrees of freedom than the coarser triangulation, we have to specify the quantum state for these additional degrees of freedom. In the case of the BF-based representation discussed here, this quantum state is defined by demanding that the curvature be vanishing, i.e. that the holonomies be trivial for all the additional (finer) cycles in the graph dual to $\Delta'$. In the same way, one can in fact define embedding maps for the AL representation, and require that the flux variables be vanishing for all the new edges in the finer graph.

Therefore, one can characterize in both cases the image of the embedding maps by constraint operators. As we will comment on later in more detail, this was also used in \cite{FGZ} in order to characterize the classical phase spaces underlying LQG. The constraints encode the vanishing of the curvature for the finer holonomies for the BF representation, and the vanishing of fluxes for the additional edges for the AL representation. Since these constraints are first class, their classical equivalents generate (gauge) transformations along the constraint hypersurfaces. Therefore, we see that the classical equivalents of the embedding maps are actually not proper maps. Instead, they map phase space points in the coarser phase space $\mathcal{M}_\Delta$ to gauge orbits in the constraint hypersurface of the finer phase space $\mathcal{M}_{\Delta'}$. This is what prevents us from using these improper maps to define an inductive limit for the phase spaces. On the other hand, one can still attempt to construct projection maps as inverses of these embeddings. However, these projections can only be defined on the constraint hypersurface (which is here given by the phase space points where the curvature of finer cycles is vanishing) in a consistent manner. The reason is that due to the vanishing curvature the choice of paths for the parallel transport of fluxes or for the holonomies does not matter.

In summary, by remaining as close as possible to the quantum theory, we see that we face the problem of having improper maps for the embeddings, and restricted projections defined only on constraint hypersurfaces. Later on, we will therefore propose a modified projective limit, taking into account that the projections can only be defined on constraint hypersurfaces.

This situation is however natural to expect if one attempts to define refining maps that preserve the symplectic structure \cite{timeevol}. In fact, post- or pre-constraints for embeddings and projections respectively always appear if one considers the canonical (time) evolution between phase spaces of different dimensions, which is defined via a generating function. A framework to deal with such an evolution between phase spaces of different dimensions has been developed in \cite{hoehn1,hoehn2}, which defines a canonical time evolution in simplicial discretizations where the triangulation in general changes from one (discrete) time step to the next. We refer the reader to \cite{qhoehn1,qhoehn2} for an implementation of such a discrete time evolution into the quantum theory, where the post- and pre-constraints also play a central role.

The post- or pre-constraints necessarily arise in order to allow for a canonical transformation between phase spaces of different dimensions. If one uses a generating function for this transformation, both set of constraints $\{\mathcal{C}^I\}_\text{pre}$ and $\{\mathcal{C}^J\}_\text{post}$ are first class (among themselves). Therefore, one has to consider instead the reduced phase spaces $\mathcal{M}_\text{pre}\sslash\{\mathcal{C}^I\}_\text{pre}$ and $\mathcal{M}_\text{post}\sslash\{\mathcal{C}^J\}_\text{post}$, which turn out to be of equal dimension. This explains how a transformation between phase spaces of a priori different dimensions can be canonical. In fact, we could also define the BF refining maps via generating functions. These generating functions should be given by the BF action associated to the building blocks that are glued to the triangulation in the various Alexander moves. This will be clearer in the quantum theory, where the quantum embedding maps will be given by a quantization of such maps. We refrain from using these maps explicitly since the rigorous construction of the associated discrete classical action (at the gauge-variant level) is much more complicated than the action of the maps which we will describe below (or the actual quantization of these maps).

For the BF vacuum, the post-constraints which appear with a refining move impose the flatness of the part of the connection that is being added to the connection on the coarse-grained triangulation. The values of the finer fluxes are not completely fixed, and instead we have ``gauge'' orbits determining that the composition of finer fluxes gives the coarse flux in $\mathcal{M}_\Delta$. Because all finer connection data are flat, the specification of how we exactly compose the finer fluxes, i.e. the choice of surface tree, does not actually matter.

On the other hand, if we consider a map from a finer to a coarser triangulation, i.e. a projection, we will have pre-constraints. These pre-constraints do again impose the flatness of the part of the connection which is eliminated when going from the finer to the coarser triangulation. The projection map can be understood as the inverse of the embedding map. Therefore the pre-constraints are first class, and each gauge orbit generated by these constraints is mapped to one point in the coarser phase space. The projections are proper maps (as opposed to the embeddings), which however can a priori only be defined on the pre-constraint hypersurface.

In the rest of this section we are going to describe the embeddings and projections in more detail and also compare to the analogous entities for the AL representation. We will propose a modification of the projective limit construction, that takes into account that the projections can only be defined on constraint hypersurfaces of a given phase space. In fact if a phase space point is off the constraint hypersurface of a given projection, it means that this phase space point describes curvature that cannot be accommodated in the phase space one is projecting to. Thus we will not demand  such a projection in our modified projective limit construction.

We will also show that a coarser phase space arises from a finer phase space by symplectic reduction with respect to the constraints. Such a symplectic reduction has been used first in \cite{FGZ} to construct the discrete phase spaces out of the continuum phase space. In our setting, both the initial and final phase spaces will however be discrete. Also, we will rather prove that the projection maps define a symplectic reduction, as opposed to using the symplectic reduction to define coarse phase spaces from (infinitely) finer ones.

Finally, in this section we will introduce a continuum observable algebra, which will basically correspond to the set of phase space functions on the modified projective limit. Such phase space functions will be represented by a (consistent) family $\{\mathcal{O}_\Delta\}$ of observables. This actually captures the fact that, in the quantum theory, we mostly deal with the Hilbert spaces associated to a given discretization, and hence also with the operators associated to this discretization. On a given triangulation, we need only a certain amount of information about the observables, i.e. the surface tree up to a certain fineness scale given by the triangulation. Moreover, since coarser phase spaces arise by symplectic reduction from finer phase spaces, one can conclude that the (spectral) properties of an observable $\mathcal{O}_{\Delta}$ on a given discrete phase space coincide with the (spectral) properties of the observable $\mathcal{O}_{\Delta'}$ from the same family, if considered on the subspace of states resulting from an embedding from the coarser triangulation $\Delta$.

\subsection{The embedding and projection maps}
\label{embpro}

\noindent Here we describe in more detail the embedding and projection maps for the BF representation (we use the term ``maps'' in a generalized sense, since we allow for multiple images in the form of post-gauge orbits). As discussed above, this will basically require the imposition of flatness conditions for finer connection degrees of freedom, and a corresponding conjugated gauge orbit for the flux degrees of freedom.

\subsubsection{The embedding maps}

\noindent Let us start by considering the embedding maps $\mathcal{E}_{\Delta,\Delta'}$, where $\Delta$ is a triangulation coarser than $\Delta'$. The connection information at the level of the gauge-invariant phase space is encoded in holonomies associated to closed paths $\gamma$ starting and ending at the root $r$. For a given triangulation, we can choose the set of independent curves such that each curve is generic, i.e. does only cross $(d-1)$-dimensional simplices and no lower-dimensional ones.

Let us consider a given set of holonomies $\{g_\gamma\}_\gamma$ determining the connection part of a point $p$ in the coarser phase space $\mathcal{M}_\Delta$. This is mapped to a set of holonomies $\{g_{\gamma'}\}_{\gamma'}$ describing (eventually) an orbit of phase space points $\mathcal{E}_{\Delta,\Delta'}(p)$ in $\mathcal{M}_{\Delta'}$, with
\be\label{defh}
g_\gamma=g_{\textbf{P}(\gamma')}.
\ee
Here $\textbf{P}(\gamma')$ denotes a path in the dual to the coarser triangulation $\Delta$, and therefore $g_{\textbf{P}(\gamma')}$ is either given by one of the holonomies of the set $\{g_\gamma\}_\gamma$ or can be reconstructed from this set. $\textbf{P}(\gamma')$ can be understood as a choice of projection of the path $\gamma'\subset \Gamma'$ to a path $\gamma\subset\Gamma$ in the dual to the coarser triangulation. Here $\gamma'$ starts and ends at the root $r'$ of the refined triangulation $\Delta'$, whereas $\gamma$ starts and ends at the root $r$ in $\Delta$. The roots $r$ and $r'$ are related in the way described in section \ref{sec:setup}.

The map $\textbf{P}$ needs to satisfy certain conditions. Since the triangulation $\Delta'$ is a refinement of the triangulation $\Delta$, we can identify simplices $\sigma_\Delta$ of $\Delta$ with complexes of simplices $\cup\sigma_{\Delta'}$ (of the same dimension as $\sigma_\Delta$) in the finer triangulation. In other words, $\cup \sigma_{\Delta'}$ gets coarse-grained to $\sigma_\Delta$. We will denote this relationship by $\textbf{P}^{-1}(\sigma_\Delta)=\cup\sigma_{\Delta'}$. Therefore, the curve $\textbf{P}(\gamma')$ enters and leaves a $d$-dimensional simplex $\sigma^d_\Delta$ in the same order and through the same neighboring simplices $\tilde\sigma^d_\Delta$ as $\gamma'$ enters and leaves the complex $\textbf{P}^{-1}(\sigma^d_\Delta)$ with respect to neighboring  complexes $\textbf{P}^{-1}(\tilde \sigma^d_\Delta)$.

Thus we can partition the simplices of the finer triangulation into sets corresponding to the simplices of the coarser triangulation, and consider the independent cycles of the dual graph to the coarser triangulation. There will be equivalence classes of path $\gamma'$ that map to the same independent cycle of $\Gamma$, and all such elements $\gamma'$ in a given equivalence class will be assigned the same value for the holonomy. The definition (\ref{defh}) also prescribes the behavior of the global reference system of the root under refinement, and the choice of reference system at the root $r$ is copied over to the refined root $r'$.

Clearly, there will be cycles in $\Gamma'$ that are mapped by $\textbf{P}$ to trivial cycles in $\Gamma$, and therefore the embedding assigns the identity element to such cycles. These conditions give rise to the post-constraints mentioned above, which here take the form
\be\label{fconstraints}
\mathcal{C}^I_{\Delta,\Delta'}=g_I-\openone\stackrel{!}{=}0.
\ee
Here the label $I$ enumerates the additional independent cycles that are present in $\Gamma'$ and not in $\Gamma$, and the post-constraints impose that these cycles carry trivial holonomies. These constraints commute and lead to a ``gauge'' action on the fluxes.
 
Indeed, as mentioned in the previous section, the embedding $\mathcal{E}_{\Delta,\Delta'}$ will not describe unique values for the fluxes. Instead, we will have ``refining gauge orbits'' along the constraint hypersurface. Given a set of independent fluxes (obtained by choosing a tree and considering the rooted fluxes associated to the leaves) $\{{\bX}_l\}_l$ for the coarser triangulation $\Delta$, the condition for the fluxes $\textbf{ X}_{l'}$ in the gauge orbit associated to $\{{\bX}_l\}_l$ is given by
\be\label{embf}
\bX_l=\bX_{l'_n}\circ\cdots\circ\bX_{l'_1},
\ee
where the right-hand side denotes composition of fluxes as defined in \eqref{compi2d} and \eqref{compi3d}. Here $\{l'_1,\ldots,l'_n\}$ is the set of links dual to the edges ($d=2$) or triangles ($d=3$) of $\Delta'$ that make up the edge or triangle dual to the link $l$. In other words, $\textbf{P}^{-1}(l^*)=(l'_1)^*\cup\cdots\cup(l'_n)^*$. Furthermore, if the parallel transport of ${\bX}_l$ to the root is done along the path $\gamma$, then the parallel transport to the root of the composed fluxes on the right hand side of \eqref{embf} has to be done with a path $\textbf{P}^{-1}(\gamma)$. This path is in the pre-image of $\gamma$ with respect to the map $\textbf{P}$, which we here extend to paths $\gamma'$ starting at the root and ending elsewhere. The path in $\textbf{P}^{-1}(\gamma)$ starts at the root in the refined dual, and ends at the root of the surface tree.

This surface tree (and its root) can be chosen arbitrarily, as long as the conditions laid out in section \ref{sf3d} are satisfied. Note that in \eqref{embf} we consider only elementary fluxes in the coarser triangulation, for which the (coarser) surface trees are trivial.

The choice of a particular surface tree does not influence the value of the composition of the fluxes. The reason for this is that we are restricted (by the embeddings) to the  post-constraint hypersurface, because the image of the embedding maps prescribes that all finer holonomies should be flat. The same holds for the choice of a particular path in the pre-image $\textbf{P}^{-1}(\gamma)$.

In section \ref{reduced}, we will show that the gauge orbit described by \eqref{embf} is indeed preserved by the flow of the constraints by showing that the Poisson brackets of the right-hand of \eqref{embf} with the constraints vanish on the constraint hypersurface.

So far we have specified embedding maps that map coarser phase spaces to (gauge orbits in) a finer phase space. Inverting these maps gives (restricted) projection maps $\mathcal{P}^\text{R}_{\Delta',\Delta}$. These are proper maps in the sense that ``gauge orbits'' in $\mathcal{M}_\Delta$ do not appear. However, these maps can only be defined on the (now) pre-constraint hypersurface, which coincides with the post-constraint hypersurface of the corresponding embedding map.

\subsubsection{The projection maps}

\noindent As we said, the restricted projection maps can be obtained by inverting the embedding maps, and are given explicitly by
\be
g_\gamma=g_{\textbf{P}^{-1}(\gamma)},\q\bX_l=\bX_{l'_n}\circ\cdots\circ\bX_{l'_1},
\ee
where again $\textbf{P}^{-1}(l^*)=(l'_1)^*\cup\cdots\cup(l'_n)^*$. The restriction to the constraint hypersurface ensures that the projection is well defined, despite the fact that one has to choose surface trees and paths in the pre-image $\mathbf{P}^{-1}(\gamma)$.

This restriction to the constraint hypersurfaces is also essential in order for the restricted projection maps (and equivalently for the embeddings) to be consistent, i.e. to satisfy
\be
\mathcal{P}^\text{R}_{\Delta',\Delta}\circ\mathcal{P}^\text{R}_{\Delta'',\Delta'}=\mathcal{P}^\text{R}_{\Delta'',\Delta}
\ee
for any triple $\Delta\prec\Delta'\prec\Delta''$ of increasingly finer triangulations. To make sense out of this equation, the composed map on the left-hand side has to be restricted to the constraint hypersurface described by $\mathcal{C}^I_{\Delta,\Delta''}$. Again, the reason for this is that we do not need to specify surface trees or a unique path in $\textbf{P}^{-1}(\gamma)$ since all the refined holonomies are flat. Thus, we just need to convince ourselves that the pre-images $\textbf{P}^{-1}$ are cylindrically-consistent, which is indeed the case, as shown in the next subsection.

\subsection{Consistency of the symplectic structure}

\noindent We are now going to show that the (restricted) projection maps preserve the symplectic structure. This also follows from the fact that we could have used a ``refining time evolution'' derived from the BF action in order to define the embeddings, or a ``coarse graining time evolution'' to define the (restricted) projections. Indeed, in this case the use of a generating function guarantees that the map is canonical \cite{hoehn1,hoehn2}.

Let us consider a function $f$ on the phase space $\mathcal{M}_\Delta$. This function can be pulled back with the restricted projection map to a function on the constraint hypersurface described by $\mathcal{C}^I_{\Delta,\Delta'}$ in the phase space $\mathcal{M}_{\Delta'}$ associated to the refined triangulation. Explicitly, this pullback is given by
\be
\big(\mathcal{P}^\text{R}_{\Delta',\Delta}\big)^*(f)\big[(g_{\ell'},\bX_{\ell'})_{\ell'}\big]\coloneqq f\big(\mathcal{P}^\text{R}_{\Delta',\Delta}\big[(g_{\ell'},\bX_{\ell'})_{\ell'}\big]\big),
\ee
and is only defined on the constraint hypersurface. On the other hand, it associates unique values to the refining gauge orbits, and thus the resulting functions on $\mathcal{M}_{\Delta'}$ are gauge-invariant with respect to the transformation generated by the pre-constraints. We are going to show this explicitly in the next subsection. For the moment, let us focus on the consistency of the symplectic structure, and on the following theorem.

\begin{Theorem}\label{PBtheorem}
The restricted projection maps preserve the Poisson brackets, i.e. the equations
\be\label{pbcond1}
\big(\mathcal{P}^\mathrm{R}_{\Delta',\Delta}\big)^*\lb f_1,f_2\rb_{\Delta}=\lb\big(\mathcal{P}^\mathrm{R}_{\Delta',\Delta}\big)^*(f_1),\big(\mathcal{P}^\mathrm{R}_{\Delta',\Delta}\big)^*(f_2)\rb_{\Delta'}
\ee 
hold on the constraint hypersurface for arbitrary (smooth) phase space functions $f_1$ and $f_2$ on $\mathcal{M}_\Delta$.
\end{Theorem}

At first sight, one might be worried about how to define the Poisson brackets between functions that are only defined on a submanifold. However, this submanifold is here defined by first class constraints, and therefore the symplectic form on $\mathcal{M}_{\Delta'}$ can simply be pulled back to the constraint hypersurface. Furthermore, by construction, the phase space $\mathcal{M}_\Delta$ and the reduced phase space $\mathcal{M}_{\Delta'}\sslash\mathcal{C}^I_{\Delta,\Delta'}$ are of equal dimension, and therefore the conditions \eqref{pbcond1} are sufficient in order to test the consistency of the Poisson brackets. In fact, using the techniques of \cite{hoehn1,hoehn2}, one can show that the pullback of the symplectic structure of $\mathcal{M}_\Delta$ with $\mathcal{P}^\text{R}_{\Delta',\Delta}$ coincides with pullback of the symplectic structure of $\mathcal{M}_{\Delta'}$ to the constraint hypersurface.

\begin{proof}[Proof of theorem \ref{PBtheorem}]
In order to prove \eqref{pbcond1}, it is sufficient to consider the pullbacks of some basic set of phase space functions. We choose these to be the holonomies and fluxes associated to the leaves $\ell$ with respect to some choice of tree in the dual graph of $\Delta$. These variables satisfy the Poisson bracket relations \eqref{SU(2)brackets}.

Let us start with a function $f=g_\ell$ corresponding to a basic holonomy. Then we have that
\be\label{pullbackhol}
\big(\mathcal{P}^\text{R}_{\Delta',\Delta}\big)^*(g_\ell)\stackrel{\text{\tiny{$\mathcal{C}$}}}{=}g_{\mathbf{P}^{-1}(\ell)},
\ee
where ${\mathbf{P}^{-1}(\ell)}$ should be understood as a (particular) path in the pre-image under $\textbf{P}$ of the cycle associated to the leaf $\ell$, i.e. including the transport from and to the root. Here, we denote with $\stackrel{\text{\tiny{$\mathcal{C}$}}}{=}$ an equation which holds on the constraint hypersurface with respect to the pre-constraints \eqref{fconstraints}. The phase space functions obtained via a pullback are a priori only defined on the constraint hypersurface. However, with a choice of a particular path in the pre-image of $\ell$, we can choose a particular extension of the phase space function away from the constraint hypersurface. Functions vanishing on the constraint hypersurfaces form an ideal with respect to multiplication, and we can define equivalence classes of (the multiplicative algebra of) functions, where two functions are equivalent if they only differ away from the constraint hypersurface.

Now, concerning the pullback of a basic rooted flux, we have that
\be\label{pullbackflux}
\big(\mathcal{P}^\text{R}_{\Delta',\Delta}\big)^*(\mathbf{X}_\ell)\stackrel{\text{\tiny{$\mathcal{C}$}}}{=}\bX_{l'_n}\circ\cdots\circ\bX_{l'_1},
\ee
where $\textbf{P}^{-1}(\ell^*)=(l'_1)^*\cup\cdots\cup(l'_n)^*$. Here again, a choice of surface tree and of parallel transport defines a particular extension of the left-hand side of this equation away from the constraint hypersurface.

It is now straightforward to check the consistency of the Poisson brackets. First of all, these relations are trivial for phase space functions involving only holonomies since holonomies commute. Let us therefore move on to the Poisson brackets between holonomies and fluxes. For this, we choose the phase space functions to be $f_1=\bX_\ell$ and $f_2=g_\ell$. By construction, the path $\textbf{P}^{-1}(\ell)$ will include one and only one of the finer links $l'_i$ appearing in the decomposition $\textbf{P}^{-1}(\ell^*)=(l'_1)^*\cup\cdots\cup(l'_n)^*$. Therefore, the only non-vanishing contribution to the Poisson bracket will be $\lb g_{\mathbf{P}^{-1}(\ell)},X_{l'_i}\rb$. In order to compute this contribution, let us rewrite $g_{\mathbf{P}^{-1}(\ell)}$ as $g_{\mathbf{P}^{-1}(\ell)}=g_{\gamma'_f}h_{l'_i}g_{\gamma'_s}$, where $g_{\gamma'_s}$ starts at the root $r'$ of the refined triangulation. We also denote by $g_{r'l'_i(0)}$ the transport from the refined root to the source vertex of $l'_1$ and then along the surface tree to the source vertex of $l'_i$. We get\footnote{See appendix \ref{appendix:1} for some helpful identities concerning $\SU(2)$ and its Lie algebra.}
\ba
\lb\big(\mathcal{P}^\mathrm{R}_{\Delta',\Delta}\big)^*\big(\bX_\ell^k\big),\big(\mathcal{P}^\mathrm{R}_{\Delta',\Delta}\big)^*(g_\ell)\rb_{\Delta'}
&=&\lb\big(\bX_{l'_n}\circ\cdots\circ\bX_{l'_1}\big)^k,g_{\mathbf{P}^{-1}(\ell)}\rb\nn\\
&=&g_{\gamma'_f}\lb\left(g^{-1}_{r'l'_i(0)}X_{l'_i}g_{r'l'_i(0)}\right)^k,h_{l'_i}\rb g_{\gamma'_s}\nn\\
&=&-2\tr\left(g^{-1}_{r'l'_i(0)}\tau^mg_{r'l'_i(0)}\tau^k\right)g_{\gamma'_f}\lb X^m_{l'_i},h_{l'_i}\rb g_{\gamma'_s}\nn\\
&\stackrel{\text{\tiny{$\mathcal{C}$}}}{=}&-2\tr\left(g^{-1}_{\gamma'_s}\tau^mg_{\gamma'_s}\tau^k\right)g_{\gamma'_f}h_{l'_i}\tau^mg_{\gamma'_s}\nn\\
&=&g_{\mathbf{P}^{-1}(\ell)}\tau^k,
\ea
where in the fourth equality we have used the fact that $g_{\gamma'_s}\stackrel{\text{\tiny{$\mathcal{C}$}}}{=}g_{r'l'_i(0)}$, since the parallel transport from the root to the leaf has to be performed for both the flux and the holonomy along the (coarse) tree. This calculation therefore confirms that \eqref{pbcond1} holds on the constraint hypersurface for Poisson brackets between pullbacks of fluxes and holonomies associated to the same leave.

Let us now consider a flux $f_1=\bX_{\ell_1}$ and an holonomy $f_2=g_{\ell_2}$ associated to two different leaves $\ell_1\neq\ell_2$. In this case, one can see that the path $\textbf{P}^{-1}(\ell_1)$ will not cross any of the simplices in $\textbf{P}^{-1}(\ell^*_2)$. Therefore, the cycle associated to $\ell_1$ only includes branches of the tree and the leaf $\ell_1$ itself. These branches and the leaf are dual to $(d-1)$-dimensional simplices $\sigma^{d-1}_{\Delta}$ that are fine-grained to certain sets $\textbf{P}^{-1}(\sigma^{d-1}_{\Delta})$, and the path $\textbf{P}^{-1}(\ell_1)$ can only intersect simplices in these sets. However, since these sets are completely disjoint from the set $\textbf{P}^{-1}(\ell^*_2)$, the pullbacks of the functions $f_1$ and $f_2$ commute, again confirming that \eqref{pbcond1} holds.

Finally, let us look at the Poisson bracket between two fluxes associated to the same leave $\ell$. Since the parallel transport needed in order to compose the finer edges or triangles into a coarse-grained edge or triangle does not cross any of the finer edges or triangles, the only non-vanishing contribution to the Poisson bracket will come from the commutation relation $\lb\bX^k_{l'_i},\bX^m_{l'_i}\rb=\eps^{km}_{~~~j}\bX^j_{l'_i}$ between the elementary fluxes. More precisely, introducing the notation $D^{jk}(g)=-2\tr\big(g^{-1}\tau^jg\tau^k\big)$ (see appendix \ref{appendix:1}), we have that
\ba\label{fluxflux}
\lb\big(\mathcal{P}^\mathrm{R}_{\Delta',\Delta}\big)^*\big(\bX_\ell^k\big),\big(\mathcal{P}^\mathrm{R}_{\Delta',\Delta}\big)^*\big(\bX_\ell^m\big)\rb_{\Delta'}
&=&\lb\big(\bX_{l'_n}\circ\cdots\circ\bX_{l'_1}\big)^k,\big(\bX_{l'_n}\circ\cdots\circ\bX_{l'_1}\big)^m\rb\nn\\
&=&\sum_{i=1}^nD^{pk}\left(g_{r'l'_i(0)}\right)D^{qm}\left(g_{r'l'_i(0)}\right)\lb X_{l'_i}^p,X^q_{l'_i}\rb\nn\\
&=& \sum_{i=1}^nD^{jp}\left(g^{-1}_{r'l'_i(0)}\right)\eps^{kmj}X^p_{l'_i}\nn\\
&=&\eps^{km}_{~~~j}\big(\bX_{l'_n}\circ\cdots\circ\bX_{l'_1}\big)^j,
\ea
which again confirms \eqref{pbcond1}. Here, we have assumed that both fluxes have the same choice of parallel transport $g_{r'l'_i(0)}$ from the root to the source of the link $l'_i$. However, even if we choose different parallel transports, these do still have to agree on the constraint hypersurface, which ensure that \eqref{fluxflux} holds at least there.

Furthermore, fluxes associated to different leaves, which commute in $\mathcal{M}_\Delta$, will commute also after the pullback. This can be seen by using the same reasoning we used for a flux and an holonomy associated to different leaves.

This completes the proof on the consistency of the projection maps with respect to the symplectic structure.
\end{proof}

In section \ref{fluxalgebra}, we are going to compute Poisson brackets between non-elementary (i.e. integrated) fluxes. These results will once again be consistent.

\subsection{Symplectomorphism between coarser and reduced finer phase spaces}
\label{reduced}

\noindent Although the existence of constraints prevents us from defining a projective limit phase space, they actually lead to an interesting feature and to a much stronger statement than the consistency of symplectic structures given by \eqref{pbcond1}. The statement is that for any pair $\Delta\prec\Delta'$ the finer phase space $\mathcal{M}_{\Delta'}$ reduced by the constraints $\mathcal{C}_{\Delta, \Delta'}^I$ is symplectomorphic to the coarser phase space $\mathcal{M}_\Delta$.

\begin{Theorem} Given any pair of triangulations such that $\Delta\prec\Delta'$, we have that
\be\label{symplectic}
\mathcal{M}_{\Delta'}\sslash\lb\mathcal{C}_{\Delta,\Delta'}^I\rb_I\simeq\mathcal{M}_\Delta,
\ee
where the symplectomorphic map is given by the restricted projection $\mathcal{P}^\mathrm{R}_{\Delta',\Delta}$.
\end{Theorem}

As mentioned above, this result would follow more directly if we had defined the embedding maps via a generating function (the BF action), which can be interpreted as leading to a refining BF time evolution. In this case, the machinery of \cite{hoehn1,hoehn2} applies, and can be used to show that the target phase space reduced by the post-constraints is symplectomorphic to the source phase space, with the symplectomorphism being described by the canonical map derived from the generating function. The reduction of the (continuum) phase space by BF constraints (almost everywhere) is also the key ingredient in \cite{FGZ} in order to define discrete phase spaces from the continuum one.

In order to prove \eqref{symplectic}, we have to show (in addition to property \eqref{pbcond1} which is already established) that the pullback $\big(\mathcal{P}^\text{R}_{\Delta',\Delta}\big)^*(f)$ of any (smooth) phase space function $f$ on $\mathcal{M}_\Delta$ commutes weakly, i.e. on the constraint hypersurface, with the constraints. In other words, we have to show that $\big(\mathcal{P}^\text{R}_{\Delta',\Delta}\big)^*(f)$ is a Dirac observable. Since this expression is only defined on the constraint hypersurface, we consider any (smooth) function $f'$ which coincides weakly with $\big(\mathcal{P}^\text{R}_{\Delta',\Delta}\big)^*(f)$. Therefore, with $\stackrel{\text{\tiny{$\mathcal{C}$}}}{=}$ denoting again an equation valid on the constraint hypersurface, we have to show the following lemma.

\begin{Lemma}\label{PBlemma}
For any smooth $f'\stackrel{\text{\tiny{$\mathcal{C}$}}}{=}\big(\mathcal{P}^\mathrm{R}_{\Delta',\Delta}\big)^*(f)$, we have that
\be\label{PBlemma equation}
\lb f',\mathcal{C}_{\Delta,\Delta'}^I\rb_{\Delta'}\stackrel{\text{\tiny{$\mathcal{C}$}}}{=}0
\ee
for each of the constraints $\mathcal{C}_{\Delta, \Delta'}^I$.
\end{Lemma}

\begin{proof}[Proof of lemma \ref{PBlemma}]
Let us start by choosing a tree $\mathcal{T}$ in the dual graph $\Gamma$ of the coarser triangulation $\Delta$. A point-separating set of phase space functions is given by the fluxes $\bX_\ell$ and holonomies $g_\ell$ associated to the leaves. Let us denote by $\bX_\ell'$ and $g_\ell'$ some choice of extension of the pullbacks of the corresponding phase space variables, as given in \eqref{pullbackhol} for the holonomies and in \eqref{pullbackflux} for the fluxes. Here, this choice of extension is made by making certain choices in the construction of the right-hand sides of \eqref{pullbackhol} and \eqref{pullbackflux}, for example for the precise parallel transport in the finer triangulation.

Now, according to a theorem in \cite{henneauxteitelboim}, any (smooth) phase space function $\phi$ on $\mathcal{M}_{\Delta'}$ that vanishes on the constraint hypersurface can be written as
\be
\phi=\sum_I\varphi_I\mathcal{C}_{\Delta,\Delta'}^I,
\ee
for some functions $\varphi_I$. Thus, any smooth function $f'\stackrel{\text{\tiny{$\mathcal{C}$}}}{=}\big(\mathcal{P}^\text{R}_{\Delta',\Delta}\big)^*(f)$ can be written as
\be \label{pro2}
f'=f_0\big(\bX_\ell',g_\ell'\big)+\sum_I\varphi_I\mathcal{C}_{\Delta,\Delta'}^I.
\ee
Since the constraints $\mathcal{C}_{\Delta,\Delta'}^I$ are functions of the holonomies only (see \eqref{fconstraints}), they commute (at least weakly) with the second term in \eqref{pro2}, as well as with the holonomies appearing in $f_0$. Therefore, in order to show \eqref{PBlemma equation} using \eqref{pro2}, we just need to show that we have the following vanishing Poisson bracket:
\be\label{pro3}
\lb\bX_\ell',\mathcal{C}_{\Delta, \Delta'}^I\rb\stackrel{\text{\tiny{$\mathcal{C}$}}}{=}0.
\ee
To this end, let us denote once again by $\textbf{P}^{-1}(\ell^*)=(l'_1)^*\cup\cdots\cup(l'_n)^*$ the union of the $(d-1)$-dimensional simplices corresponding to the refinement of the simplex $\ell^*$ dual to the leave $\ell$, as in \eqref{pullbackflux}. Furthermore, let us use the fact that the constraints $\mathcal{C}_{\Delta, \Delta'}^I$ form a basis for the more general set of (over-complete) constraints
\be
\mathcal{C}_{\gamma'}=g_{\gamma'}-\openone,
\ee
where the curves $\gamma'$ are such that $\mathbf{P}(\gamma')$ is the trivial curve. If $\gamma'$ does not cross any of the simplices $(l'_i)^*$, the flux observable will commute with $g_{\gamma'}$. Now, because the image of $\gamma'$ under $\textbf{P}$ has to be trivial, for any crossing between the curve $\gamma'$ and one of the simplices $(l'_i)^*$ there has to be an ``inverse'' crossing of one of the simplices $(l'_j)^*$ (where $j=i$ is allowed). Without loss of generality, we can reduce to the case where there are two crossings and the holonomy observable splits into $g_{\gamma'}=g_{\gamma'_f}h_j^{-1}g_{\gamma'_m}h_ig_{\gamma'_s}$ with $j\geq i$. Here, the holonomies $g_{\gamma'_f}$, $g_{\gamma'_m}$, and $g_{\gamma'_s}$ do not cross any of the simplices $(l'_k)^*$, and $h_i$ and $h_j$ are the group variables associated to the links $l'_i$ and $l'_j$ respectively.

With this, we can then compute the non-vanishing contributions to the Poisson bracket $\lb\bX_\ell',\mathcal{C}_{\gamma'}\rb$. These are given by
\ba\label{bracket XC}
&&\lb\big(\bX_{l'_n}\circ\cdots\circ\bX_{l'_1}\big)^k,g_{\gamma'_f}h_j^{-1}g_{\gamma'_m}h_{_i}g_{\gamma'_s}\rb\nn\\
&=&D^{mk}\left(g_{r'l'_j(0)}\right)g_{\gamma'_f}\lb X^m_{l'_j},h_j^{-1}\rb g_{\gamma'_m}h_ig_{\gamma'_s}+D^{mk}\left(g_{r'l'_i(0)}\right)g_{\gamma'_f}h_j^{-1}g_{\gamma'_m}\lb X^m_{l'_i},h_i\rb g_{\gamma'_s}\nn\\
&=&-D^{mk}\left(g_{r'l'_j(0)}\right)g_{\gamma'_f}\tau^mh_j^{-1}g_{\gamma'_m}h_ig_{\gamma'_s}+D^{mk}\left(g_{r'l'_i(0)}\right)g_{\gamma'_f}h_j^{-1}g_{\gamma'_m}h_i\tau^mg_{\gamma'_s}\nn\\
&=&g_{\gamma'_f} \left(-g_{r'l'_j(0)}\tau^kg_{r'l'_j(0)}^{-1}h_j^{-1}g_{\gamma'_m}h_i+h_j^{-1}g_{\gamma'_m}h_ig_{r'l'_i(0)}\tau^kg_{r'l'_i(0)}^{-1}\right)g_{\gamma'_s},
\ea
where we have denoted by $g_{r'l'_i(0)}$ and $g_{r'l'_j(0)}$ the holonomies appearing in the definition of the composed flux observable $\bX_\ell'=\ldots+g_{r'l'_i(0)}^{-1}X_{l'_i}g_{r'l'_i(0)}+g_{r'l'_j(0)}^{-1}X_{l'_j}g_{r'l'_j(0)}$ Now, we can use that $g_{r'l'_j(0)}=g_{l'_i(0)l'_j(0)}g_{r'l'_i(0)}$, where $g_{l'_i(0)l'_j(0)}$ is the holonomy along an open path in the surface tree going from the source node of $l'_i$ to the source node of $l'_j$.  With this decomposition we can rewrite the terms in bracket in the last line of  \eqref{bracket XC} as
\ba\label{pro6}
&&-g_{r'l'_j(0)}\tau^kg_{r'l'_j(0)}^{-1}h_j^{-1}g_{\gamma'_m}h_i+h_j^{-1}g_{\gamma'_m}h_ig_{r'l'_i(0)}\tau^kg_{r'l'_i(0)}^{-1}\nn\\
&=&g_{l'_i(0)l'_j(0)}\left(-g_{r'l'_i(0)} \tau^k g^{-1}_{r'l'_i(0)}g^{-1}_{l'_i(0)l'_j(0)}h_j^{-1}g_{\gamma'_m}h_i+g^{-1}_{l'_i(0)l'_j(0)}h_j^{-1}g_{\gamma'_m}h_ig_{r'l'_i(0)}\tau^kg_{r'l'_i(0)}^{-1}\right) \nn\\
&=&g_{l'_i(0)l'_j(0)}\left[g^{-1}_{l'_i(0)l'_j(0)}h_j^{-1}g_{\gamma'_m}h_i\ ,\ g_{r'l'_i(0)}\tau^kg^{-1}_{r'l'_i(0)}\right].
\ea
The path associated to the holonomy in the left entry of the commutator in \eqref{pro6} is mapped by $\textbf{P}$ to a trivial loop, and therefore the holonomy evaluates to the identity on the constraint hypersurface. This shows \eqref{pro3}, and finally achieves to prove the lemma.
\end{proof}

In summary, we have shown that the pullbacks of coarser phase space functions are Dirac observables with respect to the constraints. A counting of the phase space dimensions also shows that these pullbacks give a complete set of Dirac observables. We can divide out from the set of Dirac observables the ideal (with respect to multiplications) of functions vanishing on the constraint hypersurfaces. Two observables are equivalent if they coincide on the constraint hypersurface. The resulting (Poisson) algebra is then homeomorphic to the Poisson algebra of functions on the coarser phase space $\mathcal{M}_\Delta$. In particular, the Poisson algebra of observables is (weakly) closed, in the sense that for any two observables $f'$ and $g'$ of the form \eqref{pro2} we have
\be
\lb f',g'\rb_{\Delta'}=\big(\{f,g\}_\Delta\big)'+\sum_I\varphi_I\mathcal{C}^I_{\Delta,\Delta'}.
\ee

\subsection{Modified projective limit}
\label{projlimit}

\noindent One method for defining a continuum phase space starting from a directed partially ordered set (poset hereafter) of discrete ones is to use a projective limit. Given a consistent system of projection maps $\mathcal{P}_{\Delta',\Delta}$, the elements of such a projective limit $\mathcal{M}_{\infty}$ are defined as elements $p_\Delta$ of the direct product $\prod_\Delta\mathcal{M}_\Delta$ satisfying $p_\Delta=\mathcal{P}_{\Delta',\Delta}(p_{\Delta'})$ for each pair $\Delta\prec\Delta'$. In other words,
\be
\mathcal{M}_\infty=\lb p_\Delta\in{\prod}_\Delta\mathcal{M}_\Delta\ |\ p_\Delta=\mathcal{P}_{\Delta',\Delta}(p_{\Delta'}),\ \forall\ \Delta\prec\Delta'\rb. 
\ee

As mentioned in the beginning of this section, in the present setup we cannot apply this definition unaltered since we only have restricted projections. However, the fact that we have restricted projections means that certain phase space points (having fine-grained curvature) cannot be projected onto triangulations, since these cannot capture this fine-grained curvature. This suggests that one should alter the definition of a projective limit in order to account for the fact that the projections are restricted. Thus, instead of nets $(p_\Delta)$ that do have an element $p_\Delta$ for \textit{every} triangulation $\Delta$, we can consider nets $(p_\Delta)$ that do not have entries for triangulations $\Delta$ onto which $p$ cannot be projected due to the pre-constraints. This can also be described by assigning a symbolic value $\emptyset$ to the (restricted) projections acting on phase space points outside the pre-constraint hypersurface, and then applying the standard definition of the projective limit.  

\begin{Definition}[Modified projective limit]
We extend the restricted projections to projections on the full phase space by assigning the symbolic value $\emptyset$ to a projection acting away from the constraint hypersurface. We then define the modified projective limit as
\be
\mathcal{M}_\infty=\lb p_\Delta\in{\prod}_\Delta\mathcal{M}_\Delta\ |\ p_\Delta=\mathcal{P}_{\Delta',\Delta}(p_{\Delta'}),\ \forall\ \Delta\prec\Delta'\rb,
\ee
where we allow $p_\Delta$ to take the symbolic value $\emptyset$.
\end{Definition}

\subsection{The observable algebra}

\noindent We have seen so far that the phase spaces $\mathcal{M}_\Delta$ are connected by (restricted) projections which have embeddings as generalized inverse maps. These maps respect the symplectic structure of the phase spaces $\mathcal{M}_\Delta$, as formulated in \eqref{pbcond1}. Although it is not possible to define a continuum phase space via a standard projective limit, we can adopt the modified definition described right above in section \ref{projlimit}. Thus, we can consider observables as functions on this phase space. Let us describe these observables in terms of a family of observables defined on the discrete phase spaces. 

\begin{Definition}[Observables]\label{Observables}
An observable $\mathcal{O}=(\mathcal{O}_\Delta)$ is given by a net of (smooth) phase space functions $\mathcal{O}_\Delta\in\mathcal{C}^\infty(\mathcal{M}_\Delta)$. The poset $\mathcal{S}_\mathcal{O}$ of net labels $\Delta$ defining an observable $\mathcal{O}$ might be smaller than the poset of triangulations. We demand however that if a triangulation is an element of this poset, $\Delta\in\mathcal{S}_\mathcal{O}$, then this holds also for any refined triangulation $\Delta'\succ\Delta$. The elements of this net of phase space functions have to satisfy the condition
\be\label{extcond}
\big(\mathcal{P}^\mathrm{R}_{\Delta',\Delta}\big)^*(\mathcal{O}_\Delta)\stackrel{\text{\tiny{$\mathcal{C}$}}}{=}\mathcal{O}_{\Delta'},
\ee
for all pairs $\Delta\prec\Delta'$ with $\Delta\in\mathcal{S}_\mathcal{O}$, and where the equality refers once again to the constraint surface $\mathcal{C}^I_{\Delta,\Delta'}=0$. Finally, an observable family $\mathcal{O}'$ extends an observable family  $\mathcal{O}$ if $\mathcal{S}_\mathcal{O}\subset\mathcal{S}_{\mathcal{O}'}$ and $\mathcal{O}_\Delta=\mathcal{O}_{\Delta}'$ for all $\Delta\in\mathcal{S}_\mathcal{O}$.
\end{Definition}

This definition captures the idea that observables are naturally attached to a certain triangulation and to the associated notion of coarseness. For instance, consider a flux observable associated to a certain triangular surface (minimizing the area with respect to the auxiliary metric). This flux observable can be defined on phase spaces based on triangulations which are sufficiently fine to include this triangle as one of their simplices or as a union of their simplices. Let $\Delta_1$ be such a triangulation and $\mathcal{O}_{\Delta_1}$ the corresponding observable. $\mathcal{O}_{\Delta_1}$ can then be evaluated on $\mathcal{M}_{\Delta_1}$ and, by pullback with the projections, on all the finer phase spaces $\mathcal{M}_\Delta$ with $\Delta\succ\Delta_1$. However, on these finer phase spaces we have to restrict to the (pre-) constraint hypersurfaces with respect to the restricted projections $\mathcal{P}^\text{R}_{\Delta,\Delta_1}$. If we wish to evaluate (an extension of) $\mathcal{O}_{\Delta_1}$ away from the constraint hypersurfaces, we have in general to provide more information, e.g. about the refined surface tree or the parallel transport to the root in the refined triangulation. This takes into account the fact that different choices for objects in the pre-images of $\textbf{P}$ (used in section \ref{embpro}) may lead to observables that differ away from the constraint hypersurface. Therefore, $\mathcal{O}_{\Delta_1}$ has to be extended to $\mathcal{O}_{\Delta_2}$ with $\Delta_2\succ\Delta_1$, where the extension property is captured by the condition \eqref{extcond}. This in turn ensures that $\mathcal{O}_{\Delta_2}$ gives the same result as $\mathcal{O}_{\Delta_1}$ on sufficiently coarse phase space points (i.e. for those lying on the constraint hypersurface). Thus, $\mathcal{O}_{\Delta_1}$ can be defined on the full phase space associated to an arbitrary fine triangulation, provided that one extends this observable to this finer triangulation.

Let us make an additional remark concerning the refinement of flux observables in light of the ``curvature-induced torsion'' effect discussed in section \ref{sec:geom}. There, we found that the integrated fluxes depend on the choice of co-path if and only if curvature is present. Without curvature, the integrated fluxes depend only on the boundary of the co-path since fluxes over closed surfaces vanish in this case. The post-constraints impose the vanishing of curvature and it is therefore possible to define refined flux observables in which the exact position of the coÐpath in the finer triangulation is changed from the position in the coarser triangulation, but condition \eqref{extcond} is still satisfied. Of course, there is always an extension in which the co-path for the refined flux observable does coincide with the (image of the) co-path of the coarser flux observable.

Having observables labelled by a triangulation (or a coarseness scale) $\Delta$, we can restrict the quantization of these observables on the inductive limit Hilbert space to a subspace spanned by states that are cylindrical over $\Delta$. These are states arising from embeddings of states defined on the discrete Hilbert space $\mathcal{H}_\Delta$. This possibility is due to the theorem of section \ref{reduced}, which states that any $\mathcal{O}_{\Delta'}=(\mathcal{O}_\Delta)'$ arising as an extension of a certain $\mathcal{O}_{\Delta}$ is a Dirac observable with respect to the set $\lb\mathcal{C}^I_{\Delta,\Delta'}\rb_I$ of constraints describing the embedding of a coarser phase space $\mathcal{M}_\Delta$ into the finer one $\mathcal{M}_{\Delta'}$. 

Thus, given a Hilbert space representation $\mathcal{H}_{\Delta'}$ for the algebra of observables $\mathcal{O}_{\Delta'}$, we can attempt a Dirac quantization, i.e. impose the constraints $\lb\mathcal{C}^I_{\Delta,\Delta'}\rb_I$, and in this way find a representation of equivalence classes of observables $(\mathcal{O}_\Delta)'$. The constraints are satisfied for states in $\mathcal{H}_{\Delta'}$ that arise as embeddings of states in $\mathcal{H}_\Delta$, i.e. for states that are cylindrical over $\Delta$. These states are therefore ``physical'' states, and we can attempt to define a ``physical'' Hilbert space. Here we will (in future work) have to face the difficulty that with a standard choice of inner product, i.e. if the Hilbert spaces are of the form $L^2\big(\SU(2)^N,\de\mu_{\mathrm{Haar}}\big)$, the spectra of the constraints are continuous. Thus, as is well known, the ``physical'' states will not be normalizable with respect to this ``kinematical'' inner product. In fact, as discussed in \cite{paper1}, we expect that the inductive limit Hilbert space will lead to an inner product based on the Bohr compactification of the dual of the gauge group. This would make the spectra of the constraints discrete (in the sense that eigenstates are normalizable), and allow to identify ``physical'' Hilbert spaces as proper subspaces of the ``kinematical'' ones.

In this case, the properties of the quantum observables $\widehat{\mathcal{O}}_\Delta$ would coincide with the properties of any extension $\widehat{(\mathcal{O}_\Delta)'}$ restricted to the ``physical'' subspace spanned by states that are cylindrical over $\Delta$.

\subsection{ Comparison with the Ashtekar--Lewandowski representation}

\noindent We would now like to compare the classical  BF embedding maps to the corresponding classical embedding maps for the AL representation, which we will briefly review in this section. For more details we refer the reader to \cite{qsd7}, in which the projective limit of phase spaces for the AL representation is defined (for cubic graphs).

The AL construction is based on graphs $\Gamma$ and on the associated phase spaces $\mathcal{M}_\Gamma$. The variables and the symplectic structure characterizing a phase space associated to a given graph coincide with that of the gauge-variant phase space described in section \ref{gauge-variantPS} (we therefore choose simplicial fluxes). An AL embedding goes from a coarser phase space $\mathcal{M}_\Gamma$ to a finer phase space $\mathcal{M}_{\Gamma'}$, where the phase spaces are associated to graphs. Here, a finer graph $\Gamma'$ can be reached from a coarser graph $\Gamma$ by three basic operations on the links of the graph $\Gamma$: (a) subdividing a link, (b) adding a link, and (c) inverting a link.

Let us consider the case in which a link $l$ in $\Gamma$ is subdivided into several links $l=l'_n\circ\cdots\circ l'_1$ in $\Gamma'$, together with the associated embedding map. The AL embedding imposes that the fluxes along the subdivided link be constant\footnote{The reader should not be confused by this embedding for the fluxes, which looks so different from the conditions \eqref{embf} in which finer fluxes are added to coarser fluxes. The geometrical situations are very different. Whereas for the BF embedding we consider a triangle or an edge dual to a link $l$, which is glued from finer triangles or edges dual to links $l'_i$, in the case of the AL embedding we consider a link $l$ composed of finer links $l'_i$. Thus, triangles dual to $l'_i$ are now rather ``stacked'' behind each other. Hence, the flux going through these triangles should be constant (see also the discussion in \cite{guedes}).}. Taking the parallel transport to the various source nodes of the new links into account, this means that the embedding $\mathcal{E}_{\Gamma,\Gamma'}$ maps a phase space point $p$ with a flux $X_l$ for the link $l$ into (eventually) a set of phase space points with fluxes
\be
X_{l'_i}\coloneqq h_{l'_{i-1}}\ldots h_{l'_1}X_lh_{l'_1}^{-1}\ldots h_{l'_{i-1}}^{-1},
\ee
for all the finer links $l'_i$ that form the link $l$. The fine-grained holonomies $h_{l'_i}$ appearing here will be specified below.

In the case of the operation (b), where a new link $l'$ is added, the AL embedding sets $X_{l'}=0$. For the operation (c), where a link is inverted to form $l'=l^{-1}$, the embedding sets $X_{l'}=-h_l X_l h_l^{-1}$. This leads to post-constraints, which for the subdividing operation (a) take the form
\be
X_{l'_j}=h_{l'_{j-1}}\ldots h_{l'_i}X_{l'_i}h_{l'_i}^{-1}\ldots h_{l'_{j-1}}^{-1}
\ee
for any pair $(l'_i,l'_j)\subset l$ with $j>i$. Furthermore, we have that $X_{l'}=0$ for any link $l'$ that is added with operation (b).
 
Therefore, defining projection maps as the inverse of embedding maps is only possible on these constraint hypersurfaces, which require constant flux\footnote{The work \cite{qsd7} manages to define projections on the full phase space. This is due to the restriction to a family of cubic graphs which furthermore does not include the inversion of edges as refining operation. Indeed, the projection map in \cite{qsd7} (which selects $X_{l'_1}$ as the flux associated to the coarser link $l$) would not be consistent if such inversions were allowed. One possibility to define consistent projections on general graphs, is to not allow inversions as a refinement operation, which would eventually make less states equivalent to each other in the inductive Hilbert space construction.}. As in the case of the BF embedding, we have to deal with restricted projections.

Because of the existence of these constraints, we expect that the embeddings for the connection degrees of freedom will lead to gauge orbits (actually, the flow of the constraints also affect the fluxes themselves). Indeed, for the operation (a) the gauge orbits are described by the following condition for the finer holonomies:
\be
h_l=h_{l'_n}\ldots h_{l'_1}.
\ee
In case (b), when adding an additional edge $l'$, the new holonomy variable $h_{l'}$ is arbitrary, and the gauge orbit is parametrized by $h_{l'}$. Finally, in case (c) we have that $h_{l^{-1}}=h_l^{-1}$.

The reader will notice that the role of the fluxes and the holonomies is reversed when going from the AL embedding to the BF embedding. Indeed, these two representations are based on vacua that are dual to each other in the statistical physics sense \cite{savit} (see also \cite{BBR}). The different vacua describe different phases of lattice gauge theory. AL corresponds to the strong coupling limit, while BF describes the lattice weak coupling limit. The wish to have also a representation for the BF phase, which is underlying the spin foam construction \cite{spinfoamreview1,spinfoamreview2} and also appears as a possible phase when coarse graining spin foams \cite{eckert,holonomy,sffinite,decorated}, is one major motivation for the present work (see also the discussion in \cite{cylconsis}).

\section{Completing the commutator algebra of fluxes}
\label{fluxalgebra}

\noindent An essential feature of the BF-based representation is that it incorporates naturally the coarse graining of fluxes. This was our main motivation for introducing the integrated fluxes as observables. With this at hand, we can now easily consider the Poisson algebra of these fluxes, even in the case of intersecting surfaces that cut each other along a curve made up of edges. Note that an analogous result for the AL representation is not known. There, one usually defines the fluxes without parallel transport, and the non-commutativity of the fluxes is argued to arise because of the requirement of a regularization for the computation of the Poisson brackets \cite{AshtekarCorichi} (the fluxes are observables smeared only over two-dimensional surfaces instead of three-dimensional volumes in $d=3$). This is also the reason for which the more singular cases, for example with surfaces or lines cutting each other, is not known.

Using for the fluxes a definition involving a parallel transport, the non-commutativity becomes more apparent, as noted in particular in \cite{qsd7} (see also \cite{husain}). Here, we complete the picture of the non-commutativity of the fluxes by also considering the more ``singular'' cases for the computation of the Poisson brackets. This might eventually help to construct an interpretation of the BF-based representation as arising from a $L^2$ Hilbert space $\mathcal{H}\simeq L^2(\bar{\mathcal{E}})$ over an extension $\bar{\mathcal{E}}$ of flux configurations, in a similar way in which the AL quantization leads to a Hilbert space $L^2(\bar{\mathcal{A}})$ over generalized connections (see the discussion in \cite{aristideandco,guedes}).

\subsection{Algebra of fluxes in $\boldsymbol{d=2}$ spatial dimensions}

\noindent Let us consider two integrated fluxes $\bX_{\pi_1}$ and $\bX_{\pi_2}$ defined on co-path $\pi_1$ and $\pi_2$. Because the integrated fluxes are constructed out of basic link holonomies $h_l$ and simplicial fluxes $X_l$, and given the elementary Poisson brackets \eqref{SU(2)brackets1} between these variables, the Poisson bracket between $\bX_{\pi_1}$ and $\bX_{\pi_2}$ will get a non-vanishing contribution for each link that appears in the expression of both integrated fluxes. This can happen in various situations, which we discuss separately below. In short, a non-vanishing Poisson bracket can come from the commutation of a simplicial flux with an holonomy involved in the transport to the root (case 1 below), from the commutation of two identical simplicial fluxes (case 2 below), or from the commutation of a simplicial flux with an holonomy involved in the transport of the fluxes to the flux reference frame (case 3 below).

As we will see, the Poisson bracket between any two arbitrary integrated fluxes can be decomposed in local contributions corresponding to these three cases, which we now describe in detail. For the sake of clarity, we just give here the result of the Poisson brackets, and the explicit computations are presented in appendix \ref{appendix:2}. Let us briefly recall what the notations introduced in section \ref{sec:intflux} are. For two integrated fluxes defined on co-paths $\pi_i$, with $i\in\{1,2\}$, we denote by $g_{rl^i_1(0)}$ the transport from the root to the source node of the link dual to the first edge of $\pi_i$, and by $g_{l^i_1(0)l(0)}$ the transport along the shadow graph of $\pi_i$ between $l^i_1(0)$ and some $l(0)$.

\subsubsection{Case 1}

\noindent There can be a non-vanishing contribution to the Poisson bracket between $\bX_{\pi_1}$ and $\bX_{\pi_2}$ even if the co-paths $\pi_1$ and $\pi_2$ are completely disjoint. The reason behind this is that the integrated fluxes $\bX_{\pi_1}$ and $\bX_{\pi_2}$ are transported to the root, and it can happen that the transport involved in (say) $\bX_{\pi_1}$ cuts the co-path $\pi_2$. In this case, there will be a non-vanishing contribution to the Poisson bracket between the integrated fluxes, coming from the elementary Poisson bracket between an holonomy $h_l$ in $\bX_{\pi_1}$ and a simplicial flux $X_l$ in $\bX_{\pi_2}$.

Such a situation happens if we consider (an example is represented on figure \ref{fig:Flux-bracket-2d-1}, but the following calculation is generic) the two integrated fluxes defined by
\be\label{fpb1}
\bX_{\pi_1}=g^{-1}_{rl^1_1(0)}\bX_{\pi_1\backslash r}g_{rl^1_1(0)},\q g_{rl^1_1(0)}=g_{l(1)l_1^1(0)}h_lg_{rl(0)},
\ee
and
\be
\bX_{\pi_2}=g^{-1}_{rl^2_1(0)}\left(\ldots+g^{-1}_{l^2_1(0)l(0)}X_lg_{l^2_1(0)l(0)}+\ldots\right)g_{rl^2_1(0)}.
\ee
Here $\bX_{\pi_1\backslash r}$ denotes the flux observable without the parallel transport to the root, and the superscripts carried by the links refer to the co-paths. In other words, $g_{rl^i_1(0)}$ denotes the parallel transport from the root to the source $l^i_1(0)$ of the link dual to the first edge of $\bX_{\pi_i}$. In \eqref{fpb1}, we have decomposed the transport to the root involved in $\bX_{\pi_1}$ so that the holonomy $h_l$ appears explicitly. We see that the simplicial flux $X_l$ appears in $\bX_{\pi_2}$, whereas $\bX_{\pi_1}$ contains a term of the form $h_l^{-1}Zh_l$, where $Z$ is some Lie algebra element. This leads to a non-vanishing Poisson bracket contribution of the form
\be\label{fpb2a}
\lb\big(h_l^{-1}Zh_l\big)^i,X_l^j\rb=\eps^{ijk}\big(h_l^{-1}Zh_l\big)^k.
\ee
Including all the additional parallel transports involved in the integrated fluxes, one therefore obtains for the Poisson bracket the formula
\be\label{fpb3}
\lb\bX_{\pi_1}^i,\bX_{\pi_2}^j\rb=\eps^{ikm}D^{kj}\left(g_{rl(0)}^{-1}g_{l^2_1(0)l(0)}g_{rl^2_1(0)}\right)\bX_{\pi_1}^m.
\ee
Note that the holonomy appearing in the matrix element $D^{kj}$ in this formula is based on a closed loop.

A variation of this situation arises if the parallel transport to the root cuts the flux $X_l$ in the other direction. In this case, we replace the parallel transport to the root in \eqref{fpb1} by
\be
g_{rl_1^1(0)}=g_{l(0)l^1_1(0)}h^{-1}_lg_{rl(1)}.
\ee
Then, instead of \eqref{fpb2a} one now gets
\be\label{fpb2b}
\lb\big(h_lZh^{-1}_l\big)^i,X_l^j\rb=-\eps^{ikm}D^{kj}(h_l)\big(h_lZh_l^{-1}\big)^m,
\ee
and the total Poisson bracket becomes
\be\label{fpb4}
\lb\bX_{\pi_1}^i,\bX_{\pi_2}^j\rb=-\eps^{ikm}D^{kj}\left(g_{rl(1)}^{-1}h_lg_{l^2_1(0)l(0)}g_{rl^2_1(0)}\right)\bX_{\pi_1}^m.
\ee
Again, the holonomy appearing in the matrix element $D^{kj}$ is based on a closed loop.

The situation in which the parallel transport of one flux to the root crosses a second flux several times can be also derived by adding the contributions from each crossing and using the two formulas given above.

\begin{center}
\begin{figure}[h]
\includegraphics[scale=0.7]{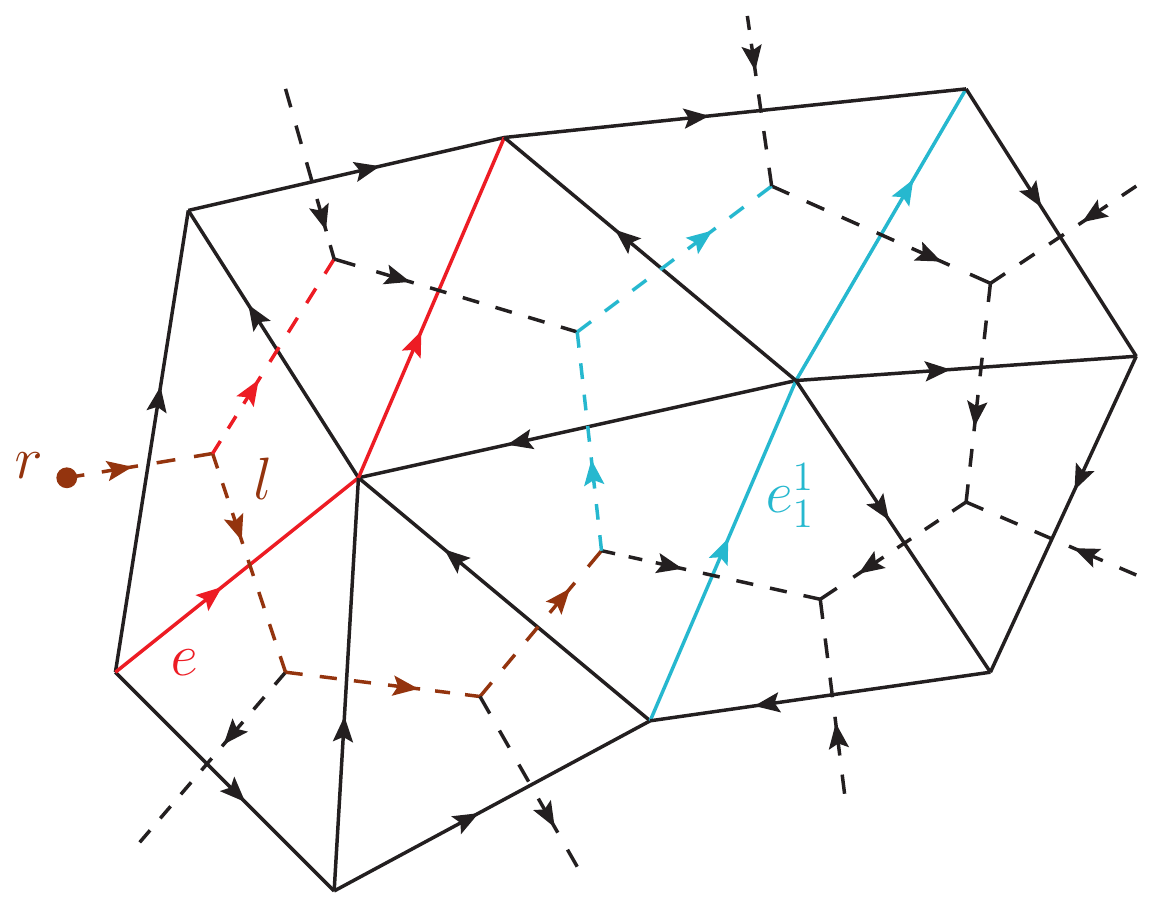}
\caption{Example of disjoint co-paths $\pi_1$ (solid blue) and $\pi_2$ (solid red) with their shadow graphs (dashed blue and red, respectively). One can see that the transport to the root along the tree (dashed brown) involved in $\bX_{\pi_1}$ is intersecting $\pi_2$.}
\label{fig:Flux-bracket-2d-1}
\end{figure}
\end{center}

\subsubsection{Case 2}

\noindent Let us put aside the case treated above, and assume that the transport to the root of the integrated fluxes $\bX_{\pi_1}$ and $\bX_{\pi_2}$ does not cross the co-paths $\pi_1$ or $\pi_2$. Now, a necessary condition for the Poisson bracket between the fluxes to be non-vanishing is that the co-paths $\pi_1$ and $\pi_2$ meet each other. Locally, it can happen that $\pi_1$ and $\pi_2$ share an edge of the triangulation or share a vertex.

We here consider the situation in which the two integrated fluxes share an edge $e$ dual to a link $l$ of the triangulation. In principle, this can also lead to the situation described above, where the parallel transport of one flux to the root cuts an edge of the other flux. Here we will however ignore such a contribution, since it can always be added at the end of the calculation to the contribution coming from the fact that the two integrated fluxes share an edge. Note that the two fluxes have to posses the same orientation in order to obtain a non-vanishing result since we have that $\lb X^i_l,X^j_{l^{-1}}\rb=0$.

One way to describe this situation is to consider the two integrated fluxes given by
\be
\bX_{\pi_1}=Z_1+g_{rl^1_1(0)}^{-1}g_{l^1_1(0)l(0)}^{-1}X_lg_{l^1_1(0)l(0)}g _{rl^1_1(0)},\q\bX_{\pi_2}=Z_2+g_{rl^2_1(0)}^{-1}g_{l^2_1(0)l(0)}^{-1}X_lg_{l^2_1(0)l(0)}g _{rl^2_1(0)},
\ee
where the Lie algebra elements $Z_1$ and $Z_2$ denote the remaining terms in the fluxes $\bX_{\pi_1}$ and $\bX_{\pi_2}$ respectively, and can be taken to be constants on the phase space for the purpose of this calculation. The non-vanishing contribution to the Poisson bracket comes from $\lb X_l^i,X_l^j\rb=\eps^{ijk}X^k_l$. Taking into account the parallel transports in the integrated fluxes, we obtain
\be
\lb\bX_{\pi_1}^i,\bX_{\pi_2}^j\rb=\eps^{ikm}D^{kj}\left(g_{rl^1_1(0)}^{-1}g_{l^1_1(0)l(0)}^{-1}g_{l^2_1(0)l(0)}g _{rl^2_1(0)}\right)\left(g_{rl^1_1(0)}^{-1}g_{l^1_1(0)l(0)}^{-1}X_lg_{l^1_1(0)l(0)}g _{rl^1_1(0)}\right)^m,
\ee
and one can see that the holonomy appearing in the matrix element $D^{kj}$ is based on a closed loop.

\subsubsection{Case 3}

\noindent Finally, the third case to discuss is the situation in which the two co-paths $\pi_1$ and $\pi_2$ only share a vertex. Let us remark that in this case the Poisson bracket between the two integrated fluxes is not necessarily non-vanishing. Indeed, if $\pi_1$ and $\pi_2$ have an opposite global orientation such that the shadow graph of $\pi_1$ never intersects $\pi_2$ (and the other way around), then the Poisson bracket will be vanishing (assuming that the two crossings treated above do not appear).

However, if $\pi_1$ and $\pi_2$ cut each other, then the shadow graph of $\pi_1$ will necessarily intersect $\pi_2$, the shadow graph of $\pi_2$ will intersect $\pi_1$, and the Poisson bracket between the integrated fluxes will get two non-vanishing contributions. Note that these contributions are a generalization of the first case treated above.

An example of the situation in which the shadow graphs and the co-paths of the integrated fluxes cut each other is represented in figure \ref{fig:Flux-bracket-2d-3}. Without loss of generality, we can describe this situation with the following integrated fluxes:
\begin{subequations}\label{3f1}
\ba
\bX_{\pi_1}&=&g^{-1}_{rl^1_1(0)}\left(Z_1+g_{l_1^1(0)l_a(0)}^{-1}X_{l_a}g_{l_1^1(0)l_a(0)}+g^{-1}_{l^1_1(0)l_b(1)}h_bY_bh_b^{-1}g_{l^1_1(0)l_b(1)}\right)g_{rl^1_1(0)},\\
\bX_{\pi_2}&=&g^{-1}_{rl^2_1(0)}\left(Z_2+g_{l^2_1(0)l_b(0) }^{-1}X_{l_b}g_{l^2_1(0)l_b(0)}+g_{l^2_1(0)l_a(0)}^{-1}h_a^{-1}Y_ah_ag_{l^2_1(0)l_a(0)}\right)g_{rl^2_1(0)}.
\ea
\end{subequations}
Here $Z_1$ and $Z_2$ denote the parts of the flux observables that will not give any contribution to the Poisson bracket. Also, $Y_a$ and $Y_b$ can be treated as constants for the sake of computing the Poisson brackets. $Y_a$ represents the part of the flux $\bX_{\pi_2}$ whose parallel transport involves $h_a$, and $Y_b$ represents the part of the flux $\bX_{\pi_1}$ whose parallel transport involves $h_b$. We will now have two contributions to the Poisson brackets. The first one, denoted by $T_1$, arises from the non-commutativity of $X_{l_b}$ with $h_b$, and the second one, denoted by $T_2$, arises from the non-commutativity of $X_{l_a}$ with $h_a$. We can therefore write that
\be
\lb\bX_{\pi_1}^i,\bX_{\pi_2}^j\rb=T_1^{ij}+T_2^{ij}.
\ee
Let us now compute these two contributions to the Poisson bracket separately.

In order to compute $T_1$, we can treat $X_{l_a}$ and $h_a$ as constants. This enables us to rewrite the flux observables \eqref{3f1} in the form
\bas
\bX_{\pi_1}&=&Z'_1+g_{rl^1_1(0)}^{-1}g_{l^1_1(0)l_b(1)}^{-1}h_bY_bh_b^{-1}g_{l^1_1(0)l_b(1)}g_{rl^1_1(0)},\\
\bX_{\pi_2}&=&Z'_2+g^{-1}_{rl^2_1(0)}g_{l^2_1(0)l_b(0)}^{-1}X_{l_b}g_{l^2_1(0)l_b(0)}g_{rl^2_1(0)},
\eas
where $Z'_1$ and $Z'_2$ are Lie algebra elements that can be treated as constants for the computation of $T_1$. With this rewriting of the fluxes, we get that the first contribution to the Poisson bracket is given by
\ba\label{T1}
T_1^{ij}&=&\lb\left(g_{rl^1_1(0)}^{-1}g_{l^1_1(0)l_b(1)}^{-1}h_bY_bh_b^{-1}g_{l^1_1(0)l_b(1)}g_{rl^1_1(0)}\right)^i,\left(g^{-1}_{rl^2_1(0)}g_{l^2_1(0)l_b(0)}^{-1}X_{l_b}g_{l^2_1(0)l_b(0)}g_{rl^2_1(0)}\right)^j\rb\nn\\
&=&-\eps^{ikm}D^{kj}\left(g_{rl^1_1(0)}^{-1}g_{l^1_1(0)l_b(1)}^{-1}h_bg_{l^2_1(0)l_b(0)}g_{rl^2_1(0)}\right)\left(g_{rl^1_1(0)}^{-1}g_{l^1_1(0)l_b(1)}^{-1}h_bY_bh_b^{-1}g_{l^1_1(0)l_b(1)}g_{rl^1_1(0)}\right)^m\nn\\
&=&-\eps^{ikm}D^{kj}\left(g_{rl^1_1(0)}^{-1}g_{l^1_1(0)l_b(1)}^{-1}h_bg_{l^2_1(0)l_b(0)}g_{rl^2_1(0)}\right)\big(\bX_{\pi_1}^{>\text{cr}}\big)^m.
\ea
Here we have used the fact that the Poisson bracket in the first line coincides with the second example of the first case treated above, which is computed in \eqref{fpb4}. In the last equality, we have introduced $\bX_{\pi_1}^{>\text{cr}}$, which is the portion of the flux that lies after the crossing.

To compute the contribution $T_2$ to the Poisson bracket, we can now rewrite the flux observables \eqref{3f1} in the form 
\ba
\bX_{\pi_1}&=&Z''_1+g^{-1}_{rl^1_1(0)} g^{-1}_{l^1_1(0)l_a(0)}X_{l_a}g_{l^1_1(0)l_a(0)}g_{rl^1_1(0)},\nn\\
\bX_{\pi_2}&=&Z''_2+g_{rl^2_1(0)}^{-1}g_{l^2_1(0)l_a(0)}^{-1}h^{-1}_aY_ah_ag_{l^2_1(0)l_a(0)}g_{rl^2_1(0)},
\ea
where $Z''_1$ and $Z''_2$ are Lie algebra elements that can be treated as constants for the computation of $T_2$. The Poisson bracket to compute therefore coincides with the first example of the first case treated above (modulo a reversed order for the entries), which is computed in \eqref{fpb3}. We therefore get
\ba\label{T2}
T_2^{ij}&=&\lb\left(g^{-1}_{rl^1_1(0)}g^{-1}_{l^1_1(0)l_a(0)}X_{l_a}g_{l^1_1(0)l_a(0)}g_{rl^1_1(0)}\right)^i,\left(g_{rl^2_1(0)}^{-1}g_{l^2_1(0)l_a(0)}^{-1}h^{-1}_aY_ah_ag_{l^2_1(0)l_a(0)}g_{rl^2_1(0)}\right)^j\rb\nn\\
&=&-\eps^{jkm}D^{ki}\left(g_{rl^2_1(0)}^{-1}g_{l^2_1(0)l_a(0)}^{-1}g_{l^1_1(0)l_a(0)}g_{rl^1_1(0)}\right)\left(g_{rl^2_1(0)}^{-1}g_{l^2_1(0)l_a(0)}^{-1}h^{-1}_aY_ah_ag_{l^2_1(0)l_a(0)}g_{rl^2_1(0)}\right)^m\nn\\
&=&\eps^{ikm}D^{kj}\left(g_{rl^1_1(0)}^{-1}g^{-1}_{l^1_1(0)l_a(0)}g_{l^2_1(0)l_a(0)}g_{rl^2_1(0)}\right)\nn\\
&&\phantom{\eps^{ikm}D^{kj}}\left(g_{rl^1_1(0)}^{-1}g^{-1}_{l^1_1(0)l_a(0)}g_{l^2_1(0)l_a(0)}g_{rl^2_1(0)}\bX_{\pi_2}^{>\text{cr}}g_{rl^2_1(0)}^{-1}g_{l^2_1(0)l_a(0)}^{-1}g_{l^1_1(0)l_a(0)}g_{rl^1_1(0)}\right)^m,
\ea
where in the last equality we have introduced the part of the flux $\bX_{\pi_2}$ that lies after the crossing.

Now we can combine this with the first contribution to find the total Poisson bracket. It is given by
\ba
&&\lb\bX_{\pi_1}^i,\bX_{\pi_2}^j\rb\nn\\
&=&\eps^{ikm}D^{kj}\left(g_{rl^1_1(0)}^{-1}g^{-1}_{l^1_1(0)l_a(0)}g_{l^2_1(0)l_a(0)}g_{rl^2_1(0)}\right)\nn\\
&&\phantom{\eps^{ikm}D^{kj}}\left(g_{rl^1_1(0)}^{-1}g^{-1}_{l^1_1(0)l_a(0)}g_{l^2_1(0)l_a(0)}g_{rl^2_1(0)}\bX_{\pi_2}^{>\text{cr}}g_{rl^2_1(0)}^{-1}g_{l^2_1(0)l_a(0)}^{-1}g_{l^1_1(0)l_a(0)}g_{rl^1_1(0)}\right)^m\nn\\
&&-\eps^{ikm}D^{kj}\left(g_{rl^1_1(0)}^{-1}g_{l^1_1(0)l_b(1)}^{-1}h_bg_{l^2_1(0)l_b(0)}g_{rl^2_1(0)}\right)\big(\bX_{\pi_1}^{>\text{cr}}\big)^m\nn\\
&=&\eps^{ikm}D^{kj}\left(g_{rl^1_1(0)}^{-1}g^{-1}_{l^1_1(0)l_a(0)}g_{l^2_1(0)l_a(0)}g_{rl^2_1(0)}\right)\nn\\
&&\phantom{\eps^{ikm}D^{kj}}\left[\left(g_{rl^1_1(0)}^{-1}g^{-1}_{l^1_1(0)l_a(0)}g_{l^2_1(0)l_a(0)}g_{rl^2_1(0)}\bX_{\pi_2}^{>\text{cr}}g_{rl^2_1(0)}^{-1}g_{l^2_1(0)l_a(0)}^{-1}g_{l^1_1(0)l_a(0)}g_{rl^1_1(0)}\right)^m-\big(\bX_{\pi_1}^{>\text{cr}}\big)^m\right],\nn\\
\label{case3PBfinal}
\ea
where for the last equality we have used the fact that
\be
g_{rl^1_1(0)}^{-1}g^{-1}_{l^1_1(0)l_a(0)}g_{l^2_1(0)l_a(0)}g_{rl^2_1(0)}=g_{rl^1_1(0)}^{-1}g_{l^1_1(0)l_b(1)}^{-1}h_bg_{l^2_1(0)l_b(0)}g_{rl^2_1(0)},
\ee
which holds since the loop by which these parallel transports differ does not include a vertex.

Therefore, ignoring the various parallel transports (which would be equal to the identity in case of vanishing curvature) in \eqref{case3PBfinal}, we obtain the difference of the parts of the fluxes lying after the crossing. In the case of vanishing curvature, this is simply equal to the flux observable associated to a co-path going from the end point of $\pi_1$ to the end point of $\pi_2$ (see section \ref{sec:geom} on the geometric interpretation of the fluxes).

Finally, we can look at the case in which two co-paths meet at a vertex without crossing each other, and have the same orientation. This situation can be described by the following structure
\bas
\bX_{\pi_1}&=&Z_1+g^{-1}_{rl^1_1(0)}g_{l_1^1(0)l_a(0)}^{-1}X_{l_a}g_{l_1^1(0)l_a(0)}g_{rl^1_1(0)}+g^{-1}_{rl^1_1(0)}g_{l_1^1(0)l_b(0)}^{-1}X_{l_b}g_{l_1^1(0)l_b(0)}g_{rl^1_1(0)},\\
\bX_{\pi_2}&=&Z_2+g^{-1}_{rl^2_1(0)}g_{l^2_1(0)l_a(1)}^{-1}h_a\left(g^{-1}_{l_a(0)l_b(0)}h_b^{-1}Y_bh_bg_{l_a(0)l_b(0)}\right)h_a^{-1}g_{l^2_1(0)l_a(1)}g_{rl^2_1(0)}.\label{touching2}
\eas
With these two integrated fluxes, we will obtain once again two non-vanishing contributions to the Poisson bracket. We denote by $T_1$ the one resulting from the bracket $\lb X_{l_a},h_a^{\pm1}\rb$, and by $T_2$ the one resulting from the bracket $\lb X_{l_b},h_b^{\pm1}\rb$. Denoting by $Y'_b$ the term in parenthesis in \eqref{touching2}, we have
\ba\label{touchingT_1}
T_1^{ij}&=&\lb\left(g^{-1}_{rl^1_1(0)}g_{l_1^1(0)l_a(0)}^{-1}X_{l_a}g_{l_1^1(0)l_a(0)}g_{rl^1_1(0)}\right)^i,\left(g^{-1}_{rl^2_1(0)}g_{l^2_1(0)l_a(1)}^{-1}h_aY'_bh_a^{-1}g_{l^2_1(0)l_a(1)}g_{rl^2_1(0)}\right)^j\rb\nn\\
&=&-\eps^{kjm}D^{ki}\left(g^{-1}_{rl^2_1(0)}g_{l^2_1(0)l_a(1)}^{-1}h_ag_{l_1^1(0)l_a(0)}g_{rl^1_1(0)}\right)\left(g^{-1}_{rl^2_1(0)}g_{l^2_1(0)l_a(1)}^{-1}h_aY_b'h_a^{-1}g_{l^2_1(0)l_a(1)}g_{rl^2_1(0)}\right)^m\nn\\
&=&-\eps^{ikm}D^{kj}\left(g_{rl^1_1(0)}^{-1}g_{l_1^1(0)l_a(0)}^{-1}h_a^{-1}g_{l^2_1(0)l_a(1)}g_{rl^2_1(0)}\right)\nn\\
&&\phantom{-\eps^{ikm}D^{kj}}\left(g_{rl^1_1(0)}^{-1}g_{l_1^1(0)l_a(0)}^{-1}g^{-1}_{l_a(0)l_b(0)}h_b^{-1}Y_bh_bg_{l_a(0)l_b(0)}g_{l_1^1(0)l_a(0)}g_{rl^1_1(0)}\right)^m,
\ea
and for the second contribution we get
\ba\label{touchingT_2}
T_2^{ij}&=&\lb\left(g^{-1}_{rl^1_1(0)}g_{l_1^1(0)l_b(0)}^{-1}X_{l_b}g_{l_1^1(0)l_b(0)}g_{rl^1_1(0)}\right)^i,\left(g^{-1}_{rl^2_2(0)}h_b^{-1}Z_bh_bg_{rl^2_2(0)}\right)^j\rb\nn\\
&=&\eps^{ikm}D^{kj}\left(g^{-1}_{rl^1_1(0)}g_{l_1^1(0)l_b(0)}^{-1}g_{rl^2_2(0)}\right)\left(g^{-1}_{rl^1_1(0)}g_{l_1^1(0)l_b(0)}^{-1}h_b^{-1}Z_bh_bg_{l_1^1(0)l_b(0)}g_{rl^1_1(0)}\right)^m,
\ea
where in the first line we have introduced
\be
g^{-1}_{rl^2_b(0)}\coloneqq g^{-1}_{rl^2_1(0)}g_{l^2_1(0)l_a(1)}^{-1}h_ag^{-1}_{l_a(0)l_b(0)}.
\ee
The holonomies appearing in the two contributions \eqref{touchingT_1} and \eqref{touchingT_2} agree. In particular, we have that
\be
g_{rl^1_1(0)}^{-1}g_{l_1^1(0)l_a(0)}^{-1}h_a^{-1}g_{l^2_1(0)l_a(1)}g_{rl^2_1(0)}=g^{-1}_{rl^1_1(0)}g_{l_1^1(0)l_b(0)}^{-1}g_{rl^2_b(0)},
\ee
since both co-paths $\pi_1$ and $\pi_2$  share the parallel transport $g_{l_a(0)l_b(0)}$. Therefore, we can conclude that 
\be
\lb\bX_{\pi_1}^i,\bX_{\pi_2}^j\rb=0
\ee
for the case in which the fluxes meet at a vertex and do not cross (regardless of the relative orientation of the co-paths). This nicely fits with case 2 treated above in the limit in which the edge $e$ on which the co-paths coincides is associated to a vanishing flux, $X_l\rightarrow0$.

We can now also discuss the case in which $\bX_{\pi_1}$ either starts or ends at a vertex of $\bX_{\pi_2}$. In this case, we obtain as a result of the Poisson bracket either the term $T^{ij}_1$ or $T^{ij}_2$, depending on how the parallel transport for one flux cuts the co-path of the other flux.

\begin{center}
\begin{figure}[h]
\includegraphics[scale=0.7]{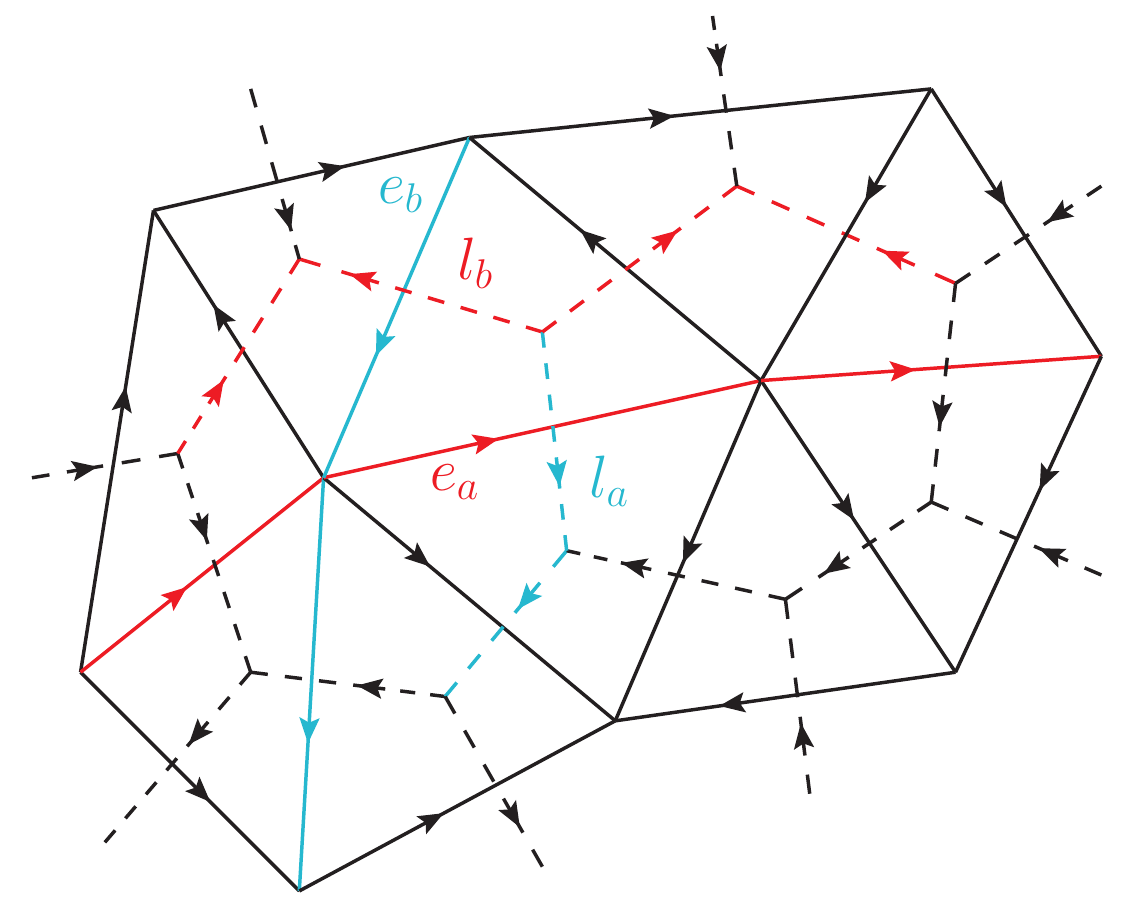}
\caption{Example of co-paths $\pi_1=\ldots\circ e_a\circ\ldots$ (solid red) and $\pi_2=\ldots\circ e_b\circ\ldots$ (solid blue) which are crossing each other at a vertex. One can see that the shadow graph $\Gamma_{\pi_2}$ (dashed blue) of $\pi_2$ is intersecting $\pi_1$, while the shadow graph $\Gamma_{\pi_1}$ (dashed red) of $\pi_1$ intersects $\pi_2$. We assume that the transport to the root of the integrated fluxes defined on $\pi_1$ and $\pi_2$ does not intersect the co-paths.}
\label{fig:Flux-bracket-2d-3}
\end{figure}
\end{center}

\subsubsection{Summary}

\noindent In general, we have seen that we obtain contributions to the Poisson bracket between two fluxes $\bX_{\pi_1}$ and $\bX_{\pi_2}$ in the following situations:

\begin{itemize}
\item The flux $\bX_{\pi_1}$ crosses the parallel transport to the root of the flux $\bX_{\pi_2}$. In this case, we obtain for the Poisson bracket an expression involving $\bX_{\pi_2}$, contracted with an epsilon tensor (and with possible parallel transports).

\item The two fluxes share an edge $e$ dual to link $l$. In this case, the contribution to the Poisson bracket comes from the non-commutativity of $X_l$ with itself. The result will be the flux $X_l$, again contracted with an epsilon tensor and with a possible decoration with some parallel transport.

\item The flux $\bX_{\pi_1}$ contains an edge $e$, dual to a link $l$, which cuts the parallel transport in the shadow graph of the second flux $\bX_{\pi_2}$ (and possibly vice-versa). This is basically a generalization of the first case. The flux $X_l$ will give a non-vanishing Poisson brackets with each of the terms in $\bX_{\pi_2}$ which involve the holonomy $h_l$ in their parallel transport. We therefore do not obtain the full flux $\bX_{\pi_2}$, but only the flux for the part of the path that lies ``after'' $h_l$ (i.e. for which $h_l$ is needed in the parallel transport). If an edge of $\pi_2$ also crosses a parallel transport in the shadow graph of $\bX_{\pi_1}$, the same result applies and we obtain another contribution to the Poisson bracket.
\end{itemize}

If several of the above cases occur, or if a single case occurs several times, the corresponding contributions add up, and in order to compute the term associated to a given case, all the other variables can be treated as constants on the phase space, as we did in the above derivation of case 3. This follows from the Leibniz rule, which holds for the Poisson brackets.

Although we computed the Poisson brackets on a discrete phase space $\mathcal{M}_\Delta$, the structure of the computation and the result is independent of the choice of (sufficiently refined) particular phase space. The structure of the crossings of the co-paths or of the parallel transports will not change under refining operations (remember that the shadow graph has to stay as near as possible to the co-path under refinement).

Therefore, we can conclude an even stronger result for these Poisson brackets than the general statement of section \ref{connecting}, which is that
\be\label{case4}
\lb\big(\bX'_{\pi_1}\big)_\Delta,\big(\bX'_{\pi_2}\big)_\Delta\rb=\lb\big(\bX_{\pi_1}\big)_\Delta,\big(\bX_{\pi_2}\big)_\Delta\rb'+\sum_I\psi_I\mathcal{C}^I_{\Delta,\Delta'}.
\ee
This means that we also determined the extension of the Poisson bracket on the right-hand side, which leads to a vanishing of the terms proportional to the constraints.

In this sense, the computation does actually give us the continuum result for the Poisson brackets. In fact, we can compute for instance the case in which two co-paths intersect, directly with the definition of the flux observables coming from the continuum connection and triad fields in \eqref{simplicial flux}. This computation is presented in appendix \ref{PB continuum}.

The continuum analogues for case 1 and case 2 are very similarly to the (continuum) computation of the Poisson brackets between an holonomy and a flux observable, and between the components of a flux observable. These computations can be found in \cite{qsd7}. All the computations in the continuum lead to results in agreement with the computations on the discrete phase space presented here.

\subsection{Algebra of fluxes in $\boldsymbol{d=3}$ spatial dimensions}

\noindent The computations for $d=3$ spatial dimensions are structurally the same as for $d=2$. For instance, the case in which one flux observable cuts through the parallel transport from the root to another flux observable, will lead to the same Poisson bracket as in case 1 treated above. Likewise, the case in which two fluxes share one triangle (without the parallel transports cutting through the fluxes) will lead to the same Poisson bracket as in case 2 treated above.

Also, for the case in which two surface co-paths intersect or meet at one edge or one vertex, the computation will follow the same structure as in $d=2$. There is however more freedom in the relative orientation of the holonomies in the parallel transports and the (cutting) fluxes, as we will see in the following example.

In this example, we assume that two surface co-paths intersect along an edge, as depicted in figure \ref{fig:Flux-bracket-3d}. The fluxes representing this situation are of the form
\bas
\bX_{\pi_1}&=&Z_1+g^{-1}_{rl^1_1(0)}g_{l^1_1(0)l_a(0)}^{-1}X_{l_a}g_{l^1_1(0)l_a(0)}g_{rl^1_1(0)}+g^{-1}_{rl^1_1(0)}g_{l^1_1(0)l_b(0)}^{-1}h_b^{-1}Y_bh_bg_{l^1_1(0)l_b(0)}g_{rl^1_1(0)},\q\\
\bX_{\pi_2}&=&Z_2+g^{-1}_{rl^2_1(0)}g_{l^2_1(0)l_b(0)}^{-1}X_{l_b}g_{l^2_1(0)l_b(0)}g_{rl^2_1(0)}+g^{-1}_{rl^2_1(0)}g_{l^2_1(0)l_a(0)}^{-1}h_a^{-1}Y_ah_ag_{l^2_1(0)l_a(0)}g_{rl^2_1(0)}.\q
\eas
Therefore, we have a parallel transport in $\bX_{\pi_2}$ cutting through the surface path of $\bX_{\pi_1}$ and vice-versa. Note that we can rearrange the surface trees so that $h_a$ or $h_b$ are both or (individually) replaced by their inverse (this would then also change $Y_a$ and $Y_b$, and lead to the appearance of $h_a$ or $h_b$ in the parallel transport for $X_{l_b}$ or $X_{l_a}$).

There are again two non-vanishing contributions to the Poisson bracket, which can be read off from the corresponding $d=2$ computations. They are given by
\ba
T_1^{ij}&=&\lb\left(g^{-1}_{rl^1_1(0)}g_{l^1_1(0)l_a(0)}^{-1}X_{l_a}g_{l^1_1(0)l_a(0)}g_{rl^1_1(0)}\right)^i,\left(g^{-1}_{rl^2_1(0)}g_{l^2_1(0)l_a(0)}^{-1}h_a^{-1}Y_ah_ag_{l^2_1(0)l_a(0)}g_{r l^2_1(0)}\right)^j\rb\nn\\
&=&\eps^{ikm}D^{kj}\left(g^{-1}_{rl^1_1(0)}g_{l^1_1(0)l_a(0)}^{-1}g_{l^2_1(0)l_a(0)}g_{rl^2_1(0)}\right)\left(g^{-1}_{rl^1_1(0)}g_{l^1_1(0)l_a(0)}^{-1}h_a^{-1}Y_ah_ag_{l^1_1(0)l_a(0)}g_{rl^1_1(0)}\right)^m,\nn\\
\ea
and
\ba
T_2^{ij}&=&\lb\left(g^{-1}_{rl^1_1(0)}g_{l^1_1(0)l_b(0)}^{-1}h_b^{-1}Y_bh_bg_{l^1_1(0)l_b(0)}g_{rl^1_1(0)}\right)^i,\left(g^{-1}_{rl^2_1(0)}g_{l^2_1(0)l_b(0)}^{-1}X_{l_b}g_{l^2_1(0)l_b(0)}g_{rl^2_1(0)}\right)^j\rb\nn\\
&=&\eps^{ikm}D^{kj}\left(g^{-1}_{rl^1_1(0)}g_{l^1_1(0)l_b(0)}^{-1}g_{l^2_1(0)l_b(0)}g_{rl^2_1(0)}\right)\left(g^{-1}_{rl^1_1(0)}g_{l^1_1(0)l_b(0)}^{-1}h_b^{-1}Y_bh_bg_{l^1_1(0)l_b(0)}g_{rl^1_1(0)}\right)^m.\nn\\
\ea
Therefore, the Poisson brackets between the two fluxes,
\be
\lb\bX_{\pi_1}^i,\bX_{\pi_2}^j\rb=T_1^{ij}+T_2^{ij},
\ee
give a flux associated to the piece of the co-path $\pi_1$ that is parallel transported with $h_2$, and a flux associated to the piece of the co-path $\pi_2$ that is parallel transported with $h_1$. These can be rewritten as one flux observable with some additional decoration with parallel transports.

\begin{center}
\begin{figure}[h]
\includegraphics[scale=0.7]{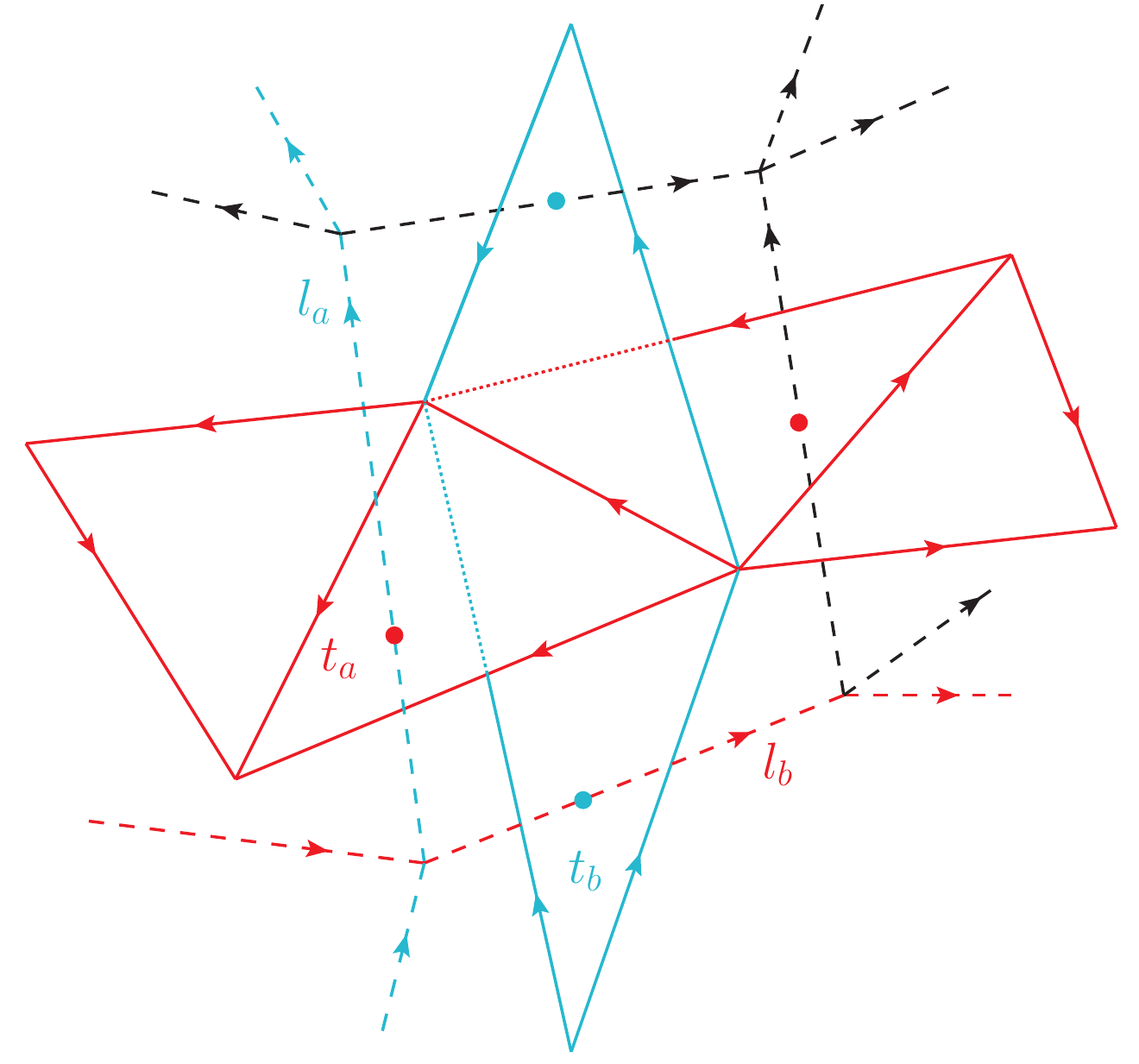}
\caption{Example of two surface co-paths $\pi_1=\ldots\circ t_a\circ\ldots$ (solid red) and $\pi_2=\ldots\circ t_b\circ\ldots$ (solid blue) which are crossing each other along an edge. The shadow graph $\Gamma_{\pi_2}$ (dashed blue) of $\pi_2$ is intersecting $\pi_1$, while the shadow graph $\Gamma_{\pi_1}$ (dashed red) of $\pi_1$ intersects $\pi_2$.}
\label{fig:Flux-bracket-3d}
\end{figure}
\end{center}

\section{Spatial diffeomorphisms in $\boldsymbol{(2+1)}$ dimensions}
\label{sec:diffeos}

\noindent In this section, we would like to comment on the action of spatial diffeomorphisms in the quantum theory. To this end, we will assume a certain structure for the quantization, which has been laid out in \cite{paper1}. In this quantization, we expect that only the exponentiated (integrated) fluxes exist as operators. We will therefore first discuss these in the next subsection. Then, we will start from a definition of the action of spatial diffeomorphisms as displacements of the embedded vertices of the underlying triangulation. Starting from this definition, we will then argue that this action is related to that generated by the spatial diffeomorphism constraints.

\subsection{Exponentiated integrated fluxes}

\noindent As a first step towards the definition of the quantum theory, we can consider the action of the various flux variables which we have introduced on functions of the holonomies. Later on we will consider kinematical states depending on a collection of group elements denoted by $\psi(g_1,\ldots,g_n)$, but for the moment we can consider for simplicity a function $\psi(g)$ of a single element $g\in G$. The action of the fluxes leaves the space of functions over a fixed dual graph invariant.

First, let us recall that the action of the simplicial fluxes \eqref{simplicial flux} on functions $\psi(g)$ over the group $G$ is given by
\be
X_l^i\triangleright\psi(g)=\i\hbar\mathcal{L}^i_l\triangleright\psi(g).
\ee
In this expression, $\mathcal{L}^i_l$ is the left-invariant derivative acting on the group variable $g$ and defined as
\be
\mathcal{L}^i\triangleright\psi(g)=\left.\frac{\de}{\de t}\psi\big(ge^{t\tau^i}\big)\right|_{t=0},
\ee
where $\{\tau^i\}_{i=1,\ldots,\text{dim}(\mathfrak{g})}$ is a basis of generators for the Lie algebra $\mathfrak{g}$. Notice that if we reverse the orientation of the link, the simplicial flux $X^i_{l^{-1}}$ acts like the right-invariant derivative
\be
\mathcal{R}^i\triangleright\psi(g)=\left.\frac{\de}{\de t}\psi\big(e^{-t\tau^i}g\big)\right|_{t=0}.
\ee
This is the reason for which the Poisson brackets \eqref{SU(2)brackets1} between holonomies and fluxes depend on the relative orientation of the two links. Now, by virtue of the elementary Poisson bracket $\lb X_l^i,h_l\rb=h_l\tau^i$, for any flux $\text{Ad}_{g^{-1}}(X^i)=(g^{-1}Xg)^i$ rotated by an adjoint action we have
\be
\lb\big(g^{-1}X^j_l\tau_jg\big)^i,h_l\rb=h_lg\tau^ig^{-1}.
\ee
This type of parallel transport of the basic fluxes is exactly the one appearing in the definition of the integrated rooted fluxes.

We are ultimately interested in the action of the exponentiated (integrated) fluxes. If we exponentiate the flow of a flux, we obtain a right translation by the group element $\exp(\alpha_i\tau^i)$, and we can write that
\be
\exp\left(\alpha_i\lb X_l^i,\cdot\rb\right)=R_l^{\exp(\alpha_i\tau^i)}.
\ee
With a slight abuse of notation, we will introduce the group element $\alpha=\exp(\alpha_i\tau^i)$, and therefore denote the action of and exponentiated flux as
\be
R_l^\alpha\triangleright\psi(g_l)=\psi(g_l\alpha).
\ee
Similarly, the action of left translations is given by
\be
L_l^\alpha\triangleright\psi(g_l)=R_{l^{-1}}^\alpha\triangleright\psi(g_l)=\psi(\alpha^{-1}g_l).
\ee
This can now straightforwardly be applied to the elementary rooted fluxes $\bX_l=g^{-1}_{rl(0)}X_lg_{rl(0)}$ and to the integrated fluxes $\bX_\pi$. If we exponentiate the flow of a transported flux, we obtain a right translation by the group element $g\exp(\alpha_i\tau^i)g^{-1}=g\alpha g^{-1}$, and we can write that
\ba\label{expflux}
\exp\left(\alpha_i\lb\big(g^{-1}X_lg\big)^i,\cdot\rb\right)=R_l^{g\alpha g^{-1}}=R^{\text{Ad}_g(\alpha)}_l.
\ea

\subsection{The generator of spatial diffeomorphisms}

\noindent The quantum theory will eventually support states that are almost everywhere flat but have distributional curvature around a finite number of (embedded) points $x_I\in\Sigma$. Such states are cylindrical over triangulations for which every point $x_I$ describes a vertex $v_I$ of the triangulation. As discussed in \cite{paper1}, we expect that such cylindrical functions can be described by a basis
\ba
\psi_{\alpha_I}\coloneqq\prod_I\delta\big(g_I\alpha_I^{-1}\big),
\ea
where $g_I$ describes a rooted holonomy variable around the vertex $v_I$, and $\alpha_I$ is a set of group elements. The inner product with respect to which this basis becomes normalizable will be discussed in future work. Here we need this basis to motivate the definition of spatial diffeomorphisms.

The basis describes curvature excitations localized at the vertices $v_I$. The position of these vertices is equivalent to the ``embedding information'', which one eventually expects to be modded out by the action of spatial diffeomorphisms.

Therefore, we can define an action of spatial diffeomorphisms by moving the vertices, i.e. by a change of the position of the vertices in $\Sigma$. This can be understood as being dual to the understanding of diffeomorphisms in the AL representation, which change the embedding of the underlying graph $\Gamma$. Vertex displacements define also diffeomorphisms for discrete approaches such as Regge gravity \cite{bahrdittrich09a}.

In what follows, we would like to motivate a particular expression for the generator of these diffeomorphisms, which would correspond to the diffeomorphism constraints. Again, the discussion will be on a heuristic level, since the diffeomorphisms may rather lead to a non-weakly continuous action and thus a generator might not exist. However, the question of defining generators for the diffeomorphisms has also appeared in the AL representation \cite{madhavan}. With this discussion, we want to point out the compelling geometrical interpretation of the BF representation.

Let us consider two geometric triangulations $\Delta$ and $\Delta'$ related by a shift of a vertex, as represented on figure \ref{fig:vertex-move}.

\begin{center}
\begin{figure}[h]
\includegraphics[scale=0.7]{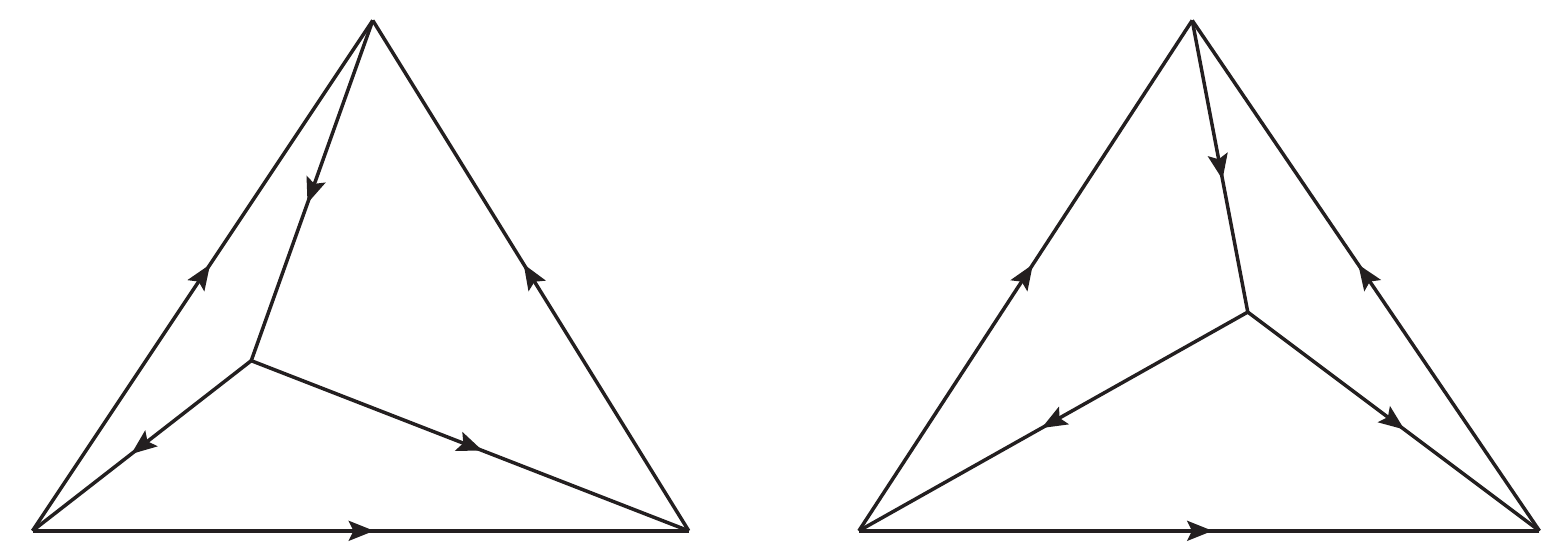}
\caption{Two geometric triangulations related by a shift of a vertex.}
\label{fig:vertex-move}
\end{figure}
\end{center}

If we want to describe the operation that shifts the curvature from one vertex to the other one, we have to consider a common refinement for these two triangulations. Such a common refinement is represented on figure \ref{fig:3d-diffeos}.

\begin{center}
\begin{figure}[h]
\includegraphics[scale=0.7]{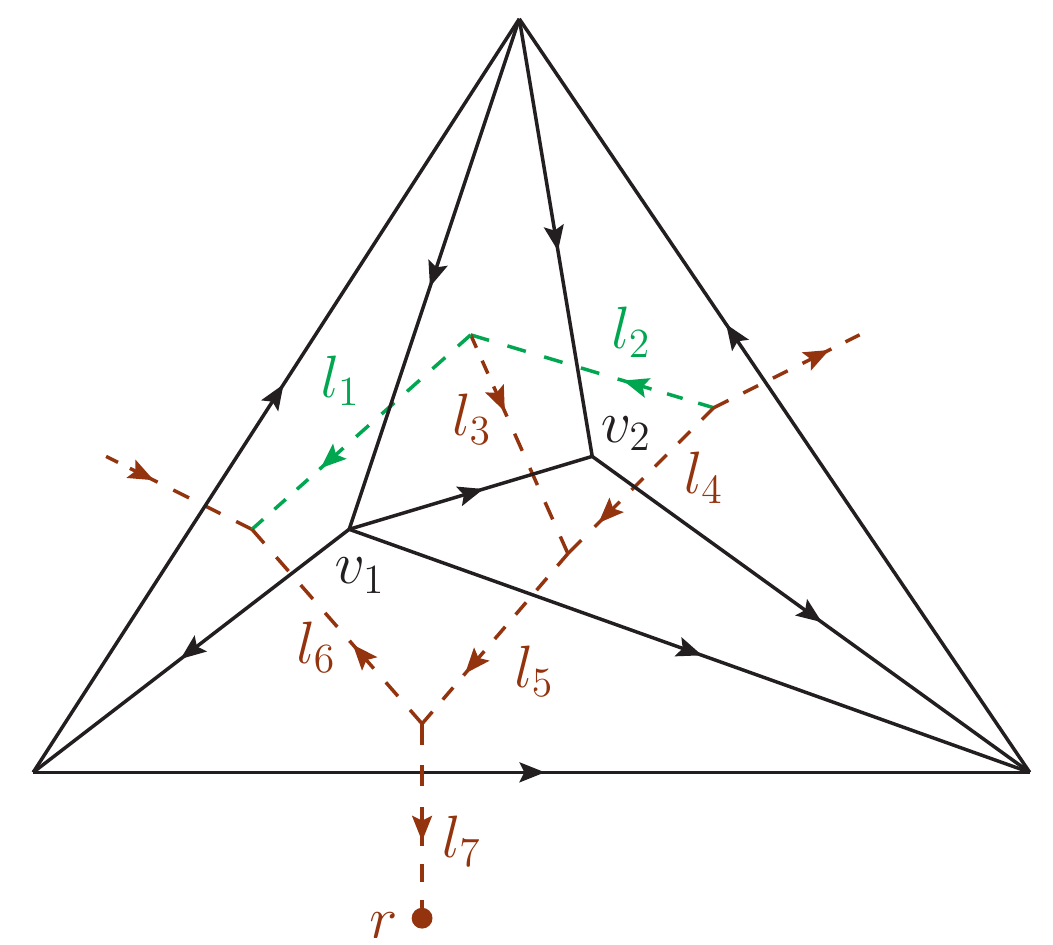}
\caption{Common refinement of the two triangulations represented on figure \ref{fig:vertex-move}, together with its dual graph and a choice of rooted tree. The leaves are represented in dashed green and the branches in dashed brown.}
\label{fig:3d-diffeos}
\end{figure}
\end{center}

Choosing a tree in the dual graph to the triangulation allows us to specify the parallel transport to the root. We denote by $g_{\ell_1}$ the rooted holonomy around the (left) vertex $v_1$, and by $g_{\ell_2}$ the rooted holonomy around the (right) vertex $v_2$ in the triangle of figure \ref{fig:3d-diffeos}. These holonomies will include the group elements $h_{l_1}$ and $h_{l_2}$ associated to the leaves $\ell_1=l_1$ and $\ell_2=l_2$. The definition of these rooted closed holonomies is given by \eqref{leaf closed holonomy}.

Ignoring the dependence of the wave functions on other closed holonomies, we can consider (basis) states defined by
\be
\psi_{\alpha_1,\alpha_2}(g_{\ell_1},g_{\ell_2})=\delta\big(g_{\ell_1}\alpha_1^{-1}\big)\delta\big(g_{\ell_2}\alpha_2^{-1}\big).
\ee
The BF vacuum corresponds to $\alpha_1=\openone=\alpha_2$. This vacuum is clearly invariant under the action of the diffeomorphism moving the vertex, since there is no curvature to move. The embedding maps discussed in section \ref{connecting} impose flatness for all additional cycles. Thus, states which result from a refinement of the left triangulation in figure \ref{fig:vertex-move} have $\alpha_2=\openone$, and states resulting from a refinement of the right triangulation have $\alpha_1=\openone$.

We can move curvature from the vertex $v_1$ to the vertex $v_2$ by using the exponentiated form of the following two rooted fluxes:
\be
\bX_{l_1}=g^{-1}_{rl_1(0)}X_{l_1}g_{rl_1(0)},\q\bX_{l^{-1}_2}=g^{-1}_{rl_2(1)}X_{l^{-1}_2}g_{rl_2(1)}.
\ee
According to \eqref{expflux}, the exponentiated flux $\exp\left(\beta_1^i\lb\big(g^{-1}_{rl_1(0)}X_{l_1}g_{rl_1(0)}\big)^i,\cdot\rb\right)$ gives rise to the action\footnote{Note that we ignore the action of $R_{l_1}$ and $R_{l_2}$ on other parts of the wave function (which amounts to setting these parts to constants). Otherwise, $R_{l_1}$ and $R_{l_2}$ would also shift the curvature around an additional vertex, which in figure \ref{fig:3d-diffeos} is the tip of the triangle. However, with an appropriate choice of prescriptors, the action of both these shift operators on the curvature around this vertex cancel each other. This allows to replace the shifts $R_{l_1}$ and $R_{l_2}$ by $R_{l_3}$ in \eqref{diff14} which then only acts on the curvatures around $v_1$ and $v_2$, but not on the curvature around the tip of the triangle.}
\ba
R_{l_1}^{g_{rl_1(0)}\beta_1g^{-1}_{rl_1(0)}}\triangleright\psi_{\alpha_1,\alpha_2}(g_{\ell_1},g_{\ell_2})&=&\delta\left(g_{l_1(1)r}h_{l_1}g_{rl_1(0)}\beta_1g^{-1}_{rl_1(0)}g_{rl_1(0)}\alpha_1^{-1}\right)\delta\big(g_{\ell_2}\alpha_2^{-1}\big)\nn\\
&=&\delta\big(g_{\ell_1}\beta_1\alpha_1^{-1}\big)\delta\big(g_{\ell_2}\alpha_2^{-1}\big),
\ea
and exponentiated flux $\exp\left(\beta_2^i\lb\big(g^{-1}_{rl_2(1)}X_{l_2^{-1}}g_{rl_2(1)}\big)^i,\cdot\rb\right)$ gives rise to
\be
R_{l_2^{-1}}^{g_{rl_2(1)}\beta_2g^{-1}_{rl_2(1)}}\triangleright\psi_{\alpha_1,\alpha_2}(g_{\ell_1},g_{\ell_2})=\delta\big(g_{\ell_1}\alpha^{-1}_1\big)\delta\big(g_{\ell_2}\beta_2\alpha_2^{-1}).
\ee

Note that, since $l_1(0)=l_2(1)$ in our example, we also have $g_{rl_2(1)}= g_{rl_1(0)}$ and therefore the parallel transport to the root is the same for the two fluxes. If we choose furthermore that $\beta_1=\beta_2=\beta$, we can use the gauge invariance of the states to rewrite the composition of the two exponentiated fluxes as an exponentiated rooted flux for $\bX_{l_3}$ (which again is transported along the same path to the root). Indeed, one has that
\be\label{diff14}
R_{l_2^{-1}}^{g_{rl_2(1)}\beta g^{-1}_{rl_2(1)}}R_{l_1}^{g_{rl_1(0)}\beta g^{-1}_{rl_1(0)}}=R_{l_3}^{g_{rl_1(0)}\beta^{-1}g^{-1}_{rl_1(0)}},
\ee
where the right-hand side is the exponentiated flux corresponding to the link $l_3$.

Now, we want to discuss the situation in which we have a state resulting from a refinement of the triangulation on the left in figure \ref{fig:vertex-move}. In this case we therefore have $\alpha_2=\openone$. In the refined triangle of figure \ref{fig:3d-diffeos}, we want to ``move the curvature'' around the left vertex $v_1$ to the right vertex $v_2$, so that the resulting state coincides with one obtained from the refinement of the right triangle in figure \ref{fig:vertex-move}.

Therefore, we should choose $\beta$ to be an operator, in turn acting as the holonomy $\hat{g}_{\ell_1}$. In the representation we are using, the holonomy operators are diagonal and act as 
\be
D(\hat{g}_{\ell_1})\triangleright\psi_{\alpha_1,\alpha_2}=D(\alpha_1)\psi_{\alpha_1,\alpha_2}
\ee
for any representation $D$ of the group (in contrast, the parallel transport matrices appearing in the exponentiated fluxes are rather a notational device to deal with the gauge-invariant description. If we were to use a gauge fixing, the exponentiated flux would only involve a flux operator).

In \eqref{diff14}, we should therefore (formally) replace $\beta$ with $\hat{g}_{\ell_1}$. However, the expression of this operator as an exponentiated flux with holonomy (operator) dependent parameter is now very involved due to the fact that the flux acts on this holonomy in the higher order terms of the exponential. At the linear level (and ignoring the fact that the exponentiated fluxes might not define a weakly continuous family), we can see that
\be
\left.\left(R_{l_3}^{g_{rl_1(0)}\beta^{-1}g^{-1}_{rl_1(0)}}\right)\right|_{\beta=\hat{g}_{\ell_1}}\sim\openone+\i\hbar\widehat{\big(g^{-1}_{\ell_1}\bX_{l_3}\big)}.
\ee
This can in fact be interpreted as a discretization of the spatial diffeomorphism constraints $C_a=F^i_{ab} E^b_i$ for gauge-theoretical formulations of gravity. 

As is well known, the constraints for $(2+1)$-dimensional gravity can be written in the form $C_\ell=F_\ell\coloneqq g_\ell-g_\ell^{-1}$, which makes them Abelian. A more transparent geometric interpretation can be obtained by contracting the constraints with fluxes, thereby obtaining
\be
C_{Y,\ell}\coloneqq-2\tr\big(FY^i\tau^i\big),
\ee
which is the constraint generating a vertex deformation in the direction of $Y$. On the constraint hypersurface $F=0$, we can replace $Y$ with any phase space dependent function, including a flux. Therefore,
\be\label{conalg}
C_{X,\ell}\coloneqq-2\tr\big(FX^i\tau^i\big)
\ee
is indeed the generator of the vertex displacement in the direction of the flux associated to the edge connecting the old ($v_1$) and new ($v_2$) positions of the vertex being moved.

In particular, it is straightforward to confirm that this constraint generates also a change in the fluxes which is consistent with this deformation (on the constraint hypersurface $F=0$). Therefore, we have derived a connection (at a heuristic level) between the notion of diffeomorphisms acting as moving the embeddings of the vertices, and the notion of diffeomorphism as changing the \textit{physical} metric data\footnote{This however should not be confused with proper Dirac observables. The physical metrical distance is the distance between vertices as given by the fluxes, as opposed to the distance of the vertices with respect to the auxiliary metric, which we use to define the (geodesic) embedding of the edges. However, the vertices themselves are not ``physical'' as long as they are not attached to point particles. The fact that the constraints generate vertex displacements shows that they are rather gauge information.} in a way consistent with the displacement of a vertex. This latter notion is also how diffeomorphisms act in discrete gravity, such as Regge gravity or three-dimensional spin foam models \cite{williams,dittrichreview,bahr,louapre}.

One can derive a consistent constraint algebra for the constraints \eqref{conalg} (defined on a discrete phase space) in the case of $(2+1)$-dimensional gravity \cite{bonzomdittrich}. This ensures that the constraints define a consistent flow, which can then be integrated to define an exponentiated flux with holonomy-dependent parameters, just like on the left hand side of \eqref{diff14}. However, the consistency of the flow depends crucially on the restriction to the constraint hypersurface of vanishing curvature. Thus, already a splitting of the constraints into a closed algebra of diffeomorphisms and a Hamiltonian is difficult if constraints at different vertices are involved. Furthermore, the algebra contains higher order terms (in powers of the constraints).

It would be interesting to establish in $(3+1)$-dimensional gravity a similar connection between the diffeomorphisms moving the embeddings of the vertices, and the diffeomorphisms changing the physical metric. There are however several difficulties with this task. One is that the symmetries of BF theory now rather describe the translation of edges instead of the translation of vertices. It has been proposed that the simplicity constraints should break this edge translation symmetry, so that the remaining symmetries describe the movement of vertices \cite{zapata}. This means that only certain combinations of edge translations should be allowed, which could in turn be interpreted as vertex translations. This could be achieved by demanding that the simplicity constraints have to describe a metrical discrete geometry, and that this property be preserved by the symmetries. These points are however so far not fully realized in the LQG phase space \cite{simplicity2,anza,spinfoamreview1,dittrichryan2}. It would be extremely beneficial to study this problem in more depth, to clarify the geometric interpretation of LQG geometries, and the question of how to impose the diffeomorphism constraints.

\section{Imposing the dynamics}
\label{sec:dyn}

\noindent In this section we would like to briefly discuss possibilities for imposing a dynamics. In  $(2+1)$-dimensional gravity without a cosmological constant, the constraints demand the vanishing of (local) curvature, which can be implemented straightforwardly. In particular, the constraints can be implemented on the discrete phase spaces (that is, they are not ``discretization-changing''), and moreover form a closed algebra \cite{bonzomdittrich}. An infinite refinement of the constraints (as is required for the AL representation \cite{PerezNoui}) is not necessary since flatness is automatically implement away from the vertices.

In the case of a non-trivial (spatial) topology, some global degrees of freedom remain. More degrees of freedom can also be added by allowing particle insertions \cite{FL1,Matschull,karim1,karim2}, which here are naturally represented by the point-like curvature defects.

Time evolution can then also be described by evolution via Pachner moves, as explained in \cite{hoehn1,hoehn2}. The framework in \cite{hoehn1,hoehn2} was developed in order to deal with phase spaces whose (kinematical) dimension can change during discrete time evolution. The need to change the triangulation can arise in the case of particles if these can for instance scatter, or at other transition points \cite{thooft1,thooft2,ziprick1,ziprick2}. The change in phase space dimension leads to pre- and post-constraints, which we discussed in section \ref{connecting} in the context of embedding coarser phase spaces into finer ones. Therefore, this simplicial evolution scheme naturally fits within the framework described in section \ref{connecting}.

A more complicated case arises if we consider $(2+1)$-dimensional gravity with a cosmological constant. Indeed, in this case physical solutions describe homogeneously curved spaces. Thus, we expect that if one starts from the phase space describing almost everywhere flat configurations, we need to go to the limit of infinite refinement if we impose constraints, for instance along the lines of \cite{dittrich14}. It would be interesting to pursue this program (which has also been studied starting from the AL representation in \cite{pranzetti-perez}) further. A question would be whether one can recover a quantum group structure as an effective description, which would describe $(2+1)$-dimensional gravity with a cosmological constant \cite{maite1,maite2,aldoetal,Turaev:1992hq,smolin,NPP,pranzetti}. This effective description could then be used to define a new inductive Hilbert space or modified projective phase space structure, for which the vacuum and embeddings describe homogeneously curved geometries instead of flat geometries. At the classical level (e.g. with the Regge action), it has already been shown that an action describing homogeneously curved tetrahedra arises from coarse graining the standard Regge action (describing flat tetrahedra) with a cosmological constant term \cite{improved}.

The most interesting case is of course that of $(3+1)$-dimensional gravity. There, the dynamics can be imposed either via a Hamiltonian constraint or in a more covariant setup. If we want to impose the Hamiltonian constraint on a given discrete phase space, the problem is to find a discretization leading to a consistent constraint algebra. In fact, since diffeomorphism symmetry is typically broken by discretizations, a discrete evolution scheme rather results in pseudo-constraints \cite{bahrdittrich09a,bahr,gambini1,gambini2,gambini3} which describe equations of motion that couple different time steps only weakly.

A way out of this problem might be to turn to the continuum phase space, which would allow to discuss discretization-changing constraints (similar to Thiemann's constraints \cite{thiemannH1,thiemannH2}). In fact, one would expect that the dynamics will necessarily turn on infinitely many (curvature) excitations, even if one starts from a configuration with only very few curvature excitations. Nevertheless, the question of finding a consistent constraint algebra on this continuum phase space remains open.

An alternative approach to working with constraints is to accept once and for all a discrete dynamics (i.e. a dynamics approximating the continuum one via a discretization), and the fact that this discrete dynamics breaks the diffeomorphism symmetry and the related constraints. One would then rather impose the discrete dynamics via a time evolution generated by an action functional. This can be implemented by Pachner moves that change the spatial triangulation, and can be interpreted in space-time as gluing simplices onto the hypersurface \cite{hoehn1,hoehn2,hoehn3,alesci-rovelli}. In the continuum limit, i.e. in the regime of very fine triangulations, one would expect that symmetries are restored, and that the dynamics imposed in this way projects onto the constraint hypersurface. A discussion about how to construct this continuum limit in the quantum case, and the relation between this construction and a renormalization flow, can be found in \cite{timeevol,dittrich14,bahr14}. Similarly to the case of $(2+1)$-dimensional gravity with a cosmological constant, one can expect that a new (effective) vacuum will arise, which would be a physical state of $(3+1)$-dimensional gravity.

\section{Summary and discussion}
\label{sec:summary}

\noindent In this paper, we have introduced the classical framework underlying a new representation for LQG based on the BF vacuum. This involves in particular a shift in emphasis from holonomies, which in the AL representation generate excitations of spatial geometry, to fluxes, which in the BF representation generate curvature excitations.

The new BF representation is based on an entirely new setup and presents notable differences with the AL representation. Let us recall what the main features of the new representation are, and the key results which we have derived.

\begin{itemize}
\item Whereas the AL representation is based on the partially ordered set of dual graphs and refinement operations on the links, the BF representation requires to introduce the partially ordered set of triangulations and a set of new refinement operations (here the Alexander moves). While in the AL representation the excitations of geometry are carried by the dual graphs, in the BF representation the curvature excitations are carried by the vertices (in $d=2$ spatial dimensions) or by the edges (in $d=3$) of the triangulation.

\item Whereas the AL representation is based on the composition and coarse graining of holonomies, the BF representation requires the composition of fluxes. We have seen that this requires the introduction of new composite objects called the integrated simplicial fluxes, and that the parallel transport involved in the definition of the simplicial fluxes is crucial in order to achieve a geometrically-meaningful composition. The resulting observables, which are related to Wilson surfaces \cite{wilsonsurfaces1,wilsonsurfaces2}, provide an alternative (dual) way of characterizing quantum geometries. Such an alternative method might in turn be very useful for the construction of the continuum limit of LQG \cite{eteradeformed,dittrich14}, since the fluxes encode information about the spatial geometry.

\item We have provided a geometric interpretation for the integrated flux observables, in particular concerning their dependence on the choice of co-path. The composition of the fluxes can be interpreted as defining coarser observables. In this context, we have observed that the coarse-grained fluxes do not necessarily satisfy the (coarse) Gauss constraints, as was also found in the context of coarse graining of spin networks in \cite{eteradeformed} (this should not be confused with a violation of the ``microscopic'' Gauss constraint, or with a violation of gauge invariance). Here, we can interpret this effect as ``curvature-induced torsion'', which only appears with non-Abelian structure groups.

\item Starting from the family of phase spaces based on triangulations, we have defined a continuum phase space. As we have seen, this requires the definition of consistent projection and embedding maps, which turn out to involve pre- and post-constraints. Such constraints are also essential in discussing canonical maps between phase spaces of different dimensions \cite{hoehn1,hoehn2}. In order to deal with this complication, we have therefore introduced a modified projective limit, which takes into account the fact that configurations with curvature cannot be coarse-grained (or projected) onto triangulations that cannot support this curvature. An alternative characterization of the family of discrete phase spaces puts the constraints at center stage. Indeed, we have also seen that coarser phase spaces arise as symplectic reductions of finer phase spaces with respect to the constraints. This was in fact used in \cite{FGZ} in order to define the discrete phase spaces as symplectic reductions from a continuum phase space. The same two characterizations can also be applied to define a phase space underlying the AL representation.

\item This whole framework has enabled us to clarify and to complete the structure of the Poisson algebra of flux observables. In particular, we have found that the Poisson bracket between two fluxes defined on intersecting co-paths does not result in a more singular object. We saw in fact that the algebra of holonomies and fluxes is closed. In all these computations, the use of simplicial (i.e. parallel transported) fluxes is essential. The results of these computations (on the discrete phase space) can also be confirmed by computations with the continuum expression for the fluxes, as presented in the appendix.

\item We have discussed the action of spatial diffeomorphisms as vertex displacements in $d=2$ spatial dimensions, and argued that this action is generated by a candidate spatial diffeomorphism constraint operator.
\end{itemize}

These are the key features and properties of the new representation, and the main new results obtained in this framework.

Importantly, let us point out that the framework presented in \cite{paper1} and developed in the present paper as well as in the follow up work \cite{toappearb}, is different from defining a BF state as a functional on the Ashtekar--Lewandowski Hilbert space, as was attempted for instance in \cite{karim1,mikovic,ATZ}. Indeed, this amounts to defining the components of a BF state, which itself is not an element of the Ashtekar--Lewandowski Hilbert space, in the spin network basis. In contrast, we have defined here a framework which classically results in a different continuum phase space and quantum mechanically in a unitarily-inequivalent Hilbert space (which we define in \cite{toappearb}). The (cyclic) vacuum state underlying this Hilbert space is a state satisfying the BF constraints, and therefore an element of the Hilbert space. Moreover, we not only defined the vacuum BF state, but an entire Hilbert space carrying a representation of the kinematical observable algebra of LQG. Any state in the new Hilbert space would be a non-normalizable state with respect to the Ashtekar--Lewandowski Hilbert space.

Let us now comment on the imposition of the various constraints. In our construction, the Gauss constraint is already implemented (except at the root). We expect that spatial diffeomorphisms will act as displacements of the vertices, as discussed in section \ref{sec:diffeos}. While this is clearer in $d=2$, in $d=3$ it might require an improved understanding of the status of the simplicity constraints \cite{twisted,simplicity2,anza,dittrichryan2}, since these are supposed to reduce the BF symmetries down to diffeomorphism symmetries \cite{zapata}. In section \ref{sec:dyn}, we have discussed various possibilities for imposing the Hamiltonian constraints, in particular for $d=3$. Here, the use of Pachner moves to describe the dynamics and a continuum limit, as discussed in \cite{cylconsis,dittrich14}, seems to be the most promising approach. The present framework based on the fluxes could indeed be ideal in order to discuss and achieve the coarse graining in geometrical terms. For $d=2$, without a cosmological constant, the imposition of the Hamiltonian constraint is already achieved by this framework. Including a cosmological constant could provide a connection to other approaches which directly start from quantum group structures, or teach us more about the construction of physical vacua by coarse graining \cite{smolin,maite1,maite2,aldoetal,improved,timeevol}.

As we have already emphasized, this framework opens up a new path towards the understanding of quantum geometry. Many properties of our construction, for instance the fact that it is based on almost everywhere flat instead of almost everywhere degenerate configurations, make it also more suitable for a discussion of the dynamics and of issues such as that of the continuum limit and coarse graining.

Finally, let us comment on some possible generalizations and connections to other works.

\begin{itemize}
\item We have based our construction on the poset of triangulations. However, we do not see any a priori difficulties in extending this poset to more general polyhedral complexes. Another possible generalization would be to work on the gauge-covariant phase spaces, and to also allow for defects that violate the (microscopic) Gauss constraints. This can accommodate the description of spinning particles \cite{FL1}. Moreover, since the macroscopic Gauss constraints may be violated, such defects might even be necessary in order to describe ``effective'' configurations.

\item At the classical level, there is no obstruction to generalizing this framework to the structure group $\SL(2,\mathbb{C})$. This would correspond to the self-dual theory defined with a complex Barbero--Immirzi parameter (see the recent discussions \cite{complexBI1,complexBI2,complexBI3,complexBI4,complexBI5,complexBI6}). The crucial issue is however a priori in the quantum theory. This being said, it might be that for the BF-based representation the obstructions arising with the AL representation do not apply. In particular, the fact that $\SL(2,\mathbb{C})$ is non-compact (and not amenable) might not be an issue anymore, since one rather needs a compactifaction of the dual of the group, both for $\SU(2)$ and $\SL(2,\mathbb{C})$.

\item For $(2+1)$-dimensional gravity, the BF-based representation can be used to work on the reduced phase space. Thus, a quantization of this setup could potentially be compared with other approaches involving reduced phase space quantization, for instance those based on Chern--Simons theory \cite{meusburger}.

\item It would very be interesting to study how the new BF-based representation affects the description of blak holes and the construction of loop quantum cosmology. A first step in this direction has already been taken in the work \cite{hannoBH}, which discusses an interesting mixture of representations in the context of black hole physics. There, the connections along the black hole horizon are rather treated in a BF-like representation, whereas the connections transversal to the horizon are kept in the AL representation. This is due to the imposition of the (isolated) horizon boundary conditions, which fix the curvature on the horizon. On a more general level, it would be useful to clarify how the use of curvature excitations challenges the usual picture of the black hole state counting, where one has links of the dual spin network graphs puncturing the horizon.

\item The AL and BF-based representations are in a certain sense dual to each other. This duality has to be understood in a statistical physics sense \cite{savit}. The AL representation corresponds to a strong coupling expansion in lattice gauge theory, whereas the BF-based representation describes the weak coupling limit. Wilson loops and Wilson surface operators (which can be seen as generalizations of 't Hooft operators \cite{thooftop1}) provide the order parameters for these two different phases respectively. In terms of states, both vacua are given by totally squeezed states. This opens the question of whether a generalization of Gaussian states is possible. In general, the difficulty is that the coarse graining of Gaussian states leads to change of the widths of the Gaussians, with requires the introduction of additional parameters \cite{holomorphic}. Another point is that the framework described here leads to first class constraints, whose imposition indeed leads to squeezed states. Alternatively, one might work with second class constraints, which indeed come up in the description of (more standard field-theoretic) vacua with Gaussian states \cite{timeevol,master2}. An interesting development in this direction is the recent work \cite{lanery1,lanery2,lanery3,lanery4}, which however still has to face the crucial question of whether (spatial) diffeomorphism constraints (apart from the Gauss constraints) can be implemented.

\item The framework developed here can also be useful for discussing more general representations. We have seen that the question of the composition of observables under coarse graining is crucial for the definition of the various vacua. We think that this property is important in order to simplify the eventual construction of a physical Hilbert space, which indeed should incorporate a coarse graining of the observables in its inductive limit structure, as explained in \cite{timeevol,dittrich14}.
\end{itemize}

\section*{Acknowledgements}

\noindent We would like to thank Abhay Ashtekar, Benjamin Bahr, Klaus Fredenhagen, Jerzy Lewandowski, Aldo Riello, Carlo Rovelli and Wolfgang Wieland for helpful discussions. This research was supported by Perimeter Institute for Theoretical Physics. Research at Perimeter Institute is supported by the Government of Canada through Industry Canada and by the Province of Ontario through the Ministry of Research and Innovation. MG is supported by the NSF Grant PHY-1205388 and the Eberly research funds of The Pennsylvania State University.

\appendix

\section{Inductive and projective limits}
\label{appendix:limits}

\noindent In order to setup an inductive or projective limit (sometimes also respectively called direct and inverse limits), we need a partially ordered and directed set $\mathcal{S}$ of labels $\{\Delta\}$, together with an order relation $\Delta\prec\Delta'$. Then, we associate an object $\mathcal{M}$ (from a given category) to each label, i.e. we consider the collection $\{\mathcal{M}_\Delta\}_{\Delta\in\mathcal{S}}$.

For a projective limit, we need a set of morphisms (i.e. projection maps) $\mathcal{P}_{\Delta',\Delta}:\mathcal{M}_{\Delta'}\rightarrow\mathcal{M}_\Delta$ which satisfy
\ba
&&\text{(a)}\q\mathcal{P}_{\Delta,\Delta}=\text{id}_\Delta,\ \forall\,\Delta\in\mathcal{S},\nn\\
&&\text{(b)}\q\mathcal{P}_{\Delta',\Delta}\circ\mathcal{P}_{\Delta'',\Delta'}=\mathcal{P}_{\Delta'',\Delta},\ \forall\,\Delta\prec\Delta'\prec\Delta''.
\ea
The projective limit can then be defined as a subset of the direct product
\be
\mathcal{M}_\infty=\lb p_\Delta\in{\prod}_\Delta\mathcal{M}_\Delta\ \big|\ \mathcal{P}_{\Delta',\Delta}(p_\Delta')=p_\Delta,\ \forall\,\Delta\prec\Delta'\rb.
\ee

For an inductive limit, we need a set of morphisms (i.e. embedding maps) $\mathcal{E}_{\Delta,\Delta'}:\mathcal{M}_{\Delta}\rightarrow\mathcal{M}_{\Delta'}$ which satisfy
\ba
&&\text{(a)}\q\mathcal{E}_{\Delta,\Delta}=\text{id}_\Delta,\ \forall\,\Delta\in\mathcal{S},\nn\\
&&\text{(b)}\q\mathcal{E}_{\Delta',\Delta''}\circ\mathcal{E}_{\Delta,\Delta'}=\mathcal{E}_{\Delta,\Delta''},\ \forall\,\Delta\prec\Delta'\prec\Delta''.
\ea
The inductive limit can then be defined as the disjoint union on which an equivalence relation is imposed:
\be
\mathcal{M}_\infty=\cup_\Delta\mathcal{M}_\Delta/\sim,
\ee
where $m_\Delta\sim m'_{\Delta'}$ if there exist a $\Delta''$ such that $\mathcal{E}_{\Delta,\Delta''}(m_\Delta)=\mathcal{E}_{\Delta',\Delta''}(m'_{\Delta'})$. In words, this means that two elements are equivalent if they become eventually equal under refinement.

\section{The Lie algebra $\boldsymbol{\su(2)}$}
\label{appendix:1}

\noindent Our choice of generators for the fundamental representation of $\su(2)$ is given by the matrices
\be
\tau_1=-\f{1}{2}
\begin{pmatrix}
0&\i\\
\i&0
\end{pmatrix},\q\q
\tau_2=\f{1}{2}
\begin{pmatrix}
0&-1\\
1&0
\end{pmatrix},\q\q
\tau_3=\f{1}{2}
\begin{pmatrix}
-\i&0\\
0&\i
\end{pmatrix},
\ee
which are defined in terms of the Pauli matrices by $\tau_i=-\i\sigma_i/2$, and satisfy the commutation relations $[\tau_i,\tau_j]=\eps_{ij}^{~~k}\tau_k$ as well as the identity
\be
\tau_i\tau_j=\f{1}{2}\eps_{ij}^{~~k}\tau_k-\f{1}{4}\delta_{ij}\openone.
\ee
Internal $\su(2)$ indices are lowered and raised with the flat Euclidean metric $\eta_{ij}=\text{diag}(1,1,1)$. The commutator between Lie algebra elements is given by
\be
\big[\tau^i,X\big]=\big[\tau^i,X^j\tau_j\big]=\eps^i_{~jk}X^j\tau^k.
\ee
The normalization of the representation matrices is such that $-4\tau_i^2=\openone$. Therefore, given a Lie algebra element $X=X^i\tau_i$, we have that
\be
\tr(X\tau^i)=-\f{1}{2}X^i.
\ee
To every group element $g\in\SU(2)$ we associate an $\SO(3)$ rotation matrix $D(g)$ with components $D^{ij}(g)$ given by
\be
D^{ij}(g)=-2\tr\big(g^{-1}\tau^ig\tau^j\big).
\ee
The action of these rotation matrices on Lie algebra elements is given by
\be
D^{ij}(g)X^j=\big(gXg^{-1}\big)^i.
\ee
Furthermore, the rotation matrices satisfy the property
\be\label{epsilonD}
\eps^{imn}D^{ij}(g)=\eps^{jm'n'}D^{m'm}\big(g^{-1}\big)D^{n'n}\big(g^{-1}\big).
\ee

\section{Poisson brackets of integrated fluxes}
\label{appendix:2}

\noindent In this appendix we give the detailed calculations of the various Poisson brackets between integrated fluxes that are considered in section \ref{fluxalgebra}. Although conceptually straightforward, these calculations are lengthy because they require to carefully keep track of all the holonomies involved in the definition of the integrated fluxes. The computations rely on the properties of the rotation matrices given in appendix \ref{appendix:1}.

\subsection{Case 1}

\noindent We consider the Poisson bracket between the fluxes
\be
\bX_{\pi_1}= g^{-1}_{rl^1_1(0)}\bX_{\pi_1\backslash r}g_{rl^1_1(0)},\q g_{rl^1_1(0)}=g_{l(1)l^1_1(0)}h_lg_{rl(0)},
\ee
and
\ba
\bX_{\pi_2}&=&g^{-1}_{rl^2_1(0)}\left(\ldots+g^{-1}_{l^2_1(0)l(0)}X_lg_{l^2_1(0)l(0)}+\ldots\right)g_{rl^2_1(0)}\nn\\
&=&\ldots+g^{-1}_{rl^2_1(0)}g^{-1}_{l^2_1(0)l(0)}X_l g_{l^2_1(0)l(0)}g_{rl^2_1(0)}+\ldots.
\ea
The Poisson bracket is given by
\ba
&&\lb\bX_{\pi_1}^i,\bX_{\pi_2}^j\rb\nn\\
&=&4\lb\tr\left(g_{rl(0)}^{-1}h_l^{-1}g_{l(1)l^1_1(0)}^{-1}\bX_{\pi_1\backslash r}g_{l(1)l^1_1(0)}h_lg_{rl(0)}\tau^i\right),\tr\left(g^{-1}_{rl^2_1(0)}g^{-1}_{l^2_1(0)l(0)}X_l^k\tau^k g_{l^2_1(0)l(0)}g_{rl^2_1(0)}\tau^j\right)\rb\nn\\
&=&-2\tr\left(g_{rl(0)}^{-1}\lb h_l^{-1}g_{l(1)l^1_1(0)}^{-1}\bX_{\pi_1\backslash r}g_{l(1)l^1_1(0)}h_l,X^k_l\rb g_{rl(0)}\tau^i\right)D^{kj}\left(g_{l^2_1(0)l(0)}g_{rl^2_1(0)}\right).
\ea
Now, using the fact that
\ba\label{app1}
\lb h_l^{-1}Zh_l,X_l^k\rb=\eps^{kmn}\big(h_l^{-1}Zh_l\big)^m \tau^n,
\ea
together with formula \eqref{epsilonD},we obtain the result
\ba
\lb\bX_{\pi_1}^i,\bX_{\pi_2}^j\rb
&=&\eps^{kmn}D^{ni}\big(g_{rl(0)}\big)\left(h_l^{-1}g_{l(0)l^1_1(0)}^{-1}\bX_{\pi_1\backslash r}g_{l(0)l^1_1(0)}h_l\right)^mD^{kj}\left(g_{l^2_1(0)l(0)}g_{rl^2_1(0)}\right)\nn\\
&=&\eps^{ikm}D^{kj}\left(g_{rl(0)}^{-1}g_{l^2_1(0)l(0)}g_{rl^2_1(0)}\right)\bX_{\pi_1}^m.
\ea
The case in which $g_{rl_1^1(0)}=g_{l(0)l^1_1(0)}h^{-1}_lg_{rl(1)}$ is computed in a similarly way, but instead of \eqref{app1} one has to use
\ba
\lb h_lZh^{-1}_l,X_l^k\rb=-\eps^{kmn}Z^mh_l\tau^nh_l^{-1}.
\ea

\subsection{Case 2}

\noindent We have to consider the Poisson bracket between the two fluxes defined by
\be
\bX_{\pi_1}=Z_1+g_{rl^1_1(0)}^{-1}g_{l^1_1(0)l(0)}^{-1}X_lg_{l^1_1(0)l(0)}g _{rl^1_1(0)},\q\bX_{\pi_2}=Z_2+g_{rl^2_1(0)}^{-1}g_{l^2_1(0)l(0)}^{-1}X_lg_{l^2_1(0)l(0)}g _{rl^2_1(0)}.\ee
The calculation proceeds straightforwardly using the identities in appendix \ref{appendix:1}
\ba
&&\lb\bX_{\pi_1}^i,\bX_{\pi_2}^j\rb\nn\\
&=&4\lb\tr\left(g_{rl^1_1(0)}^{-1}g_{l^1_1(0)l(0)}^{-1}X_l^k\tau^kg_{l^1_1(0)l(0)}g _{rl^1_1(0)}\tau^i\right),\tr\left(g_{rl^2_1(0)}^{-1}g_{l^2_1(0)l(0)}^{-1}X_l^m\tau^mg_{l^2_1(0)l(0)}g _{rl^2_1(0)}\tau^j\right)\rb\nn\\
&=&\lb X_l^k,X^m_l\rb D^{ki}\left(g_{l^1_1(0)l(0)}g _{rl^1_1(0)}\right)D^{mj}\left(g_{l^2_1(0)l(0)}g _{rl^2_1(0)}\right)\nn\\
&=&\eps^{kmn}X^n_lD^{ki}\left(g_{l^1_1(0)l(0)}g _{rl^1_1(0)}\right)D^{mj}\left(g_{l^2_1(0)l(0)}g _{rl^2_1(0)}\right)\nn\\
&=&\eps^{imn}D^{mj}\left(g_{rl^1_1(0)}^{-1}g_{l^1_1(0)l(0)}^{-1}g_{l^2_1(0)l(0)}g _{rl^2_1(0)}\right)\left(g_{rl^1_1(0)}^{-1}g_{l^1_1(0)l(0)}^{-1}X_lg_{l^1_1(0)l(0)}g _{rl^1_1(0)}\right)^n.
\ea

\subsection{Poisson brackets on the continuum phase space}
\label{PB continuum}

\noindent In this section, we compute in $d=2$ the Poisson bracket between two integrated flux observables whose co-paths cut each other. This computation will be performed at the level of the continuum phase space.

In terms of the continuum electric field $E^b_k(x)$, the integrated fluxes are defined in a manner similar to the simplicial fluxes \eqref{simplicial flux}, namely as
\ba
\bX_{\pi_I}^i&=&-2\int_{\pi_I}\tr\left(g^{-1}_{r\pi_I(s)}\tau^kg_{r\pi_I(s)}\tau^i\right)\eps_{ba}\dot{\pi}^a_I(s)E^{bk}\big(\pi_I(s)\big)\de s\nn\\
&=&\int_{\pi_I}D_I^{ki}(s)\eps_{ba}\dot{\pi}^a_I(s)E^{bk}\big(\pi_I(s)\big)\de s,
\ea
where we have introduced the notation $D_I^{ki}(s)\coloneqq D^{ki}\big(g_{r\pi_I(s)}\big)$. Here, $g_{r\pi_I(s)}$ is the holonomy going from the root to $\pi_I(0)$ and then along the co-path $\pi_I$ to the point $\pi_I(s)$. This holonomy is a function of the connection $A_a^i$. Since the only non-vanishing Poisson bracket is between the connection and the electric field, i.e. (in $d=2$ spatial dimensions)
\ba
\lb A_a^i(x),E^b_k(y)\rb=\delta^{(2)}(x,y)\delta_a^b\delta^i_k,
\ea 
we can write the Poisson bracket between the two integrated fluxes as
\be\label{rech31}
\lb\bX_{\pi_1}^i,\bX_{\pi_2}^j\rb=\int_{\pi_1}\lb D_1^{ki}(s),\bX_{\pi_2}^j\rb\eps_{ba}\dot{\pi}^a_1(s)E^{bk}\big(\pi_1(s)\big)\de s+\int_{\pi_2}\lb\bX_{\pi_1}^i,D_2^{kj}(s)\rb\eps_{ba}\dot{\pi}^a_2(s)E^{bk}\big(\pi_2(s)\big)\de s.
\ee
In order to find $\lb D_1^{ki}(s),\bX_{\pi_2}^j\rb$, we consider the Poisson bracket
\ba\label{rech32}
\lb g^{-1}_{r\pi_1(s)}\tau^kg_{r\pi_1(s)},\bX_{\pi_2}^j\rb&=&-g^{-1}_{r\pi_1(s)}\lb g_{r\pi_1(s)},\bX_{\pi_2}^j\rb g^{-1}_{r\pi_1(s)}\tau^kg_{r\pi_1(s)}+g^{-1}_{r\pi_1(s)}\tau^k\lb g_{r\pi_1(s)},\bX_{\pi_2}^j\rb\nn\\
&=&\left[g^{-1}_{r\pi_1(s)}\tau^kg_{r\pi_1(s)}\ ,\ g^{-1}_{r\pi_1(s)}\lb g_{r\pi_1(s)},\bX_{\pi_2}^j\rb\right].
\ea
We have to find the continuum Poisson bracket between an holonomy and an integrated flux. To this end, we write the holonomy as
\ba
g_{r\pi_1(s')}&=&g_{\pi_1(1/2+\epsilon)\pi_1(s')}g_{\pi_1(1/2-\epsilon)\pi_1(1/2+\epsilon)}g_{r\pi_1(1/2-\epsilon)}\nn\\
&=&g_{\pi_1(1/2+\epsilon)\pi_1(s')}\left[\openone-\int_{1/2-\epsilon}^{1/2+\epsilon}\tau^i\dot{\pi}_1^c(s)A^i_c\big(\pi_1(s)\big)\de s+O\big(\epsilon^2\big)\right]g_{r\pi_1(1/2-\epsilon)}.
\ea
Here, we assume that the intersection point between the co-path $\pi_2$ and the curve underlying the holonomy is located at $s=1/2$, and that $s'>1/2+\epsilon$. Now, introducing the function $\theta_{1/2}(s')$ which is such that $\theta_{1/2}(s')=1$ for $s'>1/2$, that $\theta_{1/2}(1/2)=1/2$, and which vanishing otherwise, we get that
\ba
\lb g_{r\pi_1(s')},\bX^j_{\pi_2}\rb&=&-\theta_{1/2}(s')g_{\pi_1(1/2+\epsilon)\pi_1(s')}\tau^ig_{r\pi_1(1/2-\epsilon)}\times\nn\\
&&\int_{\pi_1}\int_{\pi_2}\lb\dot{\pi}_1^c(s_1)A^i_c\big(\pi_1(s_1)\big),D^{kj}_2(s_2)\eps_{ba}\dot{\pi}^a_2(s_2)E^{b}_k\big(\pi_2(s_2)\big)\rb\de s_1\de s_2\nn\\
&=&-\theta_{1/2}(s')g_{\pi_1(1/2+\epsilon)\pi_1(s')}\tau^kg_{r\pi_1(1/2-\epsilon)}\times\nn\\
&&\int_{\pi_1}\int_{\pi_2}\delta^{(2)}\big(\pi_1(s_1),\pi_2(s_2)\big)\left(\eps_{ba}\dot{\pi}^b_1(s_1)\dot{\pi}^a_2(s_2)\right)D^{kj}_2(s_2)\de s_1\de s_2\nn\\
&\stackrel{\epsilon\rightarrow0}{=}&-\theta_{1/2}(s')\eps_{12}g_{\pi_1(1/2)\pi_1(s')}\tau^kg_{r\pi_1(1/2)}D_2^{kj}(1/2).
\ea
Here we have used the notation $\eps_{12}\coloneqq\text{sign}\left(\eps_{ba}\dot{\pi}^b_1(1/2)\dot{\pi}^a_2(1/2)\right)$. Using this result in \eqref{rech32}, we find that
\be
\lb D_1^{ki}(s),\bX_{\pi_2}^j\rb=-\theta_{1/2}(s)\eps_{12}\eps^{lmi}D^{nj}_2(1/2)D^{kl}_1(s)D^{nm}_1(1/2).
\ee
This enables us to conclude that the first term on the right-hand side of \eqref{rech31} is given by
\ba
\int_{\pi_1}\lb D_1^{ki}(s),\bX_{\pi_2}^j\rb\eps_{ba}\dot{\pi}^a_1(s)E^{bk}\big(\pi_1(s)\big)\de s&=&-\eps_{12}\eps^{lmi}D^{nj}_2(1/2)D^{nm}_1(1/2)\bX^l_{\pi_1[1/2]}\nn\\
&=&\eps_{12}\eps^{iml}D^{mj}\left(g^{-1}_{r\pi_1(1/2)}g_{r\pi_2(1/2)}\right)\bX^l_{\pi_1[1/2]},
\ea
where
\be
\bX^i_{\pi_I[1/2]}=\int_{\pi_I}\theta_{1/2}(s)D_I^{ki}(s)\eps_{ba}\dot{\pi}^a_I(s)E^{bk}\big(\pi_I(s)\big)\de s
\ee
is the part of the integrated flux observable associated to the part of the co-path lying after the crossing.

The second term on the right-hand side of \eqref{rech31} can be computed in the same way, and putting everything together we find that
\be
\lb\bX_{\pi_1}^i,\bX_{\pi_2}^j\rb=\eps_{12}\eps^{ikm}D^{kj}\left(g^{-1}_{r\pi_1(1/2)}g_{r\pi_2(1/2)}\right)\left(\bX^m_{\pi_1[1/2]}-D^{mn}\left(g^{-1}_{r\pi_1(1/2)}g_{r\pi_2(1/2)}\right)\bX^n_{\pi_2[1/2]}\right).\nn\\
\ee
This coincides with the result obtained in \eqref{case3PBfinal} with the computation at the level of the discrete phase space (there, we have that $\eps_{12}=-1$).

\end{document}